\documentclass[12pt]{article}
\usepackage[margin=1in]{geometry} 
\usepackage{amssymb} 
\usepackage{amsmath} 
\usepackage{amsthm} 
\usepackage{graphicx} 
\usepackage{booktabs} 
\usepackage{natbib} 
\usepackage{pgfplots} 
\usepgfplotslibrary{fillbetween} 
\usepackage{bm} 
\usepackage{xcolor} 
\definecolor{royalblue}{HTML}{0000CD} 
\usepackage[colorlinks=true]{hyperref} 
\hypersetup{
    colorlinks=true,
    allcolors=royalblue,
} 
\usepackage{threeparttable} 
\usepackage{thmtools} 
\newtheorem{proposition}{\normalfont\scshape Proposition}[]
\newtheorem{Lemma}{\bf Lemma}{}
\allowdisplaybreaks 
\usepackage{mathtools} 

\usepackage{caption} 
\usepackage{subcaption} 

\usepackage{setspace} 

\usepackage{array,multirow}
\newcommand\Tstrut{\rule{0pt}{2.6ex}}         

\usepackage{adjustbox} 
\usepackage{enumerate} 

\usepackage{algpseudocode} 

\usepackage[multiple]{footmisc} 

\begin{document}

\title{On the Convergence of Credit Risk in Current Consumer
Automobile Loans\protect\footnote{We thank conference participants at the
2023 Joint Statistical Meetings (Toronto),
	Fifth Biennial Auto Lending Conference (Federal Reserve Bank of Philadelphia's
	SURF, CFI), the 2023 New England Statistics Symposium (Boston University),
	the 2023 Boulder Summer Conference on Consumer Financial Decision Making
	(University of Colorado Boulder Leeds School of Business) and
seminar participants at Bentley University, the University of Connecticut
(School of Business, Department of Finance), the University of Connecticut
(Department of Economics), Universit\'{e} Concordia University (John Molson
School of Business, Department of Finance),
and Joseph Golec, Brian Melzer,
Jonathan A. Parker, and Jonathan Zinman.  These acknowledgments should not be
mistaken for endorsements; any remaining errors are the sole responsibility of
the authors.  Jackson P. Lautier was employed with Prudential Financial,
Inc.\ (2010-2019) but believes in good faith there are no conflicts of
interest. The other authors have no potential conflicts to disclose.
Jackson P. Lautier's work was supported by a National Science Foundation
Graduate Research Fellowship under Grant No. DHE 1747453.}}

\author{
  Jackson P. Lautier\footnote{Department of Mathematical Sciences,
  Bentley University}
  \thanks{Corresponding to Jackson P. Lautier, Bentley University,
	Room 321, Morison Hall, 175 Forest Street, Waltham, MA 02452; e-mail:
	\href{mailto:jlautier@bentley.edu}{jlautier@bentley.edu}.}
  \and
  Vladimir Pozdnyakov\footnote{Department of Statistics, University of
  Connecticut}
  \and
  Jun Yan\footnotemark[4]
}


%
%

\maketitle

\begin{abstract}
Loan seasoning and inefficient consumer interest rate refinance behavior are
well-known for mortgages.  Consumer automobile loans, which are collateralized
loans on a rapidly depreciating asset, have attracted less attention, however.
We derive a novel large-sample statistical hypothesis test suitable for loans
sampled from asset-backed securities to populate a transition matrix between
risk bands.  We find all current risk bands eventually converge to a
super-prime credit, despite remaining underwater. Economically, our results
imply borrowers forwent \$1,153-\$2,327 in potential credit-based savings
through delayed prepayment. We present an expected present value analysis to
derive lender risk-adjusted profitability.  Our results appear robust to
COVID-19.\newline
\end{abstract}

\clearpage


\doublespacing

\section{Introduction}

In chronicling cumulative loss curves
for securitization pools of individual consumer automobile
loans, there is a familiar pattern every junior credit analyst can sketch
from memory: an initial rise in the early months of the securitization followed
by a sustained flattening in the curve once the pool eventually settles into
its long-term steady state.  In higher risk or \textit{subprime} pools of
borrowers, the eventual cumulative loss percentage might be many multiples
higher than
lower risk or \textit{prime} pools of borrowers, but the overall shape follows
the familiar natural $\log$-like pattern.\footnote{
Junior credit analysts are trained to look for any sudden upward deviations in
the historical pattern, or {\it peel back},
which may indicate a rapid deterioration in the performance of the loans. 
}
We illustrate three such securitization loss curves in
Figure~\ref{fig:loss_curves}.
It is peculiar that the loss curves all eventually flatten to a similar degree.
This suggests an eventual equivalence in the instantaneous default rate
conditional on survival, despite the notable cumulative differences between
the loss curves.

This is the concept of {\it loan seasoning}, which is well-documented for
residential mortgages \citep[e.g.,][]{adelino_2019}.  Because consumer
auto loans are collateralized loans, it is natural to suspect they would
behave like mortgages.
Unlike residential homes, however, the automobile is a rapidly depreciating
asset \citep{storchmann_2004}.  This curiosity warrants additional study: how
does conditional credit risk behave for collateralized consumer
loans with rapidly
depreciating collateral values, such as consumer automobile loans?
Such a question has importance, given total issuance north of \$200 billion
in consumer auto asset-backed securities (ABS) \citep{sifma_2022}
and U.S.\ consumer automobile
debt totaling over \$1,400 billion \citep{fed_stats}.

To study the conditional credit risk of a current consumer automobile loan,
we derive a new estimation technique from the asymptotic properties of
large sample statistics.  Our novel approach allows for estimates of the
exact moment a loan's future risk profile changes.  To our knowledge, such
estimates have not yet been presented for consumer auto loans. Further,
our
methods are appropriately calibrated for discrete loan data sampled from ABS.
Hence, the techniques presented herein may have more widespread future
applications given \citet{reg_ab2}, an abundant public source of ABS
data.\footnote{
Indeed, an ancillary intention of this work is to promulgate the utility of
\citet{reg_ab2} among academic researchers.
}

Because we are able to estimate a complete transition matrix between consumer
risk bands using our novel methods, we may further assess the economic
implications to borrowers vis-\`{a}-vis a changing risk profile. We find that
current borrowers who transition into superior risk bands are slow to seek out
a {\it credit-based} refinance, all else equal.\footnote{
Traditionally, a borrower's credit risk profile is assessed upon the initial
loan application,
and the borrowing cost comes through in the form of the annual percentage rate
(APR) on the loan contract. This is the common practice of
{\it risk-based pricing} \citep[e.g.,][]{edelberg_2006, phillips_2013}.
Because the APR is set at the onset of the loan contract, it is possible that
consumers who remain current at the contract's origination APR may eventually
overpay in comparison to a credit-base market rate APR that reflects an updated
risk assessment, all else equal.
}
This differs from the
well-documented observation that borrowers refinance inefficiently in
mortgages given changes to {\it interest rates}
\citep[e.g.,][]{keys_2016, agarwal_2017, andersen_2020}.  Hence, our findings
reflect further potential consumer behavioral inefficiencies, potentially
exacerbating any lost savings due to interest rate-based refinance
inefficacy.
Of note, our conditional credit risk results present despite the rapidly
depreciating collateral value of used autos.  This observation runs counter
to traditional loan-to-value (LTV) default behavior expectations
\citep[e.g.,][]{campbell_2015}.

\begin{figure}[t!]
    \centering
    \includegraphics[width=\textwidth]{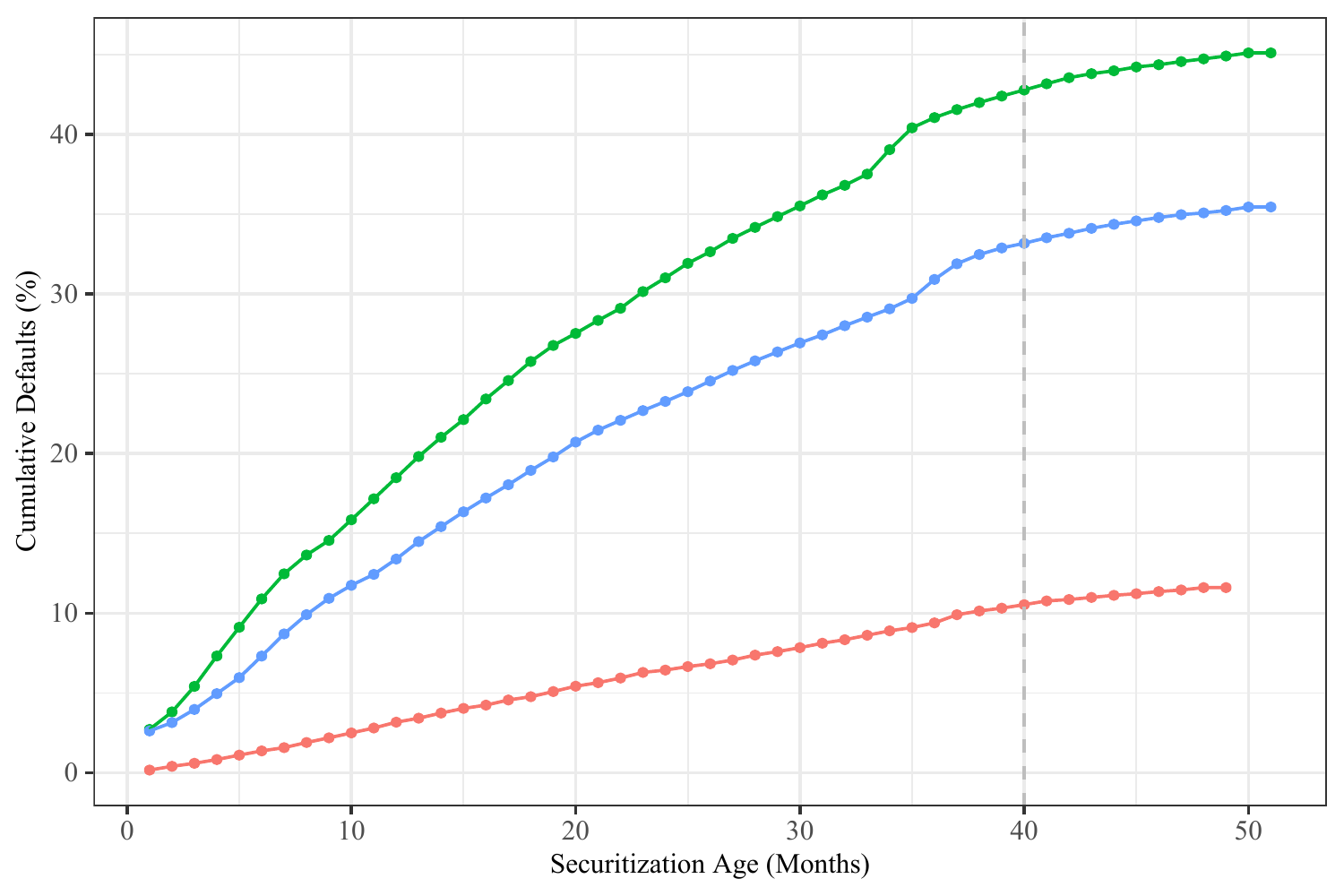}
    \caption{
\footnotesize{    
\textbf{Classical Consumer Automobile Securitization Loss Curves.}\newline
A plot of three securitization loss curves: the cumulative count (\%) of
defaults against securitization age (months).
The higher two loss curves correspond to riskier (i.e., {\it subprime}) pools
of loans in terms of traditional credit metrics, and the lower curve juxtaposes
these subprime pools to a {\it prime} risk pool.
It is striking that all curves eventually flatten (e.g., 
approximately after the dashed vertical line),
despite the large differences in underlying borrower credit
quality and rapid deterioration of collateral value.
}
}
    \label{fig:loss_curves}
\end{figure}

We thus present a paper with two intertwined contributions.  The first is
a new financial econometric technique to statistically demonstrate when risk
converges between risk bands for current loans.  The second is a
subsequent analysis on consumer automobile loans using our novel
financial econometric tools to indicate apparent market inefficiencies
vis-\`{a}-vis credit-based refinance, assess COVID-19's impact by risk band,
estimate lender expected risk-adjusted profitability, and study prepayment
behavior by risk band.  We are thus within the space of consumer
finance, and it all derives from novel methods to empirically validate a
current borrower's declining credit risk over time.

To expound, we first derive
a financial econometric hypothesis testing technique to
estimate the exact loan age at which the instantaneous credit risk conditional
on survival (i.e., conditional on a loan being current) between two different
risk bands becomes equivalent (i.e., converges).\footnote{
It should be stated that two borrowers from different risk bands that
completely repay their loans will trivially converge in credit risk to a
zero probability of default.
The major analysis, therefore, is devising a method to formally demonstrate
if and when such convergence occurs prior to loan termination (and, if so,
its financial implications).
}
Our methods rely on large sample
asymptotic statistics and a survival (or time-to-event) analysis tool known as
the hazard rate.  The hazard rate measures the probability of an event
conditional on survival, and it is thus the ideal quantity of interest for a
current loan analysis.  Due to some incomplete data challenges in working with
loan data sampled from securitization pools, we take some effort to arrive at
an estimator that is theoretical suitable for our application.  In particular,
we utilize a competing risks framework to estimate a cause-specific hazard
(CSH) rate, which allows us to model both conditional default and conditional
prepayment probabilities.\footnote{
All of our analysis adjusts for prepayment
behavior. In other words, while we consider both prepayments and repayments
as ``non-defaults", the distribution of loan lifetimes, and therefore all
subsequent analysis, is adjusted for the timing of observed prepayments.
}

The estimator we utilize also has convenient large sample properties, which we
state formally in Appendix~\ref{subsec:csh_props} and prove in the
Online Appendix~\ref{subsec:proof2}
(a contribution to the statistical literature in
its own right).  Specifically, it yields a large sample hypothesis test
to determine if the CSH rate for default between two different risk
bands is different at a statistically significant level.  If not, then we
cannot claim the CSH rates are different between the two risk bands.  In other
words, a failure to reject the null hypothesis suggests conditional default
risk between different risk bands has converged.\footnote{
It may be of assistance to note our process is the inverse of a usual
statistical analysis, in which researchers actively search for statistical
significance.  That is, our failure to reject the null is the indicator for a
potential convergence rather than a rejection of the null.
}
An advantage of our approach is that it is completely data-driven. In
other words, the data alone informs our distributional estimates, which
may then serve as a benchmark from which future economic research may be
calibrated.
A further advantage is that we may utilize data sampled from ABS pools, which
opens a rich data source \citep{reg_ab2} beyond traditional direct
consumer loan data (see Section~\ref{sec:data} for data details).

Using our techniques, we find evidence of convergence beginning between
disparate risk bands of 72-73 month auto loans after just one year. Even for
risk bands with large differences in their initial credit risk assessment,
converge in conditional credit risk occurs well before scheduled termination
(e.g., subprime loans behave like super-prime credits after 42-48 months of
payments).  For complete risk band transition matrix details, see
Section~\ref{subsec:emp_res} and Table~\ref{tab:risk_conv_mat}.  We find these
results are robust to sensitivity analysis considering the economic impact of
COVID-19 (see Section~\ref{subsec:COVID}; itself a topic of stand alone
interest), collateral type, and the business
model of the loan originator's parent company (for the latter two, see
Section~\ref{subsec:new_cars}).  Notably, because collateralized loans on
used autos have rapidly depreciating collateral values, these results cannot
be explained by traditional LTV default behavior expectations (see
Figure~\ref{fig:ltv_by_age}).
The entirety of the methodological treatment
and subsequent data analysis to estimate the point two different risk
bands converge in a go-forward assessment of credit risk, i.e.,
the {\it credit risk convergence} analysis,
may be found in Section~\ref{sec:CRC}.

The second intended contribution is a multi-part applied study of the financial
implications of our credit risk convergence and novel methodological results.
In Section~\ref{subsec:lend_prof}, we apply the CSH rates for default we
estimate in Section~\ref{subsec:emp_res} within an actuarial analysis to solve
for an expected risk-adjusted rate of lender profitability conditional on loan
survival.  We find that lender profits are
back-loaded, which is consistent with the insurance-like pricing for
pools of risky loans.  Because of the competing risks methodology, the CSH
rates for default are adjusted for conditional prepayments.
We consider the consumer perspective in Section~\ref{subsec:cons_perp}.
To first estimate the potential savings available to consumers,
we assume the average borrower in one
risk band refinanced at the average rate in a superior risk band, once
eligible based on our credit risk convergence point estimates
(ceteris paribus).  Section~\ref{subsec:cons_perp} then briefly assess
potential motivations of observed borrower behavior.

We find that the riskiest borrowers (deep subprime, subprime) can potentially
save between \$11-61 dollars in monthly payments
or \$193-1{,}616 in total by refinancing.  Our estimates suggest deep subprime
and prime borrowers should refinance after about 42-50 months, when they
become prime borrowers.  We find evidence that these
borrowers generally wait too long to refinance.  In a surprise, we find that
less risky loans (near-prime, prime) leave even more money on the table, with
total savings ranging from \$160-2{,}327 (or \$13-56 in monthly payments).
Our estimates suggest that near-prime and prime borrowers should refinance
quickly, after about only one year, but they also generally wait too long.
Hence, in a
result counter to expectations about borrower sophistication, it is the
near-prime and prime loans that behave less efficiently.\footnote{
One benefit of greater affluence is the mental freedom that 
accompanies an ability to overpay with limited consequences.  We thank
Susan Woodward for this observation.
}
These savings are attributable to a potential credit-based refinance, which
differs from the traditional interest rate-based refinance analysis. Given
that consumer automobile loans remain deep underwater close to termination
(see Figure~\ref{fig:ltv_by_age}), such conditional credit performance is
not just an artifact of becoming {\it in-the-money}
\citep[e.g.,][]{deng_1996}.
These results suggest potential economic interpretations.  As such,
Section~\ref{subsec:cons_perp} closes by opining market frictions may
exist that prevent both borrowers and lenders alike from reducing these
suspected consumer auto refinance market inefficiencies.
In hopes of encouraging related research, we also
proffer that lenders may consider offering
new loan products that reward borrowers for good performance or potential
regulator interventions.

Our study aligns with previous studies in the consumer automobile
lending space. \citet{heitfield_2004} present an initial competing risks study
of borrower behavior within subprime auto loans.
We find similar prepayment behavior acceleration with
loan age, and we also find subprime borrowers are sensitive to aggregate
shocks.  We closely consider conditional default risk based on survival,
however, which provides an alternative perspective to the analysis of
\citet{heitfield_2004}.
\citet{agarwal_2008} use a competing risks model to consider automobile choice
and ultimate borrower loan performance.  They find borrowers that select
luxury automobiles have a higher probability of prepayments, and loans on most
economy automobiles have a lower probability of default.  We again focus our
analysis on default rates conditional on survival, however.  Our statistical
methods are also more precisely attuned to discrete time.  Further, our
methods demonstrate how to utilize ABS data from Reg AB II
\citep{reg_ab2}, whereas
\citet{agarwal_2008} uses conventional direct loan-level data.  See also
\citet{agarwal_2007}, which is a close relative to \citet{agarwal_2008}.

More broadly within consumer automobile research,
there is evidence that consumers are subject to
various forms of troubling economic behavior.  For example, racial
discrimination has been found in studies that span decades
\citep[e.g.,][]{ayres_1995, edelberg_2007, butler_2022}.  For an overview of
the used car industry and the challenges presented to poor consumers in
purchasing and keeping transportation, see \citet{karger_2003}.
\citet{adams_2009} look at the effect of borrower liquidity on short-term
purchase behavior within the subprime auto market.  Namely, they observe
sharp increases in demand during tax rebate season and high sensitivity to 
minimum down payment requirements.  \citet{grunewald_2020} find that
arrangements between auto dealers and lenders lead to incentives that increase
loan prices.  They also find consumers are less responsive to finance charges
than vehicle charges and that consumers benefit when dealers do not have 
discretion to price loans.  While consumer auto loans and subprime borrowers 
have attracted the attention of previous researchers,
we again do not find consideration
of the borrower risk profile over the lifespan of the loan.
Within this backdrop, therefore, the results of
Section~\ref{subsec:cons_perp} find an additional challenge to consumers
with auto loans in that such borrowers struggle to recoup additional
potential savings available from a credit-based refinance.

There are related results within mortgages.
\citet{deng_2000} employ a competing risks model within the mortgage
market.  They find similar conditional default behavior.
Within the mortgage market, see also
\citet{calhoun_2002}.  \citet{ambrose_2003} apply a competing risks model to
commercial loans underlying commercial mortgage-backed securities.  
Recent literature
reviews provide a more complete examination of default behavior as a function
of mortgage age \citep[e.g.,][]{jones_2019}.
None of these studies are specific to consumer study automobile loans,
however, which are conversely loans subject to rapidly depreciating
collateral values.

\section{Data}
\label{sec:data}

On September 24, 2014, the Securities and Exchange Commission (SEC) adopted 
significant revisions to Regulation AB and other rules governing the offering,
disclosure, and reporting for ABS \citep{reg_ab2}.
One component of these large scale revisions, which took effect November 23, 
2016, has required public issuers of ABS to make freely available pertinent 
loan-level information and payment performance on a monthly basis 
\citep{cfr_229}.  We have utilized the
Electronic Data Gathering, Analysis, and Retrieval (EDGAR) system operated by
the SEC
to compile complete loan-level performance data for the consumer 
automobile loan ABS bonds CarMax Auto Owner Trust 2017-2 \citep{cmax_2017}
(CARMX), Ally Auto Receivables Trust 2017-3 \citep{aart_2017} (AART),
Santander Drive Auto Receivables Trust 2017-2 \citep{sdart_2017} (SDART),
and Drive Auto Receivables Trust 2017-1 \citep{drive_2017} (DRIVE).
By count, the total number of
loans for CARMX, AART, SDART, and DRIVE were 55,000, 67,797, 80,636, and
72,515, respectively.

The bonds were selected because of the credit profile of the underlying loans,
the lack of a direct connection to a specific auto manufacturer, and the
observation window of each bond's performance spanning approximately the same
macroeconomic environment. We elaborate on each point in turn.
The credit profile of a DRIVE borrower
is generally deep subprime to subprime, SDART is subprime to near-prime, CARMX
is near-prime to prime, and AART is prime to super-prime.\footnote{
The standard definitions of the terms deep subprime to super-prime stem from
the borrower's credit score.  Specifically, credit scores below 580 are
considered ``deep subprime", credit scores between 580-619 are ``subprime",
620-659 is ``near-prime", 660-719 is ``prime", and credit scores of 720 and
above are ``super-prime" \citep{cfpb_2019}.}  Thus, the collection of all four
bonds taken together span the full credit spectrum of individual borrowers.
Figure~\ref{fig:2017_summary} in Section~\ref{subsec:summary}
provides additional details.  Next, it is
common that an auto manufacturer will originate loans using its financial
subsidiary (e.g., Ford Credit Auto Owner Trust).  The bonds selected do not
have a direct connection to a specific auto manufacturer, however, and so we
may allay concerns our convergence point estimates
may be influenced by oversampling loans secured
by a specific brand of automobile.\footnote{
We acknowledge the business objectives of CarMax, a used auto sales company,
differ from the traditional banks of Santander and Ally.
We sensitivity test this point in the
robustness checks of Section~\ref{subsec:new_cars}.
}
Lastly, the bonds were selected
to span approximately the same months to ensure all underlying loans were 
subject to the same macroeconomic environment.  Specifically, CARMX, AART,
SDART, and DRIVE began actively paying in March, April, May, and April of 2017, respectively, and each trust was active for 50, 44, 52, and 52 months,
respectively.


\subsection{Loan Selection and Defining Risk Bands}
\label{subsec:loan_filter}

To ensure the underlying loans in our analysis are as comparable as possible,
we employ a number of filtering mechanisms.  First, we remove any loan 
contracts that include a co-borrower.  Second, we require each loan to have
been underwritten to the level of ``stated not verified"
(\texttt{obligorIncomeVerificationLevelCode}), which is a prescribed
description of the amount of verification done to a borrower's stated income 
level on an initial loan application \citep{cfr_229}.
Third, we remove all loans originated 
with any form of subvention (i.e., additional financial incentives, such as 
added trade-in compensation or price reductions on the final sale price).  
We then require all loans to correspond to the sale of a used
vehicle.\footnote{
This was mainly to keep the loans from CARMX, of which used
cars predominate.  We sensitivity test this requirement in the robustness
checks of Section~\ref{subsec:new_cars}.
}
We further drop any loan with a current status of ``repossessed" as of the 
first available reporting month of the corresponding ABS.  Further, to minimize
the chance of inadvertently including a loan that has been previously 
refinanced or modified, we only consider loans younger than 18 months as of
the first available ABS reporting month.  For loan
term, we only include loans with an original term of 72 or 73
months.\footnote{Pragmatically, the most common loan term in the data was 72/73
months, and so our loan term choice allows us to maximize the sample size.}

As a final data integrity check, we remove any loans that did not pay enough
total principle to pay-off the outstanding balance as of the first month the
trust was active and paying but had a missing value (\texttt{NA}) for the
outstanding balance in the final month the trust was active and paying.  In
other words, the loan outcome was not clear from the data; the loan did not
pay enough principal to pay off the outstanding balance nor default but
stopped reporting monthly payment data.  In total, this final data integrity
check impacts only 2{,}630 or 4.3\% of the filtered loan population.   
We are left with 58{,}118 individual consumer auto loan contracts in total,
summary details of which may be found in Section~\ref{subsec:summary}. Complete
replication code and other data details may be found in the online
supplementary material.

Next, we assign each loan into a credit risk category or \textit{risk band} 
depending on the original interest rate
(\texttt{originalInterestRatePercentage}) assigned to the contracted loan. The
interest rate is the ideal measure of perceived borrower risk within a
risk-based pricing framework \citep{edelberg_2006, phillips_2013} because a
borrower's risk profile is a multidimensional function of factors like
credit score, loan amount, down payment percentage (\% down),
vehicle or collateral value, income,
payment-to-income (PTI), etc., in addition to many of the factors of which we have
already filtered.  In other words, given we have already controlled for
prevailing market rates by selecting loans originated within a close
temporal proximity, the interest rate serves as the market's best estimate of a
loan's risk profile.

We now formalize this discussion slightly.  Working from \citet{phillips_2013},
a borrower's interest rate in risk band $a$, $r_a$, is
\begin{equation*}
r_a = r_c + m + l_a,
\end{equation*}
where $r_c$ is the cost of capital, $m$ is the added profit margin, and $l_a$
is a factor that varies by risk band.  The components $r_c$ and $m$ will be
shared by all risk bands, and so there exists some functional relationship
\begin{equation*}
l_a \equiv f(\text{PTI}, \text{\% down}, \text{Loan Amt}, \text{Vehicle Val},
\ldots).
\end{equation*}
Rather than attempt to recover this unknown $f$, therefore, we are in effect
treating the lender's credit scoring model as an accurate reflection of the
borrower's risk.\footnote{Indeed, these models are often quite
sophisticated \citep{einav_2012}.}
Specifically, we assign borrower's with an APR of 0-5\% to the
super-prime risk band, 5-10\% to the prime risk band, 10-15\% to the near-prime
risk band, 15-20\% to the subprime risk band, and 20\%+ to the deep subprime
risk band.  In a review of Figure~\ref{fig:2017_summary} in
Section~\ref{subsec:summary}, we can see that the risk bands assigned by
interest rate compare favorably to the traditional credit score borrower risk
band definition \citep{cfpb_2019}.

\subsection{Summary of Selected Loans}
\label{subsec:summary}

After the data cleaning and filtering of Section~\ref{subsec:loan_filter},
we have payment performance for 58{,}118 consumer auto 
loans that span a wide range of borrower credit quality based on the 
traditional credit score metric.  Figure~\ref{fig:2017_summary} presents a
summary of each bond by obligor credit score
and interest rate as of loan origination.  Judging by credit score,
we can see that generally DRIVE is a
deep subprime to subprime pool of borrowers, SDART is a subprime to near-prime
pool, CARMX is a near-prime to prime pool, and AART is a prime to super-prime
pool of borrowers \citep{cfpb_2019}.  As expected, in a risk-based pricing
framework, the density plot of each borrower's interest rate has an inverse
relationship to the density plot of each borrower's credit score: lower credit
scores correspond to higher interest rates (compare the first two rows of
Figure~\ref{fig:2017_summary}).  As such, we can see the annual percentage
rates (APRs) are higher for the DRIVE and SDART bonds, generally sitting
within a range around 20\% and then declining to under 15\% for CARMX and 
finally under 10\% for AART.  The bottom two rows of
Figure~\ref{fig:2017_summary} demonstrate that defining risk bands by interest
rate corresponds closely to the traditional credit score risk band
definitions \citep{cfpb_2019}, as the expected inverse relationship holds.

\begin{figure}[t!]
    \centering
    \includegraphics[width=\textwidth]{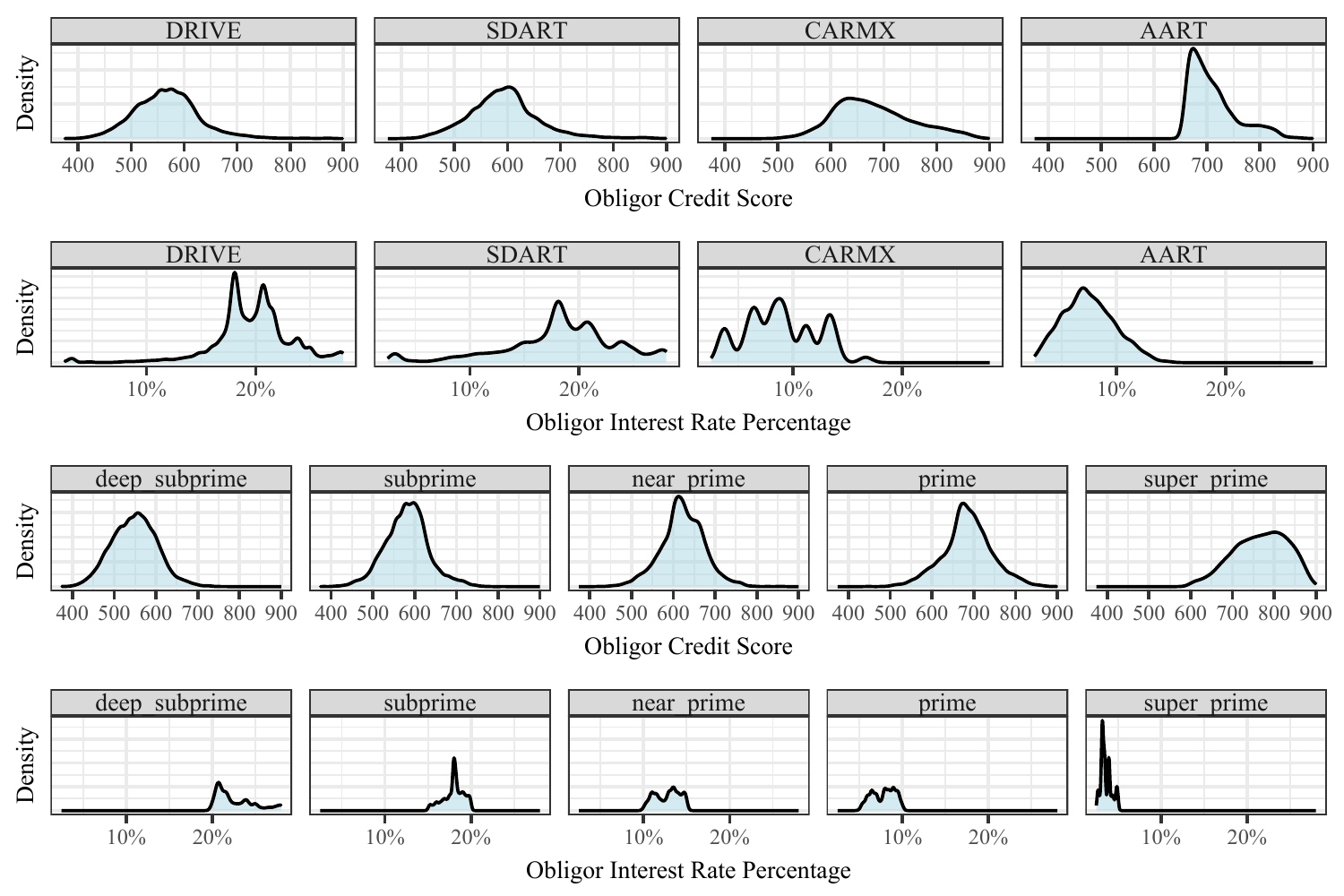}
    \caption{
\footnotesize{     
    \textbf{Borrower Credit Profile and APR by Bond, Risk Band.}\newline
    Borrower credit profiles (1st row) and charged APR
    (2nd row) of the 58{,}118 filtered
    consumer automobile loans used in the analysis
    of Sections~\ref{sec:CRC} and~\ref{sec:finc_imp}
    by ABS bonds CarMax Auto Owner Trust
    2017-2 \citep{cmax_2017} (CARMX, 6{,}835), Ally Auto Receivables Trust
    2017-3 \citep{aart_2017} (AART, 2{,}171), Santander Drive Auto Receivables
    Trust 2017-2 \citep{sdart_2017} (SDART, 20{,}192), and Drive Auto
    Receivables Trust 2017-1 \citep{drive_2017} (DRIVE, 28{,}920).
    Distribution of credit scores (3rd row) and interest rates (4th row),
    by APR-based risk band classification:
    super-prime (0-5\%), prime (5-10\%), near-prime (10-15\%),
    subprime (15-20\%), and deep subprime (20\%+)
    for the same set of 58{,}118 loans.
}
}
    \label{fig:2017_summary}
\end{figure}

The loans are well dispersed geographically among all 50 states and Washington,
D.C., with the top five concentrations of Texas (13\%), Florida (12\%), 
California (9\%), Georgia (7\%), and North Carolina (4\%).  Similarly, the 
loans are well diversified among auto manufacturers, with the top five
concentrations of Nissan (13\%), Chevrolet (10\%), Ford (7\%), Toyota (7\%),
and Hyundai (7\%).  Thus, our sample is not overly
representative to one state-level economic locale or auto manufacturer.
For additional details on the makeup of the loans, see the
associated prospectuses
\citep{aart_2017,cmax_2017,drive_2017,sdart_2017}.

Table \ref{tab:bond_summary} provides a summary of borrower counts by bond and
performance.  The total pool of 58{,}118 loans is weighted towards deep 
subprime and subprime borrowers, which are each 37\% of the total and together
74\%.
Similarly, DRIVE and SDART supply around 85\% of the total loans
in our sample.  The smallest risk band is super-prime, which totals 2{,}179
loans for 4\% of the total of 58{,}118.\footnote{
Our asymptotic results scale by sample size, so the confidence
interval width adjusts appropriately.
}

In terms of loan performance, we can observe some clear trends in Table
\ref{tab:bond_summary}.  First, more than half of all deep subprime risk band
loans defaulted,\footnote{
We use a strict definition of default in that three consecutive missed payments
is a default.  This was defined within our code (see 
Online Appendix~\ref{subsec:loan_algo}) to ensure a consistent default definition
between servicers.
}
and this percentage declines by risk band until super-prime,
in which only 4\% of loans defaulted during the observation window. We also see
that performance is fairly consistent by risk band, even among different bonds.
For example, super-prime default percentages are within a tight range (3-6\%)
across each bond.  The same may be said for deep subprime defaults.  We see
some wider ranges in the default percentages of the subprime (33-40\%), prime
(8-19\%), and near-prime (17-24\%) risk bands by bond,
but they remain close enough to suggest there is not a worrisome
difference between the credit scoring models employed by each different issuer.
Overall, the percentage of defaulted loans declines as the credit quality of
the risk band increases.  This is further evidence that our APR-based risk band
definition has yielded appropriate classification results.

\begin{table}[t!]
\caption{
\footnotesize{
\textbf{Borrower Counts by Risk Band, Bond, and Loan Outcome.}\newline
This table reports the summary statistics and loan outcomes of the 58{,}118 filtered
consumer automobile loans summarized in Figure~\ref{fig:2017_summary}. Table specific
abbreviations are DRIVE (DRV), SDART (SDT), CARMX (CMX) and AART (AAT).
Percentages may not total to 100\% due to rounding.
}
}
{
\begin{adjustbox}{max width=\textwidth}
\begin{tabular}{cccccccc}
        && deep subprime & subprime & near-prime & prime & super-prime & Total\\
        \toprule
        &Total & 21,630 (37\%) & 21,332 (37\%) & 6,677 (11\%) 
        & 6,300 (11\%) & 2,179 (4\%) & 58,118 (100\%)\Tstrut\\
        \midrule
        &DRIVE & 14,079 (65\%) & 12,884 (60\%) & 1,443 (22\%) & 220 (3\%) 
        & 294 (13\%) & 28,920 (50\%)\\
        &SDART & 7,551 (35\%) & 8,327 (39\%) & 2,782 (42\%) & 861 (14\%) 
        & 671 (31\%) & 20,192 (35\%)\\
        &CARMX & 0 (0\%) & 120 (1\%) & 2,128 (32\%) & 3,752 (60\%) 
        & 835 (38\%) & 6,835 (12\%)\\
        &AART & 0 (0\%) & 1 (0\%) & 324 (5\%) & 1,467 (23\%) & 379 (17\%) 
        & 2,171 (4\%)\\
        \hline
        &Total & 21,630 (100\%) & 21,332 (100\%) & 6,677 (100\%) 
        & 6,300 (100\%) & 2,179 (100\%) & 58,118 (100\%)\Tstrut\\
        \midrule
        &Defaulted & 11,210 (52\%) & 7,900 (37\%) & 1,422 (21\%) & 
        624 (10\%) & 92 (4\%) & 21,248 (37\%)\\
        &Censored & 3,547 (16\%)  & 4,599 (22\%) & 1,997 (30\%) & 
        2,556 (41\%) & 948 (44\%) & 13,647 (23\%)\\
        &Repaid & 6,873 (32\%) & 8,833 (41\%) & 3,258 (49\%) & 
        3,120 (50\%) & 1,139 (52\%) & 23,223 (40\%)\\
        \hline
        &Total & 21,630 (100\%) & 21,332 (100\%) & 6,677 (100\%) 
        & 6,300 (100\%) & 2,179 (100\%) & 58,118 (100\%)\Tstrut\\
        \midrule
        \multirow{3}{*}{\rotatebox[origin=c]{90}{\tiny{DRV}}} 
        & Defaulted & 7,518 (53\%) & 5,115 (40\%) & 351 (24\%) & 42 (19\%) 
        & 14 (5\%) & 13,040 (45\%)\\
        & Censored & 2,214 (16\%) & 2,641 (20\%) & 324 (22\%) & 60 (27\%)
        & 119 (40\%) & 5,358 (19\%)\\
        & Repaid & 4,347 (31\%) & 5,128 (40\%) & 768 (53\%) & 118 (54\%)
        & 161 (55\%) & 10,522 (36\%)\\
        \hline
        & Total & 14,079 (100\%) & 12,884 (100\%) & 1,443 (100\%) 
        & 220 (100\%) & 294 (100\%) & 28,920 (100\%)\Tstrut\\
        \midrule
        \multirow{3}{*}{\rotatebox[origin=c]{90}{\tiny{SDT}}} 
        & Defaulted & 3,692 (49\%) & 2,740 (33\%) & 590 (21\%) & 105 (12\%) 
        & 29 (4\%) & 7,156 (35\%)\\
        & Censored & 1,333 (18\%) & 1,915 (23\%) & 715 (26\%) & 255 (30\%)
        & 299 (45\%) & 4,517 (22\%)\\
        & Repaid & 2,526 (33\%) & 3,672 (44\%) & 1,477 (53\%) & 501 (58\%)
        & 343 (51\%) & 8,519 (42\%)\\
        \hline
        & Total & 7,551 (100\%) & 8,327 (100\%) & 2,782 (100\%) & 861 (100\%) 
        & 671 (100\%) & 20,192 (100\%)\Tstrut\\
        \midrule
        \multirow{3}{*}{\rotatebox[origin=c]{90}{\tiny{CMX}}} 
        & Defaulted & 0  & 45 (38\%) & 427 (20\%) & 296 (8\%) 
        & 25 (3\%) & 793 (12\%)\\
        & Censored & 0 & 43 (36\%) & 854 (40\%) & 1,736 (46\%)
        & 392 (47\%) & 3,025 (44\%)\\
        & Repaid & 0 & 32 (27\%) & 847 (40\%) & 1,720 (46\%)
        & 418 (50\%) & 3,017 (44\%)\\
        \hline
        & Total & 0 & 120 (100\%) & 2,128 (100\%) & 3,752 (100\%) 
        & 835 (100\%) & 6,835 (100\%)\Tstrut\\
        \midrule
        \multirow{3}{*}{\rotatebox[origin=c]{90}{\tiny{AAT}}} 
        & Defaulted & 0  & 0 (0\%)  & 54 (17\%) & 181 (12\%) 
        & 24 (6\%) & 259 (12\%)\\
        & Censored & 0  & 0 (0\%)  & 104 (32\%)  & 505 (34\%)
        & 138 (36\%) & 747 (34\%)\\
        & Repaid & 0  & 1 (100\%) & 166 (51\%)  & 781 (53\%)
        & 217 (57\%) & 1,165 (54\%)\Tstrut\\
        \hline
        & Total & 0 & 1 (100\%) & 324 (100\%) & 1,467 (100\%) & 379 (100\%) 
        & 2,171 (100\%)\Tstrut\\
        \bottomrule
    \end{tabular}
\end{adjustbox}
}
\label{tab:bond_summary}
\end{table}

\section{Credit Risk Convergence}
\label{sec:CRC}

This section comprises the methodological contribution of this work:
a new financial econometric hypothesis testing technique to estimate the exact
age two different risk bands converge in conditional default risk (i.e., the
point of {\it credit risk convergence}).
We begin with a review of the relevant
statistical results in Section~\ref{subsec:stat_results}.  We will introduce
the field of survival analysis along with its subfield of competing risks
within the context of loan default modeling.
We then present the financial econometric tools we derive in the form of an
estimator and the associated large sample statistical hypothesis test relied on throughout. The formal statements
are available for reference in Appendix~\ref{subsec:csh_props}, and we provide
complete proofs in the Online Appendix~\ref{subsec:proof2}.
We then move to
using these new financial econometric tools
to perform an empirical study of the ABS data of
Section~\ref{sec:data}.  Specifically, the empirical estimates of credit risk
convergence points
may be found in Section~\ref{subsec:emp_res}.  Because our data
spans the economic events of the Coronavirus pandemic, we include related
COVID robustness analysis in Section~\ref{subsec:COVID}.  This section
closes with Section~\ref{subsec:new_cars}, which is
additional sensitivity analysis related the loan filtering of
Section~\ref{subsec:loan_filter}.

\subsection{Relevant Statistical Results}
\label{subsec:stat_results}

From an economic perspective, not all defaults are equivalent.  For example,
there is an obvious profitability difference between a loan that defaults
shortly after it is originated versus a loan that defaults after a much longer
period of time: a loan that makes more payments before defaulting will be more
profitable, ceteris paribus.\footnote{
The same may be said for prepayments.  We note that all of our results adjust
for prepayments, and these prepayment probabilities are estimable using these
techniques.  See Section~\ref{subsec:cons_perp} for additional details.
}
Therefore, we seek a time-to-event distribution estimate, where
the general event of interest is the end of a loan's payments.
Indeed, we require
this information to adequately address our research question centered around
analyzing a loan's conditional probability of default given its survival.
We are thus in the realm of \textit{survival analysis}, which is a branch
of statistics dedicated to estimating a random time-to-event
distribution.\footnote{
As a nuanced but important theoretical point of emphasis, our data is sampled
from pools of consumer automobile loans found in publicly traded ABS
(see Section \ref{sec:data} for details).  Thus, we must consider an
estimator appropriately calibrated to work in both discrete-time
and with incomplete data subject to random left-truncation and
random right-censoring.  For extended details on these incomplete data
challenges with ABS applications, see the discrete-time work of
\citet{lautier_2021} for the case of left-truncation and
\citet{lautier_2023} for the discrete-time case of both
left-truncation and right-censoring.  Neither \citet{lautier_2021} nor
\citet{lautier_2023} allow for competing risks, however.}

In addition to estimating a time-to-event random variable, we also desire to 
distinguish between the type of event.  Again, from an economic perspective,
this is natural: a loan that is repaid (or prepaid) in a given month is
more profitable than a loan that defaults in the same month, ceteris paribus.
Succinctly,
we wish to differentiate between loans ending in default and loans ending in 
prepayment.  To do so, we can define the problem in terms of a
\textit{competing risks} framework, which is a specialized branch of survival
analysis.
We elect to define competing risks in terms of a multistate process,\footnote{
See \citet[][Example III.1.5]{andersen_1993} or \cite{beyersmann_2009}
for an introduction.
}
which allows us to estimate the conditional CSH rates defined in
\eqref{eq:csh} directly.
Formally, we will be using a multistate process adjusted for left-truncation
and right-censoring in discrete-time but over a known, finite time horizon for
two competing events.\footnote{
In effect, we desire to combine the continuous time estimator of
\citet[][Example IV.1.7]{andersen_1993} adjusted for left-truncation and
right-censoring with the complete data discrete case of
\citet[][pg. 94]{andersen_1993} to obtain a competing risks distribution
estimator suitable for discrete-time and adjusted for left-truncation and
right-censoring.
}
Specifically, we will
generalize the discrete-time, left-truncation and right-censoring work of
\citet{lautier_2023} to the case of two competing events: default and 
repayment.

We now present the mathematical details of the estimator in the context of an
automobile loan ABS.  We will follow the
notation of \citet{lautier_2023} and, for completeness, include some details
regarding accounting for incomplete data.
Define the random time until a loan contract ends by the random variable $X$.
The classical quantity of interest in survival analysis is the 
\textit{hazard rate}, which in discrete-time represents the probability of a 
loan contract terminating in month $x$, given a loan has survived until month
$x$.  We denote the hazard rate by the traditional, $\lambda$, and so
formally,
\begin{equation}
\lambda(x) = \Pr(X = x \mid X \geq x) = \frac{ \Pr(X = x) }{ \Pr(X \geq x) }.
\label{eq:haz_rate}
\end{equation}
Because we desire to model the probability of loan payments terminating given a
loan remains current, it is clear that \eqref{eq:haz_rate} is the ideal
quantity of interest.  Additionally, let $F$ represent the cumulative
distribution function (cdf) of $X$.  If we can reliably estimate
\eqref{eq:haz_rate}, we can recover the complete distribution of $X$ by the
uniqueness of the cdf because
\begin{equation}
1 - F(x-) = \Pr(X \geq x) = \prod_{x_{\min} \leq k < x} \{1 - \lambda(k)\},
\label{eq:haz_surv}
\end{equation}
where $x_{\min}$ is the lower bound of the distribution of $X$. We take the
the convention $\prod_{k=x_{\min} + 1}^{x_{\min}} \{1 - \lambda(k)\} = 1$.

We now account for incomplete data.  To address random left-truncation,
let $Y$ represent the
left-truncation random variable, which is a shifted random variable derived
from the random time a loan is originated and the securitized trust begins
making monthly payments.  That is, we observe $X$ if and only if $X \geq Y$.  
We further assume $X$ and $Y$ are independent, an important assumption we now
briefly justify within a securitization context.  The random variable $Y$
represents
the time an ABS first starts making payments.  Typically, the decision to
issue a securitization is more related to investment market conditions and the
financing needs of the parent company than the performance of the underlying 
assets, in this case automobile loans.  In other words, the forming and 
subsequent issuance of an ABS bond has little to do with the time-to-event
distribution of each individual loan, which is represented by $X$.\footnote{
Indeed, this is the main economic motivation of the securitization process.
}
Hence,
the assumption that $X$ and $Y$ are independent is actually quite reasonable
within the context of the securitization process.
To account for right-censoring,
define the censoring random variable as $C = Y + \tau$, where $\tau$ 
is a constant that depends on the last month the securitization is active and 
making monthly payments.  Note that independence between $X$ and $C$ follows
trivially from the assumed independence of $X$ and $Y$.  We thus observe the
exact loan termination time, $x$, if $x \leq C \mid X \geq Y$, and we only know
that $x > C$ if $x > C \mid X \geq Y$.

For those familiar 
with incomplete data from observational studies, we can think of the period
of time the ABS is active and paying as the observation window. Hence, random
left-truncation occurs because we only observe loans that survive long enough to
enter into the trust, and right-censoring occurs because we only observe the exact
termination time of loans that end prior to end of the securitization.
For completeness, we will assume discrete-time because a borrower's monthly
obligation is considered satisfied as long as the payment is received before
the due date.  Therefore, we may assume the recoverable distribution of $X$ is
integer-valued with a minimal time denoted by $\Delta + 1$ for nonrandom
$\Delta \in \{\mathbb{N} \cup 0\}$, where $\mathbb{N}$ denotes the
natural numbers,
and a finite maximum end point, which we denote by $\xi \geq \Delta + \tau$,
for nonrandom $\xi \in \mathbb{N}$.
We emphasize the word \textit{recoverable}, further discussion of which may be
found in \citet{lautier_2021} and \citet{lautier_2023}.

We now generalize \cite{lautier_2023} to the case of two competing risks as 
follows.\footnote{
For statistically inclined readers,
Appendix~\ref{subsec:csh_props} provides formal statements, and the 
Online Appendix~\ref{subsec:proof2} provides complete proofs of these novel
results.
}
First, consider two competing risks as a multistate process,
such as in Section 3 of \citet{beyersmann_2009}.  Formally, let 
$\{Z_x\}_{\Delta + 1 \leq x \leq \xi}$ be a set of random variables with
probability distributions that depend on $x$, $\Delta + 1 \leq x \leq \xi$.
More specifically, given a loan terminates at time $x$, we assume the loan must
be in one of two states, 
$Z_x \in \{1, 2\}$:\footnote{
It may be of help to see the related \citet[][Figure 1]{beyersmann_2009}.
}
\begin{enumerate}
    \item This is the \textit{event of interest}.  Loans move into this state
    if a default occurs.  The probability of moving into state 1 at time $x$ is
    the \textit{cause-specific} hazard rate for state 1, denoted 
    $\lambda^{01}(x)$.
    \item This is the \textit{competing event}.  Loans move into this state if 
    a prepayment occurs.  The probability of moving into state 2 at time $x$ is
    the cause-specific hazard rate for state 2, denoted $\lambda^{02}(x)$. 
\end{enumerate}
The discrete-time CSH\footnote{
We abbreviated cause-specific hazard as CSH earlier in the manuscript.
For ease of exposition, we have used the full term within the definition of the
states of $Z_x$ to unify the competing risks theory.
}
rate is then defined as
\begin{equation}
    \lambda^{0i}(x) = \Pr(X = x, Z_x = i \mid X \geq x) 
    = \frac{ \Pr(X = x, Z_x = i) }{\Pr(X \geq x)},
    \label{eq:csh}
\end{equation}
for $i = 1, 2$.
Conveniently, therefore, from the law of total probability, we have
\begin{align*}
    \lambda(x) = \frac{ \Pr(X = x) }{ \Pr(X \geq x) }
    &= \frac{ \Pr(X = x, Z_x = 1) }{\Pr(X \geq x)} + 
    \frac{ \Pr(X = x, Z_x = 2) }{\Pr(X \geq x)}\\
    &= \lambda^{01}(x) + \lambda^{02}(x).
\end{align*}
Within a competing risks framework, $\lambda(x)$ 
may be referred to as the \textit{all-cause hazard}.\footnote{
It may be illuminating to review Table \ref{tab:X_probs} in the simulation
study of the  Online Appendix~\ref{sec:sim_study} for a numeric 
example of our competing risk model.
To make the
economic connection between loan default risk over time and the cause-specific
hazard for default (cause 01), we will elucidate the probability that the 
cause-specific hazard rate represents.  Suppose the current age of a loan is 
$x$ months, where $\Delta + 1 \leq x \leq \xi$.  Then the quantity 
$\lambda^{01}_{\tau}(x)$ denotes the probability that a loan will end in 
default in month $x$, \textit{given} it has survived at least $x$ months.
Therefore, if the instantaneous (i.e., monthly)
default risk of a current borrower changes over time, a plot of the
month-by-month hazard rate for a given risk band will provide a current risk
estimate for the borrowers within that risk band who have continued to make 
ongoing payments (i.e., ``survived").  If the hazard rate remains constant,
then the monthly default risk does not change as a loan seasons and continues
to remain actively paying.  On the other hand, if the hazard rate declines
(increases), this would suggest that the current month default risk declines 
(increases) as a loan seasons.
}

Given this framework, it is not difficult to account for securitization data
subject to right-censoring and left-truncation along the lines of
\citet{lautier_2023}. Formally, assume a trust consists of 
$n > 1$ consumer automobile loans.  For $1 \leq j \leq n$, let $Y_j$ denote the
truncation time, $X_j$ denote the loan ending time, and $C_j = Y_j + \tau_j$ 
denote the loan censoring time.  Because of the competing events, we also have 
the event-type random variable $Z_{X_j} = i$, where we observe $Z_{X_j}$ given
$X_j$ for $i = 1,2$.\footnote{
We emphasize that the observable data from a trust,
$\{X_j, Y_j, C_j, Z_{X_j}\}_{1 \leq j \leq n}$ differs from the random
variables,
$\{X, Y, C, Z_X\}$.  For example, the random variables $X$ and $(Y,C)$ are
independent, whereas $X_j$ and $(Y_j, C_j)$ clearly are not.  The reference
\citet{lautier_2021} expounds on this point thoroughly.
}
In what follows, we
will use a subscript of $\tau$ where appropriate to remind us that
right-censoring is present in the data.

If we assume independence between $Y$ and the random vector $(X, Z_{X})$
(not at all unreasonable given the securitzation backdrop and our earlier
discussion), then we may
derive estimators for \eqref{eq:csh} along the same lines as
\citet{lautier_2023}.  We demonstrate as follows.
Let $\alpha = \Pr(Y \leq X)$ and for $i = 1, 2$, define
\begin{align*}
    f^{0i}_{*, \tau}(x) &= \Pr(X_j = x, X_j \leq C_j, Z_{X_j} = i)\\
    &= \Pr(X = x, X \leq C, Z_x = i \mid X \geq Y)\\
    &= \Pr(X = x, x \leq C, Z_x = i, x \geq Y) / \Pr(X \geq Y)\\
    &= (1 / \alpha) \{ \Pr(X = x, Z_x = i) \Pr(Y \leq x \leq C) \},
\end{align*}
and
\begin{align*}
    U_{\tau}(x) &= \Pr(Y_j \leq x \leq \min(X_j, C_j))\\
    &= \Pr(Y \leq x \leq \min(X,C) \mid X \geq Y)\\
    &= (1 / \alpha) \{ \Pr(Y \leq x \leq C) \Pr(X \geq x) \}.
\end{align*}
Thus,
\begin{equation}
    \lambda^{0i}_{\tau}(x) = \frac{ \Pr(X = x, Z_x = i) }{\Pr(X \geq x)}
    = \frac{ f^{0i}_{*, \tau}(x) }{ U_{\tau}(x) }.
    \label{eq:csh_fc}
\end{equation}
In terms of our observable data, for a given loan $j$, $1 \leq j \leq n$, we
observe $Y_j$, $\min(X_j, C_j)$, and $\mathbf{1}_{X_i \leq C_i}$, where 
$\mathbf{1}_{Q} = 1$ if the statement $Q$ is true and 0 otherwise.  Further, 
if we observe an event for loan $j$, we will also observe the information 
$Z_{X_j} = i$, $i = 1, 2$.  Therefore, using the standard estimators
vis-\`{a}-vis the observed frequencies
\begin{equation*}
    \hat{f}^{0i}_{*, \tau, n}(x) =
    \frac{1}{n} \sum_{j=1}^{n} 
    \mathbf{1}_{X_j \leq C_j} 
    \mathbf{1}_{Z_{X_j}=i} 
    \mathbf{1}_{\min(X_j,C_j)=x},
\end{equation*}
and
\begin{equation*}
    \hat{U}_{\tau,n}(x) = 
    \frac{1}{n} \sum_{j=1}^{n} 
    \mathbf{1}_{Y_j \leq x \leq \min(X_j, C_j)},
\end{equation*}
we obtain the estimate for \eqref{eq:csh_fc}
\begin{equation}
    \hat{\lambda}^{0i}_{\tau,n}(x)
    = \frac{ \hat{f}^{0i}_{*, \tau, n}(x) }{ \hat{U}_{\tau,n}(x) }
    = \frac{
    \sum_{j=1}^{n} 
    \mathbf{1}_{X_j \leq C_j} 
    \mathbf{1}_{Z_{X_j}=i} 
    \mathbf{1}_{\min(X_j,C_j)=x}
    }
    {
    \sum_{j=1}^{n} 
    \mathbf{1}_{Y_j \leq x \leq \min(X_j, C_j)}
    }.
    \label{eq:csh_est}
\end{equation}
Pleasingly, \eqref{eq:csh_est} is equivalent to the related classical work of
\citet{huang_1995}, despite our assumption of discrete-time at the problem's
onset.

It is instructive to discuss the form of \eqref{eq:csh_est}.  The nature of the
estimator is completely informed by the data.  In other words, it is an
empirical tool designed to recover the probability distribution of $(X, Z_X)$
based solely on when defaults are observed within the sampled data.  Hence,
any estimates using \eqref{eq:csh_est} are agnostic to any economic model,
which offers a statistical robustness to potentially undue assumptions.
Phrased differently, estimates derived using \eqref{eq:csh_est} may be used as
a benchmark from which future economic models may be calibrated.  Such a
benchmark may serve as a valuable starting point for future research into the
credit risk of individual consumer loans that have been sampled from
ABS.\footnote{
Especially valuable, perhaps, given the recent glut of ABS data now available
via Reg AB II \citep{reg_ab2}.
For a detailed data discussion, see Section~\ref{sec:data}.
}

We are now prepared to introduce our novel financial econometric hypothesis
test.  For a single sample, we refer to \eqref{eq:csh_est} as an estimate. If
we instead consider a population of all possible samples, then we may interpret
\eqref{eq:csh_est} as an {\it estimator}.  Under this interpretation,
\eqref{eq:csh_est} is now a random variable,
and any single estimate is just one possible realization.
As such, the random variable estimator version of \eqref{eq:csh_est} has
a number of attractive asymptotic properties.  First, the complete vector of
estimators over the recoverable space of $X$,
$\hat{\bm{\Lambda}}^{0i}_{\tau,n} = 
(\hat{\lambda}^{0i}_{\tau,n}(\Delta + 1),
\ldots, 
\hat{\lambda}^{0i}_{\tau,n}(\xi))^\top$,
is asymptotically unbiased for the true CSH rates.  Further, 
$\hat{\bm{\Lambda}}^{0i}_{\tau,n}$ is asymptotically multivariate normal
with a completely specifiable diagonal covariance structure (i.e., two
estimators within $\hat{\bm{\Lambda}}^{0i}_{\tau,n}$
are asymptotically independent).  The formal
statement of these properties may be found in Proposition~\ref{thm:asymH} in
Appendix~\ref{subsec:csh_props}.  Additionally, we may use
Proposition~\ref{thm:asymH} to produce asymptotic confidence intervals that are
appropriately bounded within $(0,1)$.  These asymptotic confidence intervals
are available through Lemma~\ref{cor:haz_ci}, which is stated formally in
Appendix~\ref{subsec:csh_props}.

Finally, the asymptotic confidence intervals of Lemma~\ref{cor:haz_ci} and
asymptotic independence of Proposition~\ref{thm:asymH} may be combined
to form a straightforward large sample statistical hypothesis test.  Formally,
for two risk bands $a$, $a'$, where $a \neq a'$ (i.e., $a, a'$ would represent
one of the risk bands deep subprime, subprime, near-prime, prime, or
super-prime), we may test
\begin{equation}
H_0: \lambda_{\tau, (a)}^{0i}(x) = \lambda_{\tau, (a')}^{0i}(x),
\quad \text{v.s.} \quad
H_1: \lambda_{\tau, (a)}^{0i}(x) \neq \lambda_{\tau, (a')}^{0i}(x),
\label{eq:H0}
\end{equation}
for each age $x$ by determining if the $(1-\theta)$\% asymptotic confidence
intervals of the estimators $\hat{\lambda}_{\tau, n, (a)}^{0i}(x)$ and
$\hat{\lambda}_{\tau, n, (a')}^{0i}(x)$ overlap for
$\Delta + 1 \leq x \leq \xi$ and $i = 1, 2$.  The decision rules and
interpretations are as follows.  Fix $x \in \{\Delta + 1, \ldots, \xi\}$ and
$i = 1$.  If the asymptotic confidence intervals of
$\hat{\lambda}_{\tau, n, (a)}^{01}(x)$ and
$\hat{\lambda}_{\tau, n, (a')}^{01}(x)$ overlap, we fail to reject
$H_0$, and we cannot claim
$\lambda_{\tau, (a)}^{01}(x) \neq \lambda_{\tau, (a')}^{01}(x)$.  That is,
conditional default risk given survival to time $x$ has potentially converged.
On the
other hand, if the asymptotic confidence intervals do not overlap, we reject
$H_0$, and we may claim with $(1-\theta)$\% confidence that
$\lambda_{\tau, (a)}^{01}(x) \neq \lambda_{\tau, (a')}^{01}(x)$.  That is,
conditional default risk given survival to time $x$ has not yet converged.
Figure~\ref{fig:conv_demo} is a visualization of this procedure.

\subsection{Empirical Estimates}
\label{subsec:emp_res}

We now apply the financial econometric tools
of Section~\ref{subsec:stat_results} to
the consumer loan data of Section~\ref{sec:data}.  Specifically, we both plot
estimates of the CSH rates for default by loan age and risk
band, $\hat{\lambda}_{\tau,n}^{01}$, and perform the hypothesis test described
by \eqref{eq:H0} to the filtered
loan population summarized in Figure~\ref{fig:2017_summary} and
Table~\ref{tab:bond_summary}.  For convenience of exposition, we will initially
focus our discussion on two risk bands: subprime and prime borrowers.  A
plot of $\hat{\lambda}_{\tau,n}$ by loan age may be found in
Figure~\ref{fig:conv_demo} for the 21{,}332 subprime loans (solid line) and
6{,}300 prime loans (dashed line).

\begin{figure}[t!]
    \centering
    \includegraphics[width=\textwidth]{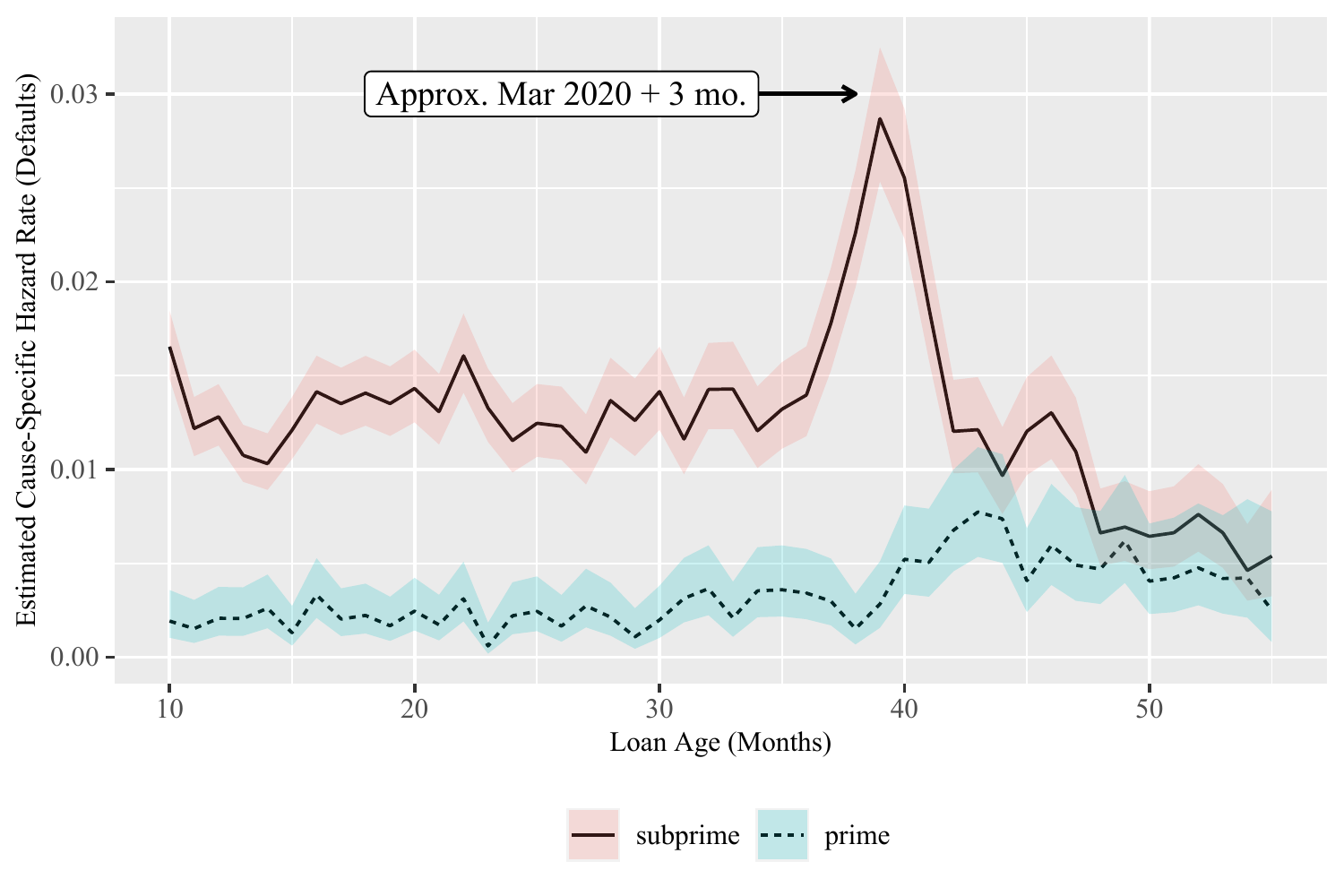}
    \caption{
\footnotesize{   
\textbf{Credit Risk Convergence: Subprime and Prime Loans.}\newline    
A plot of $\hat{\lambda}_{\tau,n}^{01}$ (defaults) defined in
\eqref{eq:csh_est} by loan age for the subprime
and prime risk bands within the sample of 58,188 loans
(Table~\ref{tab:bond_summary}), plus 95\% confidence intervals using
Lemma~\ref{cor:haz_ci}.  We may use the hypothesis test described in
\eqref{eq:H0} by searching for the minimum age that the confidence intervals
overlap.  In this case, we see the first
evidence of credit risk convergence by approximately loan age 42 months. The
upward spike in $\hat{\lambda}_{\tau,n}^{01}$ for the subprime risk
band by loan age 40 is related to the economic impact of COVID-19
(see Section~\ref{subsec:COVID}).
}    
}
    \label{fig:conv_demo}
\end{figure}

As an initial observation, we can see that the estimated default CSH rates for
subprime loans are initially higher than the default CSH rates for prime loans.
This
is expected given our expectations about credit risk, risk-based pricing, and
the difference in APRs between the two risk bands.
This pattern does not maintain for the full lifetime of the loan, however.
As the subprime loans continue to stay current (i.e., given survival), the
CSH rate declines.  This is an alternative visualization of the loan
seasoning observed in Figure~\ref{fig:loss_curves}.
Interestingly, the CSH rates for prime loans in this
sample appear to increase slightly, though they remain generally stable even as
loan age increases. We remark here that, due to left-truncation and
right-censoring, we are unable to fully recover the complete loan term for all
risk bands.\footnote{
We have reliable estimates from approximately $5 \leq X \leq 60$ for all risk
bands, though we report $10 \leq X \leq 55$ for conservatism.  In the instance
of no observed defaults at a particular loan age within the recoverable window,
we interpolate with a constant hazard rate.
}

This brings us to the major methodological result of this paper,
which is the lower right corner of Figure~\ref{fig:conv_demo}.
In addition to plotting the point
estimates, we also provide the asymptotic confidence intervals (shaded regions
surrounding each line).  Eventually, as the two lines slowly approach each
other, the confidence intervals begin to overlap.  The first evidence of this
is around loan age 42, and it is consistent by approximately loan age 50 for
these 72-73 month consumer auto loans.  With the statistical test outlined in
\eqref{eq:H0}, therefore, for any age in which we observe overlapping
confidence intervals,
we cannot claim the true CSH rates for default are different between
the subprime and prime risk bands within this sample.  It is this point at
which two CSH rates for default between two different risk bands
become statistically indistinguishable that we estimate as the
point of \textit{credit risk convergence}.\footnote{
We also remark that measuring default risk conditional on survival gleans
additional insight in comparison to a binary default analysis, such as that
performed in Table~\ref{tab:bond_summary}.  Indeed, 40\% of all subprime loans
in the sample of 58{,}118 defaulted at some point, versus only 10\% of prime
loans.  Given just this analysis, it is not surprising the subprime borrowers
received a higher APR than the prime borrowers.  What we show in
Figure~\ref{fig:conv_demo} is that the default rates conditional on survival
are not constant, however, and it implies that subprime borrowers that do not
refinance are eventually overpaying based on an updated assessment of their
risk profile, all else equal.
We come back to this point much more extensively in
Section~\ref{subsec:cons_perp}.
}

Table~\ref{tab:risk_conv_mat} (top) provides a transition matrix of the
estimated month of credit risk convergence among the five risk bands
considered for the sample of 58{,}118 filtered loans.
For conservatism, we defined the point of credit risk convergence as the
minimum of (1) two consecutive months of confidence interval overlap after
a loan age of 10 months or (2) the minimum shared age that the hazard
estimates are consistently zero.
Based on these results, we would say that a deep subprime loan
eventually converges in risk to a subprime loan after three years, and it
converges to a prime risk after 50 months and a super-prime risk after 52
months.  Similarly, subprime loans converge in risk to prime loans after 42
months, and they become super-prime risks after four years.  Near-prime
loans become prime risks quite quickly, just after one year, and then
become super-prime risks after 34 months.  For completeness, we plot the full
five-by-five matrix of CSH rate and confidence interval comparisons along the
lines of Figure~\ref{fig:conv_demo} in Figure~\ref{fig:haz_grid_default} found
in Appendix~\ref{subsec:emp_add_det}.  Given consumer auto loans are
collateralized with rapidly depreciating assets in the form of used cars
(see Figure~\ref{fig:ltv_by_age}), these results cannot be explained by
traditional LTV optionality arguments found in mortgages
\citep[e.g.,][]{campbell_2015}.

\begin{table}[t!]
\caption{
\footnotesize{
\textbf{Credit Risk Convergence: Transition Matrix.}\newline
This table reports
a summary matrix of the estimated month of credit risk convergence for 72-73
month consumer automobile loans.  The top matrix corresponds to the sample of
58{,}118 loans issued in 2017 (Table~\ref{tab:bond_summary}).  The bottom
matrix corresponds to the sample of 65{,}802 loans issued in 2019 (see
Section~\ref{subsec:COVID}).  For conservatism, the month of credit risk
convergence is defined as the earlier of
(1) the first of two consecutive months after ten
months that the asymptotic confidence intervals for
$\hat{\lambda}_{\tau,n}^{01}$ overlap or
(2) once $\hat{\lambda}_{\tau,n}^{01}$ is consistently zero for
both risk bands.
Visually, it is helpful to compare
Figure~\ref{fig:conv_demo} with the top matrix (subprime, prime) and
Figure~\ref{fig:COVID_demo} with the bottom matrix (subprime, prime).  Full
comparisons may be made with
Figure~\ref{fig:haz_grid_default} in Appendix~\ref{subsec:emp_add_det} and
Figure~\ref{fig:haz_grid_default2019} in Appendix~\ref{subsec:covid_detail}.
}
}
{
\begin{center}
\begin{adjustbox}{max width=\textwidth}
\begin{tabular}{cccccc}
    & \multicolumn{5}{c}{2017 Issuance}\\
    \toprule
    & deep subprime & subprime & near-prime & prime & super-prime\\
    \midrule
    deep subprime & 10 & 36 & 50 & 50 & 52\Tstrut\\
    subprime & & 10 & 23 & 42 & 48\\
    near-prime & & & 10 & 13 & 34\\
    prime & & & & 10 & 10\\
    super-prime & & & & & 10\\
    \bottomrule
    &\\
    & \multicolumn{5}{c}{2019 Issuance}\\
   \toprule
    & deep subprime & subprime & near-prime & prime & super-prime\\
    \midrule
    deep subprime & 10 & 31 & 51 & 58 & 58\Tstrut\\
    subprime & & 10 & 23 & 25 & 42\\
    near-prime & & & 10 & 15 & 15\\
    prime & & & & 10 & 10\\
    super-prime & & & & & 10\\
    \bottomrule
\end{tabular}
\end{adjustbox}
\end{center}
}
\label{tab:risk_conv_mat}
\end{table}

We also see a large spike in the default CSH rate for the subprime risk band by
approximately loan age 40.  Similarly, it appears the prime risk band also
reports
a small increase in its default CSH rate shortly after the same age.  With some
approximate date arithmetic from the first payment month of the ABS bonds
(March-April-May 2017), we find that a loan age of 40
months corresponds to approximately Spring 2020 (when adjusted for
left-truncation).  If we recall the economic impact of the Coronavirus, which
effectively stopped
most economic activity in Spring 2020, it is not difficult to understand why so
many loans defaulted around loan age 40.  This also provides informal 
validation that the data sorting and estimation of the default CSH
rate has been effective.  It is interesting to compare the difference in impact
to subprime and prime borrowers.  That is, the economic shutdown brought on
by the Coronavirus pandemic appears to have had a smaller impact on the prime
risk band than the subprime risk band.  In Section~\ref{subsec:COVID}, we
provide more discussion.

\subsection{Impact of COVID-19}
\label{subsec:COVID}

As alluded to in Section~\ref{subsec:emp_res}, we have attributed the 
large increase around loan age 40 for the default CSH rate 
estimate observable in Figure~\ref{fig:conv_demo} to the Spring 2020
economic shutdown resulting from the initial rapid spread of the Coronavirus 
disease.  Because the point of credit risk convergence occurs after month
40 for some pairs of risk bands in Table~\ref{tab:risk_conv_mat} (e.g.,
deep subprime and prime credit risk convergence occurs by loan age 50), there
is a concern that the point estimate of default risk converging for
disparate risk bands is due to the filtering effect of the shock of the
economic shutdown rather than due to some inherent property of loan risk
behavior.  In other words, only the strongest credits could survive such a
shock, and credit risk convergence may occur later or not at all otherwise.  
While we feel the economic shutdown has played some role, we believe it is not
adequate on its own to explain the credit risk convergence we observed 
in our sample.  We argue as follows.

First, if we return again to Table~\ref{tab:risk_conv_mat}, we can see 
that pairs of risk bands converge earlier than loan age 40 (e.g., deep 
subprime and subprime, near-prime and prime, near-prime and super-prime,
and prime and super-prime).  Thus,
we have examples of risk bands that converge in conditional 
monthly default risk prior to the onset of the Spring 2020 economic shutdown.
Second, if credit risk 
convergence is completely driven by the Spring 2020 economic shutdown, we would
expect to see it occur much earlier in a sample of bonds issued closer to
Spring 2020 when subject to the same loan selection process and risk band
definitions of Section~\ref{subsec:loan_filter}.  Hence, 
we obtained loan level data from the same four consumer auto loan ABS
issuers but from bonds issued closer to Spring 2020:
SDART 2019-3 \citep{sdart_2019}, 
DRIVE 2019-4 \citep{drive_2019},
CARMX 2019-4 \citep{cmax_2019},
and AART 2019-3 \citep{aart_2019}.\footnote{
The filtered 2019 sample mirrors the distribution of the 2017 filtered sample
summarized in Table~\ref{tab:bond_summary}.  For example, there are
31{,}221 DRIVE 2019-4 loans, 19{,}962 SDART 2019-3 loans, 11{,}724 CARMX
2019-4 loans, and 2{,}895 AART 2019-3 loans, for a total of 65{,}802.
By risk band, there are 24{,}107 (37\%) deep subprime loans, 20{,}874 (32\%)
subprime loans, 9{,}930 (15\%) near-prime loans, 8{,}625 (13\%) prime loans,
and 2{,}266 (3\%) super-prime loans.
}
These bonds began paying in late
Summer 2019, whereas the bonds introduced in Section~\ref{sec:data} began 
paying in Spring 2017.

Figure~\ref{fig:COVID_demo} is a repeat of Figure~\ref{fig:conv_demo}; it
presents the estimated CSH rates for default plus asymptotic 95\% confidence 
intervals for the 2019 sample.  As expected, we see the large spike in the
CSH rate for defaults in subprime loans around 10 months, which, when 
adjusted for left-truncation, corresponds to the Spring 2020 economic shutdown.
We also display the estimated credit risk convergence matrix for the 2019
issuance in the bottom portion of Table~\ref{tab:risk_conv_mat}. In reviewing
the matrix, we see evidence of earlier convergence.  Hence, the shock of the
economic shutdown of Spring 2020 has likely played some role.
It is not the whole
story, however.  For example, the subprime risk band in the 2019 issuance does
not converge with the super-prime risk until loan age 42.
In the 2017 issuance,
the subprime risk band converges with the super-prime risk band at loan age 48.
This suggests that loan age or loan seasoning
also plays a role.  Similarly, while 
convergence between risk bands occurs earlier for the 2019 sample, it takes 
more months after the shutdown shock for most disparate risk bands to converge
than after the same shock in the 2017 sample.  For example, the subprime
and prime risk bands converge by loan 25 in the 2019 sample, which is 15 months
after the economic shutdown shock.  For the 2017 sample, however, the
subprime risk band converges with the prime risk band at loan age 42, which is
only 2 months after the economic shutdown.  This further suggests that the
converge results of Table~\ref{tab:risk_conv_mat} are not solely attributable
to the economic event of COVID.

\begin{figure}[t!]
    \centering
    \includegraphics[width=\textwidth]{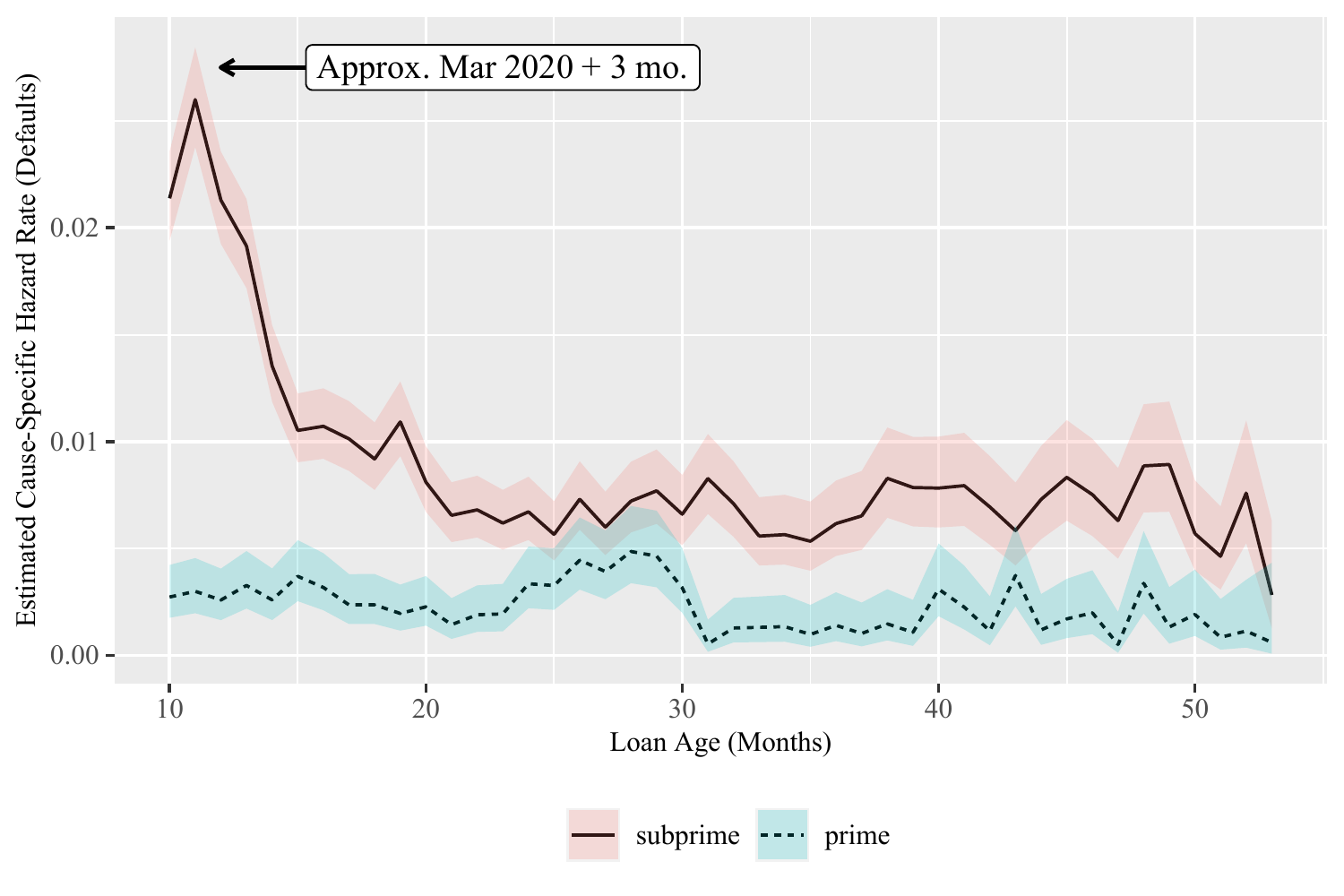}
    \caption{
\footnotesize{    
\textbf{Credit Risk Convergence: COVID Sensitivity.}\newline    
A plot of $\hat{\lambda}_{\tau,n}^{01}$ (defaults) defined in
\eqref{eq:csh_est} by loan age for the subprime
and prime risk bands within the sample of 65{,}802 loans issued in 2019, plus
95\% confidence intervals using Lemma~\ref{cor:haz_ci}.
We may use the hypothesis test described in
\eqref{eq:H0} by searching for the minimum age that the confidence intervals
overlap between two disparate risk bands.
Because the 2019 bonds were issued closer to Spring 2020, the large upward
spike in $\hat{\lambda}_{\tau,n}^{01}$ occurs much earlier for the subprime
risk band, closer to loan age 10 (compare with Figure~\ref{fig:conv_demo}).  We
see some evidence of earlier credit risk convergence around loan age 25 in
comparison to Figure~\ref{fig:conv_demo}.
}    
}
    \label{fig:COVID_demo}
\end{figure}

We also remark that in the last twenty years it is difficult to find a
span of 72 consecutive months in which there was not a large scale economic 
shock (e.g., September 11, 2001; 2007-2009 global financial crisis; 2009-2014
European sovereign debt crisis, COVID-19, etc.).  Hence,
credit risk convergence may be perpetually present, even if it may be
partially explained by the filtering effects of an economic crisis.

\subsection{Additional Sensitivity Analysis}
\label{subsec:new_cars}

With some rudimentary data sorting, the techniques of
Section~\ref{subsec:stat_results} may be used for sensitivity testing.  To
illustrate, we now consider an additional sensitivity test. We instead sort
the data for new cars at the point of sale.  This will give us exposure to a
potentially different borrower profile and depreciating collateral value
pattern.  It will also greatly reduce our exposure to the CARMX bond.
Reduced exposure to CARMX is of interest
because the parent company, CarMax, has an entirely different
business model and therefore financing incentive than either Santander or Ally,
the origination banks of the DRIVE, SDART, and AART ABS bonds.
Because of this, it is possible that CARMX loans behave differently than
loans originated by banks.

We again return to the original collective pool of over 275{,}000 consumer auto
loans of the 2017 issuance of the four bonds introduced in
Section~\ref{sec:data}: CARMX, AART, DRIVE, and SDART.  We then perform the
identical risk band APR-based sorting and loan filtering of
Section~\ref{subsec:loan_filter}, except rather than used cars we restrict our
sample to new cars.  This leaves a total sample of 16{,}412 loans, with bond
exposures of DRIVE (7{,}692), SDART (7{,}369), ALLY (1{,}342) and CMAX (9).
As expected, restricting the sample to new cars has eliminated almost all loans
from CMAX, whose parent company, CarMax, specializes in used auto sales. Thus,
the current sample of 16{,}412 loans consists of loans originated by
traditional banks, Santander and Ally.  In terms of risk band, we maintain
dispersed exposure with deep subprime (3{,}892), subprime (8{,}242),
near-prime (2{,}132), prime (1{,}407), and super-prime (739).  Finally, all
loans consist of a new vehicle at the point of sale, and so we are now
considering an entirely different collateral depreciation pattern and even
potentially borrower profile.  We present an update of both
Figure~\ref{fig:conv_demo} and Figure~\ref{fig:COVID_demo} in
Figure~\ref{fig:new_demo}.

\begin{figure}[t!]
    \centering
    \includegraphics[width=\textwidth]{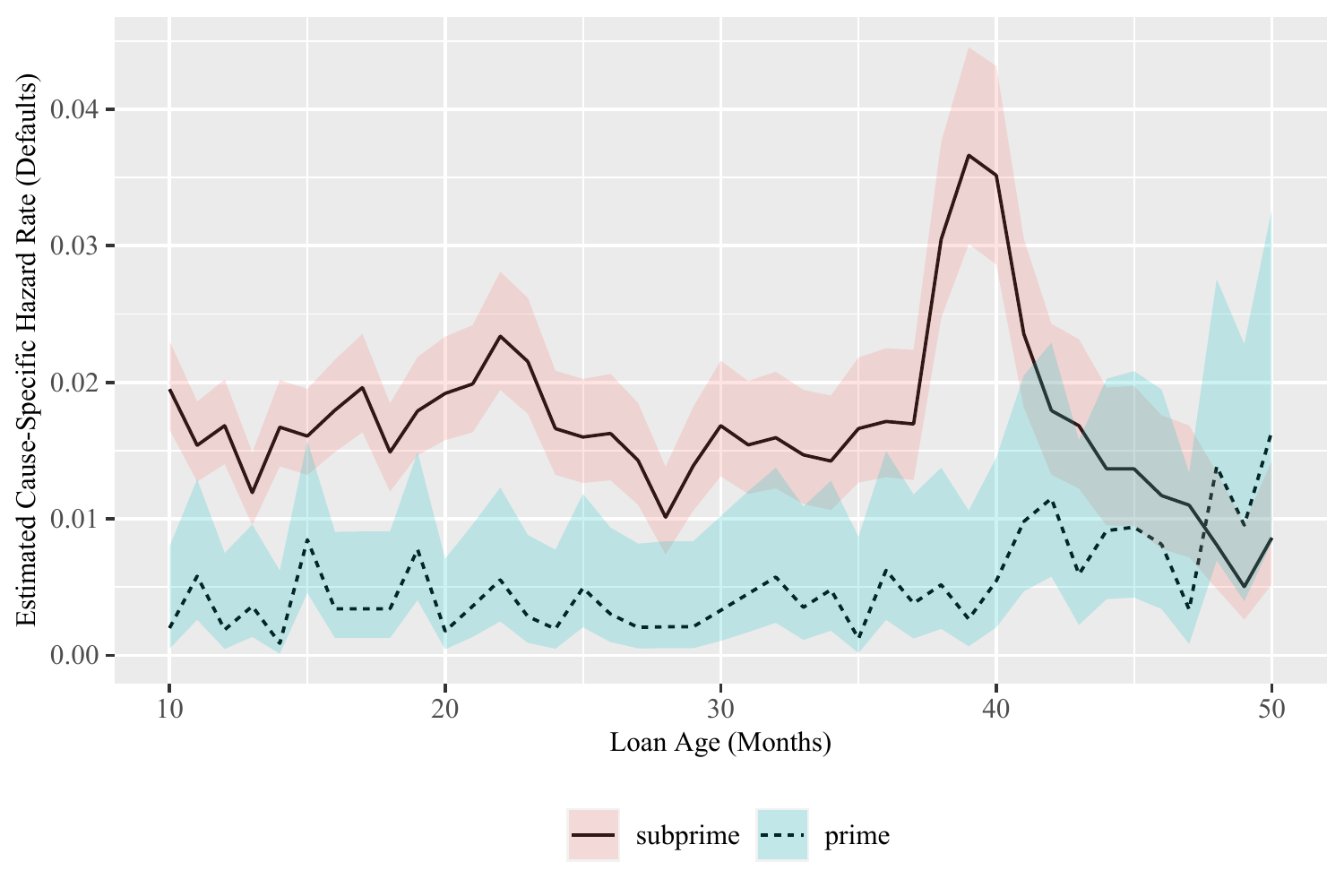}
    \caption{
\footnotesize{    
\textbf{Credit Risk Convergence: Collateral Sensitivity.}\newline    
A plot of $\hat{\lambda}_{\tau,n}^{01}$ (defaults) defined in
\eqref{eq:csh_est} by loan age for the subprime
and prime risk bands within the sample of 16{,}412 loans issued in 2017
with new cars at the point of sale, plus
95\% confidence intervals using Lemma~\ref{cor:haz_ci}.
We may use the hypothesis test described in
\eqref{eq:H0} by searching for the minimum age that the confidence intervals
overlap between two disparate risk bands.
Because of the smaller sample, the asymptotic confidence interval for the CSH
rate of prime loans is wider.  The overall pattern is very similar to
Figure~\ref{fig:conv_demo}, however, and so
the point estimates of credit risk convergence appear to be
robust to collateral type at
the point of sale (i.e., new or used).  The sample of 16{,}412 new car loans
also has minimal exposure to CARMX.  Thus, the CSH estimates further appear
robust to different business incentives of the loan originator.
}    
}
    \label{fig:new_demo}
\end{figure}

Immediately, we see that the overall pattern of Figure~\ref{fig:new_demo}
closely mirrors that of Figure~\ref{fig:conv_demo}.  The subprime loans have
a default CSH rate estimate that is consistently higher
than prime loans in the early
months of a loan's age.  We also see the large increase in the CSH rate for
subprime loans around loan age 40, which correspond to the timing of the
economic shutdown due to COVID-19 in Spring of 2020.  As with the used
cars-at-the-point-of-sale loans, there appears to be minimal impact from
COVID-19 for prime loans.  The two CSH rates for the subprime and
prime risk bands eventually converge, however, which we
see at the lower right corner of Figure~\ref{fig:new_demo}.  The asymptotic
confidence intervals begin to consistently overlap beginning shortly after
loan age 40, which corresponds to row two, column four of the top matrix of
Table~\ref{tab:risk_conv_mat}.
Thus, our credit risk convergence point estimates appear
to be robust in consumer auto loans to the collateral type
at the point of sale (i.e., new or used).  Because the sample of 16{,}412 new
car loans has such minimal exposure to CARMX, we also see that
our credit risk convergence point estimates appear
to be robust to potentially different business
incentives of the parent company to the loan originator (i.e., used car sales
versus traditional banking).

As a final note on collateral type, a close
inspection of Figure~\ref{fig:new_demo} in comparison to
Figure~\ref{fig:conv_demo} reveals wider asymptotic confidence intervals for
the CSH rate for default in prime loans.  This is driven by the smaller sample
size, and it is exacerbated for super-prime loans written on new cars (i.e.,
there are very few defaults for super-prime loans written on new cars in our
sample of 739).  Hence, we have avoided reporting the credit risk convergence
point estimate
matrix of Table~\ref{tab:risk_conv_mat} for the sample of 16{,}412 new car
loans to avoid potentially erroneously conclusions due to faulty asymptotic
statistics stemming from a small default sample.  Instead, we report the point
CSH rate estimates for default for all five risk bands in
Figure~\ref{fig:new_conv_all}.  In this case, a simple line plot speaks
volumes.  In the young ages of a loan, we see that the CSH rates for default
is the highest for deep subprime loans, and it progresses sequentially downward
by risk band until super-prime loans, of which there are very few defaults.
This pattern is expected.  As the loans age, however, we see all CSH rates for
default for each risk band converge together in the bottom right of
Figure~\ref{fig:new_conv_all} near loan age 50. Given consumer auto loans
on new car sales are also collateralized with rapidly depreciating assets,
these results similarly cannot be explained by traditional LTV optionality
arguments found in mortgages \citep[e.g.,][]{campbell_2015}.

\begin{figure}[t!]
    \centering
    \includegraphics[width=\textwidth]{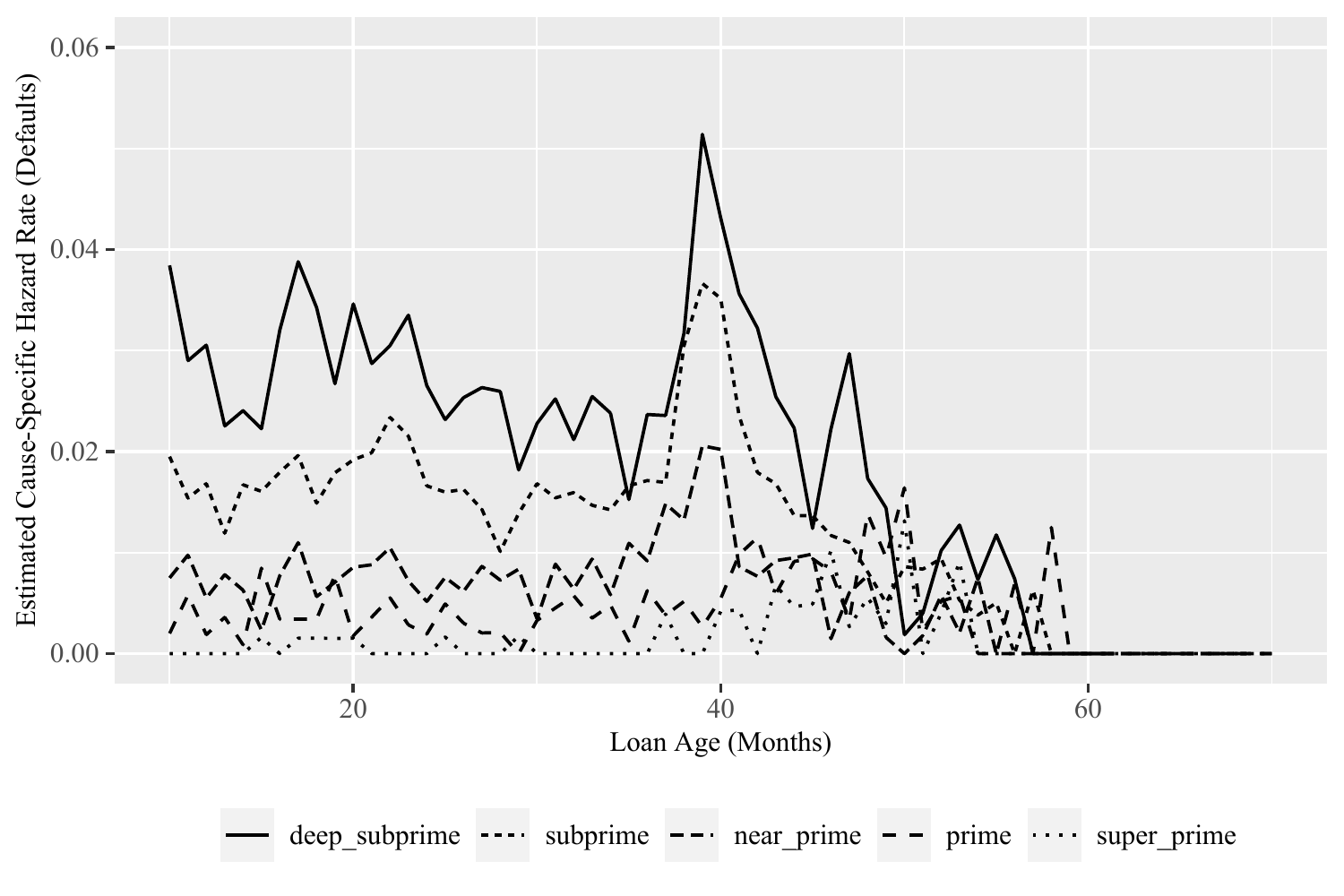}
    \caption{
\footnotesize{   
\textbf{Credit Risk Convergence: All Risk Bands, Point Estimates.}\newline    
A plot of $\hat{\lambda}_{\tau,n}^{01}$ (defaults) defined in
\eqref{eq:csh_est} by loan age for all risk bands within the sample of
16{,}412 loans issued in 2017 with new cars at the point of sale.
As expected, the CSH rates are the highest for deep subprime loans and then
trend downwards until super-prime loans at the onset of loan lifetimes.
As the loans mature and stay current, however, we see that all CSH rates
eventually converge towards zero at the bottom right.
This is an alternative visualization of loan seasoning, to be
compared with Figure~\ref{fig:loss_curves}.
}    
}
    \label{fig:new_conv_all}
\end{figure}

\section{Financial Implications}
\label{sec:finc_imp}

We now apply the novel methods of Section~\ref{sec:CRC} to offer new
financial perspectives on consumer auto loans.
The present section proceeds in two parts.  In
Section~\ref{subsec:lend_prof}, we demonstrate how the CSH estimates
may be used to visualize the back-loading of a lender's expected profits.
Related details for estimating a recovery upon default assumption and
extensions may be found in
Appendix~\ref{subsec:recov} and Online Appendix~\ref{subsec:cer}.  In
Section~\ref{subsec:cons_perp}, we then focus our analysis on the individual
consumer.  By presenting a counterfactual of a perfectly efficient borrower
in terms of credit-based refinancing behavior,
we find that borrowers in all non-super-prime risk bands delay
prepayment inefficiently, all else equal.
In a surprise based on expectations or borrower sophistication,
we find that borrowers in lower risk bands, near-prime and
prime, operate less efficiently than borrowers in higher risk bands, deep
subprime and prime.  Details may be found in Table~\ref{tab:pot_savings}.
We also evaluate borrower conditional prepayment behavior using
the sibling estimator
\eqref{eq:csh_est} for prepayment.  In a visual analysis, we find that
borrower's prepayment decisions appear to be driven by
economic stimulus payments and unusual used auto markets
rather than by financial sophistication.  The section closes with brief
thoughts on why the market for mature consumer auto loans may be
operating at our estimated sub-optimal level of efficiency with respect
to credit-based refinancing.

\subsection{Lender Profitability Analysis}
\label{subsec:lend_prof}

Conventional profitability wisdom of risk-based pricing from the perspective
of a lender is that the
high-returns of high-risk loans that don't default help offset the losses from
the high-risk loans that do default.  In other words, there is an implied
insurance arrangement in which the cost of the losses are dispersed among the
individual borrowers.  Furthermore, it can be argued that through its
precision, risk-based pricing has been attributed to lowering the cost of 
credit for a majority of borrowers and expanding credit availability to higher 
risk borrowers \citep{staten_2015}.\footnote{
See also \citet{livshits_2015} for a more thorough introduction to risk-based
pricing.
}
These are positive economic outcomes, and we do not attempt to argue against
the overall practice of risk-based pricing.  All loans considered
within our analysis have been sampled from pools of securitized bonds, however.
That is, the risk of default has already been transferred off the lender's
balance sheet after the point of sale into the securitized trust.
What we will argue, supported by our novel methods, however,
is that the consumer auto loan market is likely capable
of operating more efficiently with respect to a dynamic view of conditional
default risk.  As one component of this argument,
it is illuminating to perform a lender expected profitability analysis,
especially in light of the default and prepayment probabilistic estimates we
obtained in Section~\ref{subsec:emp_res}.

A common term to describe the profit of a high-risk, high-interest-rate loan
that remains current is \textit{back-loaded}.\footnote{
We thank Jonathan A. Parker for this concise descriptive term.
}
Quite simply, a high-risk, high-interest-rate
loan gradually becomes more profitable as it continues
paying, and it is these increased profits later in the loan's life that offset
the losses taken on other similar loans that have defaulted.  To provide some
formality to this idea, we will utilize an actuarial approach to calculate an
implied, expected risk-adjusted return for each month a loan stays current.
Specifically,
we will examine a rolling monthly expected annualized rate of return assuming
an investor purchases a risky fixed-income asset at a price of the outstanding 
balance of the consumer loan at age $x$ for risk band $a$, $B_{a \mid x}$, with
a one-month term.  This hypothetical risky asset pays either (1) the
outstanding balance at loan age $x+1$ for risk band $a$, $B_{a \mid x+1}$, plus
the next month's payment due, $P_a$, with probability
$1 - \lambda_{\tau,(a)}^{01}(x)$ or (2) the recovery amount at time $x+1$ in
the event of default, $R_{x+1}$, with probability
$\lambda_{\tau,(a)}^{01}(x)$.  Because we utilize a competing risks framework,
the CSH rates are adjusted for prepayment probabilities.
The subscript $a$ denotes one of the five
standard risk bands: deep subprime, subprime, near-prime, prime, and super
prime.  We illustrate this hypothetical asset in Figure~\ref{fig:epv_diagram}.

\begin{figure}[t!]
    \centering
    \includegraphics[width=\textwidth]{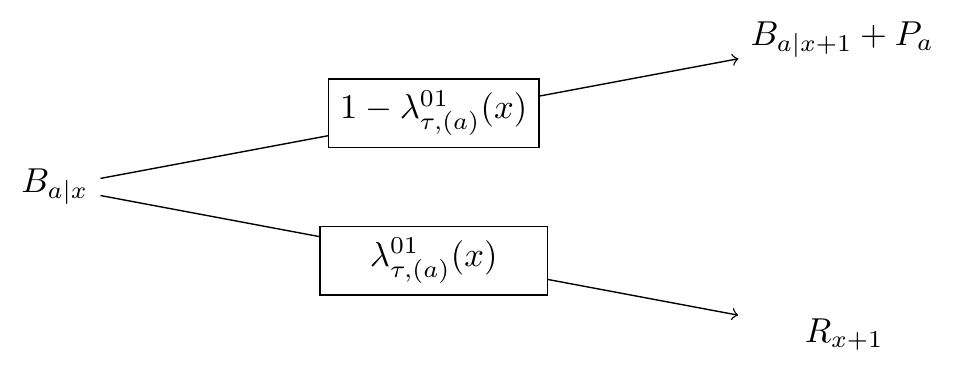}
        \caption{
\footnotesize{
    \textbf{Hypothetical Risky Fixed-Income Asset and Path Probabilities.}\newline    
This hypothetical risky asset pays either (1) the
outstanding balance at loan age $x+1$ for risk band $a$, $B_{a \mid x+1}$, plus
the next month's payment due, $P_a$, with probability
$1 - \lambda_{\tau,(a)}^{01}(x)$ or (2) the recovery amount at time $x+1$ in
the event of default, $R_{x+1}$, with probability
$\lambda_{\tau,(a)}^{01}(x)$. The subscript $a$ denotes one of the five
standard risk bands: deep subprime, subprime, near-prime, prime, or super-prime.
Because we use a competing risks framework, the CSH rates for default are
adjusted for prepayments.
}
}
    \label{fig:epv_diagram}
\end{figure}

To calculate the annualized risk-adjusted return by month, we first define the
expected present value (EPV) of a $B_{a \mid x}$ risky one-month loan depicted
in Figure~\ref{fig:epv_diagram} as
\begin{equation}
\text{EPV}_{a \mid x}^1 =
\lambda_{\tau, (a)}^{01}(x)
\bigg{[} \frac{ R_{x+1} }{1 + \tilde{r}_{a \mid x}} \bigg{]} + 
(1 - \lambda_{\tau, (a)}^{01}(x))
\bigg{[} \frac{ B_{a \mid x+1} + P_a }{1 + \tilde{r}_{a \mid x}} \bigg{]},
\label{eq:EPV_one}
\end{equation}
where $\tilde{r}_{a \mid x}$ is some unknown one-month effective rate of
interest.  To calculate the annualized risk-adjusted return, we can interpret
the outstanding balance of an age $x$ loan in risk band $a$, $B_{a \mid x}$, as
the market-implied price of a risky zero coupon bond following the payment
pattern of Figure~\ref{fig:epv_diagram}.  Therefore, we can use
\eqref{eq:EPV_one} to solve for $\tilde{r}_{a \mid x}$ such that
$\text{EPV}^1_{a \mid x} = B_{a \mid x}$.  This rate, $\tilde{r}_{a \mid x}$,
is then the expected monthly effective risk-adjusted return, which can then be
annualized.\footnote{
We remark that implicit in this analysis is the assumption that the remaining 
payments beyond month $x + 1$ are a tradable asset with no friction
(i.e., the risky asset may be traded
at time $x+1$ for $B_{a \mid x + 1}$).  We can instead perform an expected
risk-adjusted
return calculation over the entire remaining lifetime of the loan (i.e.,
assuming uncertainty for each future payment following the estimates in
Section~\ref{subsec:emp_res}).  The details of how to
perform this full calculation may be found in the Online
Appendix~\ref{subsec:cer}.
}
The calculation in \eqref{eq:EPV_one} also requires an estimate for the
recovery upon default, $R_{x+1}$, for each age $x$.  We perform this estimate
separately for each filtered sample of loans: the 58{,}118 loans from the
four ABS (CARMX, ALLY, SDART, DRIVE) issued in 2017 and summarized in Section~\ref{subsec:summary} and the 65{,}802 loans from the same four
ABS bonds issued in 2019 and summarized in Section~\ref{subsec:COVID}.
The complete details, including a
depicted recovery curve, may be found in Appendix~\ref{subsec:recov} and
Figure~\ref{fig:recovery_est}, respectively.  The probabilities,
$\lambda^{01}_{\tau, (a)}(x)$, for each age, $x$, and risk band, $a$, may be
estimated using the methods of Section~\ref{sec:CRC}.

For ease of interpretation, we consider a single loan of \$100 for
72 months with a payment and amortization schedule determined by the average
APR of each risk band: deep subprime (22.65\%), subprime (17.97\%), near-prime
(12.74\%), prime (7.82\%), and super-prime (3.59\%).\footnote{
These averages are for the 2017 sample of 58{,}118 loans.  The averages for
the 2019 sample of 65{,}802 loans are similar:
deep subprime (22.66\%), subprime (17.67\%), near-prime
(12.55\%), prime (8.34\%), and super-prime (4.49\%).
}
The estimated results for both the 2017 and 2019 issuance may be found in
Figure~\ref{fig:rolling_exp_ret}.
For the 2017 issuance (top), we see that the deep subprime, subprime,
near-prime, and prime risk bands generally group together around 7.5\%
during the earlier part of the loan's lifetime.  This demonstrates that the
risk-adjusted pricing is generally accurate by risk band, as the higher APRs
help offset the higher default risk.  It also reveals that the overall consumer
auto lending is quite efficient across risk bands at origination.
The super-prime
risk-band consistently hovers around a 2.5\% annualized expected risk-adjusted
return.\footnote{
One possible explanation for the super-prime risk-band hovering below the
7.5\% expected
risk-adjusted return of the other risk bands is that there may be other
non-financial benefits to the lender for writing super-prime loans, such as
reduced capital charges for a bank's capital requirements.
}
We then see the negative impact of COVID-19 around loan age 40, which is
consistent with the discussion in Sections~\ref{subsec:emp_res} and
\ref{subsec:COVID}.  It is notable that the impact on the expected
risk-adjusted return for the super-prime risk band due to COVID-19 is minimal.
As the loans mature, however, and credit risk convergence
begins, we see the expected risk-adjusted returns for the higher APR loans
begin to accelerate.  This is a visualization of back-loaded profits.  For the
2019 issuance (bottom), there is a similar clustering in the early months of a
loan's lifetime.  The impact of COVID-19 is much sooner, however, and we also
estimate earlier credit risk convergence between risk bands for the 2019
issuance (see Table~\ref{tab:risk_conv_mat}).  Thus, the risk-adjusted returns
by risk band separate earlier.

\begin{figure}[t!]
    \centering
    \includegraphics[width=\textwidth]{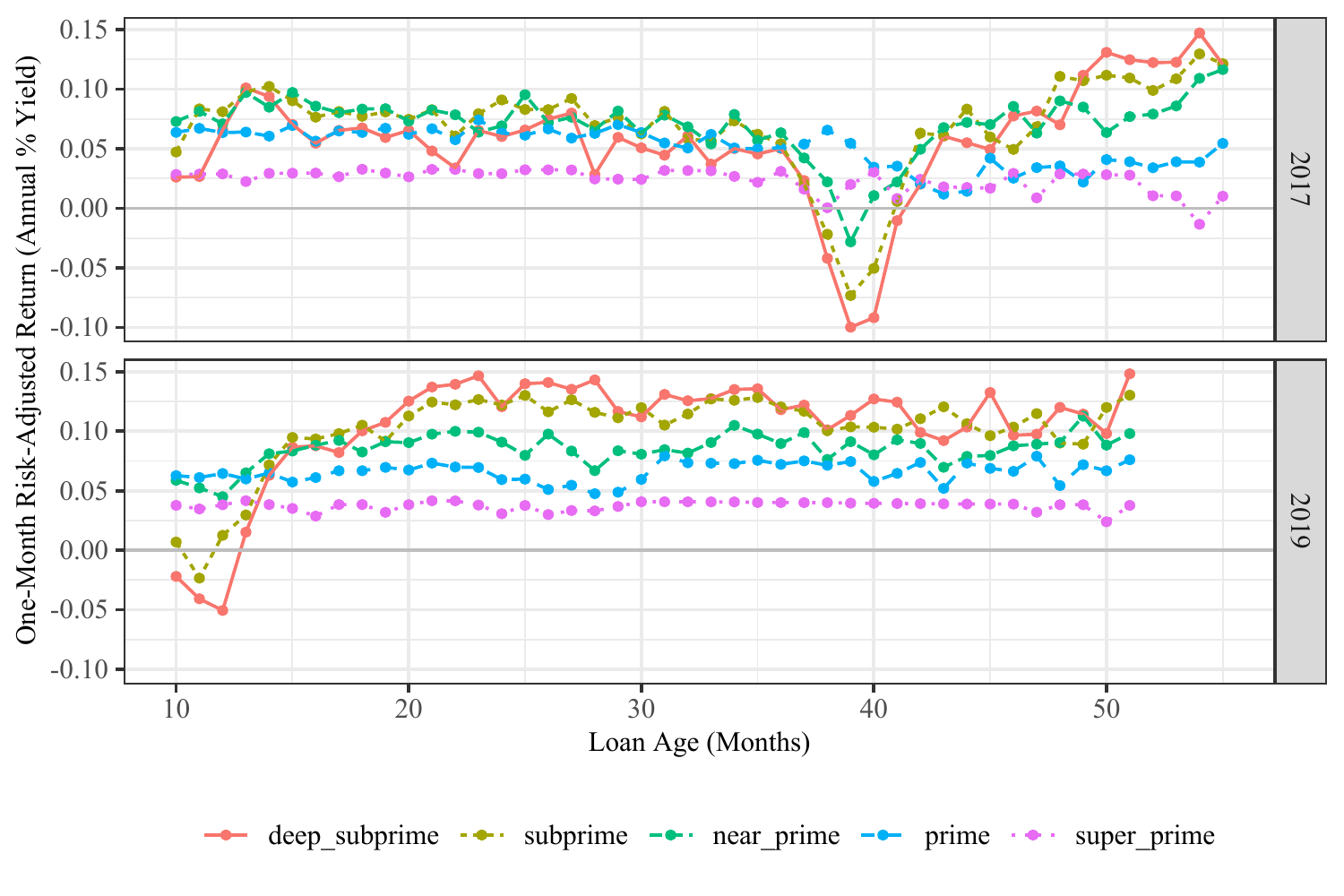}
        \caption{
\footnotesize{
\textbf{Estimated Expected Rolling Risk-Adjusted Return by Age, Issuance.}\newline
A plot of the annualized, expected risk-adjusted one-month return,
$\tilde{r}_{a \mid x}$, by loan age, risk band, and issuance year for the
filtered loan populations summarized in Sections~\ref{subsec:summary} and
\ref{subsec:COVID}.  The calculations utilize \eqref{eq:EPV_one} and the
two-path risky zero coupon bond formulation from Figure~\ref{fig:epv_diagram}.
The probabilities of each path stem from \eqref{eq:csh}, and they may be
estimated with \eqref{eq:csh_est}.  In the 2017 plot (top),
the one-month expected annualized risk-adjusted rate of return is roughly
equal to 7.5\% for the deep
subprime, subprime, near-prime, and prime risk bands until the point of credit
risk convergence (approximately age 40), after which the higher APR risk bands
show accelerating returns.  The clear negative impact of COVID-19 is also
apparent near loan age 40.  The pattern for 2019 is similar, though the impact
of COVID-19 occurs earlier, near loan age 10.
Because we use a competing risks framework, the CSH rates for default are
adjusted for prepayment probabilities.
}
}
    \label{fig:rolling_exp_ret}
\end{figure}

\subsection{Consumer Perspectives}
\label{subsec:cons_perp}

If a borrower's default risk conditional on survival declines as a loan stays
current, but the loan's original APR is a single point-in-time estimate of
risk at origination,
then it is possible a gradual credit-based economic inefficiency from the
perspective of the consumer may develop.
The purpose of the present section is an
attempt to quantify this inefficiency and offer potential explanations for its
appearance, which may be done using the techniques of
Section~\ref{subsec:stat_results}.
The first part will be dedicated to estimating a dollar amount via
a comparison with the counterfactual of a perfectly efficient borrower in terms
of credit-based refinancing behavior (ceteris paribus),
and the second part will offer observations on consumer behavior and larger
market behavior.

As an initial starting point, it may be tempting to trivialize the results
of Table~\ref{tab:risk_conv_mat} as a simple artifact of collateralized
loans.  Without careful thought, it is not unreasonable to suppose that a
72-month auto loan with less than two years remaining will almost certainly
be ``in-the-money'', and so the declining conditional default risk would
naturally follow.  This faulty reasoning ignores the rapidly depreciating
value of the collateral of used automobiles.  As a reference point,
\citet{storchmann_2004} estimates an average annual depreciation of 31\% in
Organization for Economic Co-operation and Development (OECD) countries.
Further, it is not uncommon to see deep subprime loans with APRs north of
20\% (see Figure~\ref{fig:2017_summary}), hindering a borrower's ability to
pay down principal.
Hence, Figure~\ref{fig:ltv_by_age} presents an estimated LTV by loan age
for current loans in our filtered sample of 51{,}118 loans.  It is not
until loan age 60 that super prime loans finally get under an LTV of 100\%,
and the riskier bands possess LTVs largely north of 150-200\% well beyond
the convergence point estimates of Table~\ref{tab:risk_conv_mat}.  Given
these estimates, it is of interest that we find conditional credit risk
behavior that cannot be explained by the standard in-the-moneyness analysis
of mortgages \citep[e.g.,][]{deng_1996},
a perhaps unique economic feature of consumer auto loans.\footnote{
In a robustness check, we halve the 31\% depreciation rate of
\citet{storchmann_2004} to 16\% annually, and the deep subprime and
subprime risk bands keep LTVs north of 100\% beyond loan age 52, the latest
convergence point in Table~\ref{tab:risk_conv_mat}.
}

\begin{figure}[t!]
    \centering
    \includegraphics[width=\textwidth]{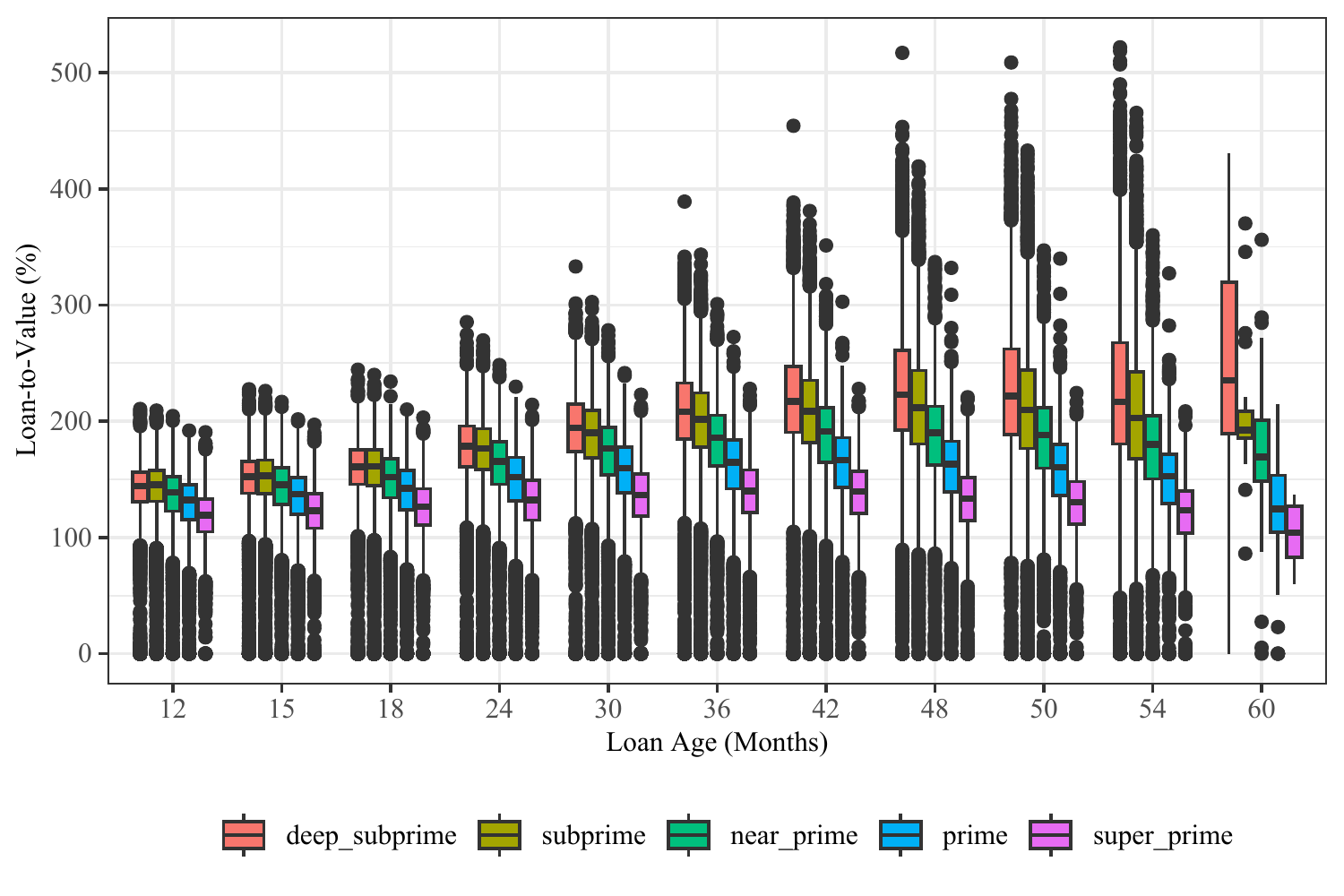}
        \caption{
\footnotesize{
\textbf{Outstanding Loan-to-Value by Loan Age, Risk Band.}\newline
A standard box plot of the outstanding LTV by loan age and risk band
for current loans in the
filtered sample of 51{,}118 loans from the ABS bonds CARMX, ALLY,
SDART, and DRIVE issued in 2017 and summarized in Section~\ref{subsec:summary}.
Average LTVs remain well-above 100\% past the point of credit risk convergence,
which suggests that improved credit performance is not attributable to
borrowers building equity in the collateral.
Auto depreciation estimates are based on \citet{storchmann_2004}.
}
}
\label{fig:ltv_by_age}
\end{figure}

We estimate such potential credit-based refinance savings in
Table~\ref{tab:pot_savings}, all else equal.  Moving
left-to-right along the column headings, we first report a count of the current
loans by loan age.  Next, of the active loans, we present an average
outstanding balance, average payment, and average APR.  The ``Pmts (\#)''
column calculates the remaining payments needed to pay-off the average loan
balance given the average payment.  The next four columns represent the
potential savings in monthly payment if a borrower refinances at the average
APR of the lower risk band, after the estimated point
of credit risk convergence.  If two
risk bands have not yet converged in credit risk, the numbers are not provided
in the table.  The estimated convergence points may be found in
Table~\ref{tab:risk_conv_mat}.  The calculations do not assume any upfront
refinancing charge.

We find that borrowers in all four non-super-prime risk bands,
deep subprime, subprime, near-prime, and prime, appear to
leave money on the table.  On a monthly payment basis,
deep subprime borrowers begin to overpay between \$11-63 per month around loan
age 36, for a total potential savings between \$193-1{,}153.  Based on our
estimates, deep subprime borrowers would benefit the most in terms of total
savings by waiting until approximately loan age 50, when they converge in risk
with prime borrowers.  In terms of monthly payment savings, deep subprime
borrowers should wait to refinance until they converge in credit risk with
super-prime borrowers.  Encouragingly, we see that most deep subprime borrowers
have prepaid or refinanced by loan age 60, which suggests some self-correction,
albeit slower than our calculations would recommend.
The situation for subprime
borrowers is similar; they benefit the most in total savings by refinancing by
loan age 42, when they converge in credit risk with prime borrowers.  Overall,
the potential total savings over the life of the loan for subprime borrowers
ranges between \$299-1{,}616.  In terms of monthly payment, subprime borrowers
benefit the most by waiting until loan age 48, when they converge in credit
risk to super-prime borrowers. In total, the potential monthly payment savings
for subprime borrowers ranges between \$22-61.  As with deep subprime
borrowers, it seems most have refinanced by loan age 60.  While this is slower
than our calculations would suggest, it still indicates borrowers may be
attempting to self-correct.  These results would be in addition to any
consumer refinance inefficiency attributable to changes in interest rates
\citep[e.g.,][]{keys_2016, agarwal_2017, andersen_2020}.

In moving to discuss borrowers in lower risk bands, we find slightly different
results.  As with deep subprime and subprime borrowers, we also find evidence
that near-prime and prime borrowers operate inefficiently with respect to a
credit-based refinance, ceteris paribus.  We
estimate that near-prime borrowers are eligible for a potential monthly payment
savings of \$13-56 with a potential total savings of \$160-2{,}206.
The figures
for prime borrowers are similar; a potential \$18-39 in monthly savings with a
potential total savings of \$261-2{,}327.  On the other hand, we find that both
near-prime and prime borrowers should refinance as soon as possible, after 15
months for near-prime borrowers when they converge in credit risk with prime
borrowers and after 12 months for prime borrowers when they converge in credit
risk with super-prime borrowers.  We find that both near-prime and prime
borrowers do not start refinancing in earnest until approximately loan age 60,
similar to borrowers in the higher risk bands.  This suggests that near-prime
and prime borrowers manage their loans less efficiently than deep subprime and
subprime borrowers, a result that is surprising given typical expectations
about borrower sophistication and credit score.\footnote{
As an alternative interpretation, it may be that the greater affluence of
near-prime and prime borrowers allows a non-optimal efficiency to persist out
of the perceived inconvenience of going through a refinancing versus the
potential savings.  We thank Susan Woodward for this observation. It also of
interest to compare this finding with the profitability analysis of FHA-insured
mortgages in \citet{deng_2006}.
}
We note that the savings
assuming the 2019 transition matrix (Table~\ref{tab:risk_conv_mat}), given
its earlier convergence points, are generally more substantial for the
recoverable estimates.\footnote{
We have omitted these figures for brevity and
conservatism, but the corresponding author may be contacted for details.
}

\begin{table}[t!]
\caption{
\footnotesize{
\textbf{Estimated Savings by Risk Band, Loan Age.}\newline
Potential savings for a borrower who refinances at the average interest
rate of a superior risk band after the point of credit risk convergence in
Table~\ref{tab:risk_conv_mat} (S = subprime, NP = near-prime, P = prime,
SP = super-prime).
}
}
\scriptsize{
{
\begin{adjustbox}{max width=\textwidth}
\begin{tabular}{ccccccccccccccc}
    & & & \multicolumn{3}{c}{Averages} & & \multicolumn{4}{c}{Mo Pmt Savings (\$)}
    & \multicolumn{4}{c}{Total Savings (\$)}\\
    \cmidrule(r{0.5em}){4-6} \cmidrule(r{0.5em}){8-11} \cmidrule(l{0.5em}){12-15}
    & Age & Cnt & Bal (\$) & Pmt (\$) & APR (\%) & Pmts (\#) & S & NP & P & SP & S & NP & P & SP\\
    \toprule
    \multirow{11}{*}{\rotatebox[origin=c]{90}{\small{deep subprime}}} 
&12&17,558&14,245&365&22.58&65&&&&&&&&\\
&15&16,125&13,844&364&22.56&62&&&&&&&&\\
&18&14,375&13,520&363&22.54&60&&&&&&&&\\
&24&11,628&12,836&361&22.50&56&&&&&&&&\\
&30&9,492&11,973&361&22.46&50&&&&&&&&\\
&36&7,746&10,985&359&22.46&44&16&&&&586&&&\\
&42&6,050&9,833&357&22.46&38&16&&&&490&&&\\
&48&4,899&8,799&358&22.43&33&18&&&&438&&&\\
&50&4,622&8,312&358&22.44&30&12&33&52&&267&729&1,153&\\
&54&3,568&7,485&360&22.37&26&11&30&47&61&193&531&845&1,093\\
&60&12&6,923&377&22.00&23&21&39&54&63&251&466&643&759\\
    \midrule
    \multirow{11}{*}{\rotatebox[origin=c]{90}{\small{subprime}}} 
&12&18,261&16,693&395&17.97&64&&&&&&&&\\
&15&17,021&16,126&394&17.96&61&&&&&&&&\\
&18&15,487&15,619&393&17.95&59&&&&&&&&\\
&24&12,997&14,621&389&17.94&54&&32&&&&1,557&&\\
&30&11,021&13,420&388&17.94&48&&30&&&&1,275&&\\
&36&9,309&12,194&386&17.94&42&&25&&&&904&&\\
&42&7,481&10,835&384&17.93&37&&29&54&&&857&1,616&\\
&48&6,192&9,506&383&17.92&31&&22&44&61&&526&1,055&1,473\\
&50&5,901&8,953&383&17.93&29&&23&44&60&&508&963&1,325\\
&54&4,542&7,975&386&17.94&25&&22&40&55&&389&723&988\\
&60&22&7,021&414&17.47&20&&25&40&50&&299&477&596\\
\midrule
    \multirow{11}{*}{\rotatebox[origin=c]{90}{\small{near-prime}}}
&12&5,807&19,111&411&12.79&64&&&&&&&&\\
&15&5,587&18,245&407&12.76&60&&&39&&&&2,206&\\
&18&5,315&17,617&405&12.74&58&&&40&&&&2,158&\\
&24&4,692&16,204&402&12.72&52&&&35&&&&1,657&\\
&30&4,146&14,694&400&12.71&47&&&37&&&&1,546&\\
&36&3,592&13,187&398&12.71&41&&&31&56&&&1,116&2,000\\
&42&3,041&11,446&394&12.67&35&&&28&49&&&847&1,481\\
&48&2,622&9,862&394&12.68&29&&&21&39&&&494&928\\
&50&2,455&9,283&395&12.69&27&&&20&37&&&436&811\\
&54&1,663&8,218&400&12.69&24&&&29&44&&&526&798\\
&60&63&6,435&413&11.98&17&&&13&22&&&160&269\\
\midrule
    \multirow{11}{*}{\rotatebox[origin=c]{90}{\small{prime}}}
&12&5,173&18,582&358&7.83&64&&&&39&&&&2,327\\
&15&5,283&17,611&354&7.81&60&&&&33&&&&1,880\\
&18&5,315&16,706&350&7.78&57&&&&30&&&&1,627\\
&24&4,971&15,097&346&7.76&52&&&&32&&&&1,535\\
&30&4,538&13,503&345&7.74&46&&&&30&&&&1,245\\
&36&4,096&11,866&344&7.73&39&&&&21&&&&755\\
&42&3,697&10,274&342&7.72&34&&&&23&&&&703\\
&48&3,191&8,615&343&7.71&28&&&&21&&&&513\\
&50&2,963&8,101&345&7.71&26&&&&21&&&&460\\
&54&1,898&7,075&351&7.66&22&&&&18&&&&324\\
&60&92&4,756&328&7.38&16&&&&22&&&&261\\
    \bottomrule
\end{tabular}
\end{adjustbox}
}
}
\label{tab:pot_savings}
\end{table}

It is of further interest to examine loan prepayment behavior, which is also
possible with the techniques of Section~\ref{subsec:stat_results}.
Specifically, the
CSH rate estimator defined in \eqref{eq:csh_est}, $\hat{\lambda}^{02}_{\tau,n}$,
is a direct estimator for prepayment behavior, also conditional on survival.
Hence, we can report similar figures to Section~\ref{subsec:emp_res} but
instead focus on borrower prepayment behavior conditioning on the set of
current loans.  From this, we can attempt to
explain consumer behavior and assess if borrowers are acting on the potential
savings reported in Table~\ref{tab:pot_savings}.  For context, we also overlay
two additional economic variables.  The first is the Manheim Used Auto Price
Index\footnote{
For reference, the Bloomberg ticker is \texttt{MUVVU}.  We report seasonally
adjusted figures.
} \citep{manheim_2023}, which is a common industry assessment of the prevailing
value of used automobiles.  Given the unusual observations in the used auto
market during the COVID-19 pandemic \citep{rosenbaum_2020}, it is possible that
higher-than-expected trade-in values motivated consumers to prepay their
loans.\footnote{
This would likely reduce the annual depreciation rate of used automobiles.
In a robustness check, we halve the 31\% depreciation rate of
\citet{storchmann_2004} to 16\% annually, and the deep subprime and
subprime risk bands keep LTVs north of 100\% beyond loan age 52, the latest
convergence point in Table~\ref{tab:risk_conv_mat}.
}
Additionally, the United States federal government provided individuals with
three direct payments known as Economic Impact Payments (EIPs) and expanded the
Childcare Tax Credit (CTC) during the observation period of our sample
\citep{gao_2022}.  It is thus possible that borrowers, upon receiving these
cash payments, made the decision to purchase a different vehicle and thus
prepay. The results are presented in Figure~\ref{fig:repay_all}.

\begin{figure}[t!]
    \centering
    \includegraphics[width=\textwidth]{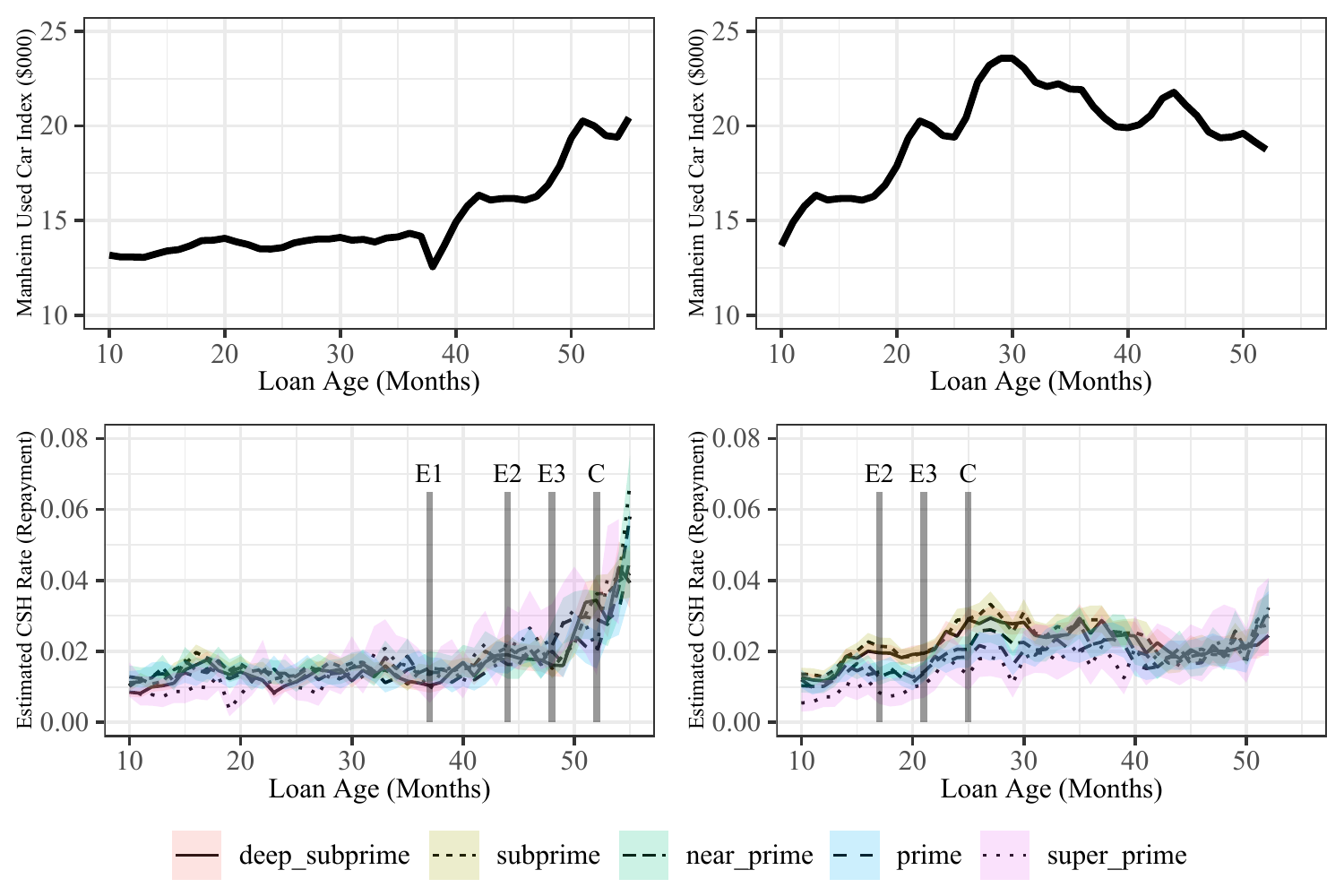}
    \caption{
\footnotesize{  
\textbf{Consumer Prepayment Behavior, Used Autos, Economic Stimulus.}\newline  
(top) A plot of the Manheim Used Auto Index (price) \citep{manheim_2023}
by approximate loan age for the sample of 58{,}118 filtered loans issued
in 2017 (left) and 65{,}802 loans issued in 2019 (right).
(bottom) A plot of $\hat{\lambda}^{02}_{\tau,n}$ (prepayments) defined in
\eqref{eq:csh_est} by loan age for all risk bands within the sample of
58{,}118 filtered loans issued in 2017 (left) and 65{,}802 loans issued in
2019 (right), plus 95\% confidence intervals using Lemma~\ref{cor:haz_ci}.
By the hypothesis test defined in \eqref{eq:H0}, there is very little
difference in prepayment behavior conditional on survival by risk band.
The labels E1, E2, E3, and C indicate the timing of the Economic Impact
Payments and Childcare Tax Credit expansion \citep{gao_2022}.
}    
}
    \label{fig:repay_all}
\end{figure}

There appears to be very
little difference in prepayment behavior by risk band throughout the life of
the loan, which differs significantly from default rates.  This is especially
so for the sample of 58{,}118 loans issued in 2017.  Visually, we see some
differences in the sample of 65{,}802 loans issued in 2019, but many of the
asymptotic confidence intervals still overlap by risk band.  Further, there
does appear to be a meaningful connection between prevailing used auto prices
and borrower prepayment behavior.  That is, as the value of used autos rose,
borrowers of current loans appear to increase prepayment frequency.  Indeed,
prepayments occur at a higher rate sooner in the 2019 issuance, when the value
of used autos increased earlier in the loan's lifetimes in comparison to the
2017 issuance.  Furthermore, the timing of economic stimulus payments plotted
against prepayment behavior is also telling.  The prepayment rates for both the
2017 and 2019 issuance also increase shortly after individuals would have
received the first direct EIP from the U.S. federal government.  Because of
the potential savings we observe in Table~\ref{tab:pot_savings}, it is possible
that the EIPs may have also provided individuals with further implicit economic
gains, if they used the EIPs to refinance at a lower interest rate.  The
results of Figure~\ref{fig:repay_all} in connection with
Table~\ref{tab:pot_savings} taken together suggest that individual borrowers
may not consider their updated risk profile in deciding to prepay.  Instead,
the borrowers may be more motivated by economic indicators that are more
tangible, such as direct cash payments or higher trade-in values.

Given Table~\ref{tab:pot_savings}, it begs the question: why
does the market for mature consumer auto loans appear to operate inefficiently
with respect to credit-based refinancing?
A natural starting point is a lack of borrower sophistication in performing
an updated personal risk assessment as a loan remains current.  Generally,
the typical consumer has a poor reputation in making financial decisions
\citep[e.g.][]{gross_2002, stango_2011, lusardi_2013, campbell_2016,
heidhues_2016, dobbie_2021},
and the type of calculations we perform herein assume some advanced expertise,
such as a working understanding of actuarial mathematics.  An inability to
self-assess creditworthiness within financial markets against a current APR
seems to plague borrowers within all risk bands, as we find the surprising
result that it is actually the near-prime and prime borrowers that leave the
most money on the table by delaying prepayment, ceteris paribus.

It may not be fair to blame this perceived borrower inefficiency solely on the
borrowers, however.  A borrower's main tool to assess creditworthiness is their
credit score.  While consumers have obtained better access to credit scores,
they may update too slowly within the context of a 72-73
month consumer auto loan to motivate a borrower to seek out a lower rate.
Additionally, such borrowers may face friction in
attempting to refinance mature auto loans, either through limited options,
refinancing fees, or perceived hassle.  Indeed, encouraging borrowers to
self-correct has proven to be less effective in practice
\citep[e.g.,][]{keys_2016, agarwal_2017}.
From this point of view, we see an opportunity for lenders to target these
mature loans from borrowers in higher risk bands.\footnote{
There are examples of specialty finance companies in the student loan space
that attempt to refinance borrowers into lower interest rates (e.g.,
SoFi).  The size (and potential profitability) of such loans may be larger than
auto loans, however.  In addition, given all students loans are originally
subject to the same underwriting standards and the wide disparity of ultimate
educational outcomes, the level of risk mispricing is likely more egregious
and thus easier to exploit than for auto loans.
}
Because a borrower that stays current eventually outperforms their initial risk
profile and loan APRs are constant throughout the life of the loan, it is not
a leap in logic to suggest there exists a lower rate that would both lower 
this borrower's financing cost and be profitable to a second lender.
On the other hand, lenders themselves may face similar market frictions, such
as an inability to identify these borrowers or unattractive returns after
accounting for the full scope of origination costs.\footnote{
For example, refinancing a secured automobile loan will require updating the
value of the underlying collateral.  Because it is an automobile, the possible
depreciation may outpace the equity position, especially for high-APR loans on
used automobiles.  This may potentially change the loan-to-value ratio at the
point of refinance (see Figure~\ref{fig:ltv_by_age}).
Housing generally appreciates, conversely, which is an interesting
contrast.  For the estimates within this manuscript, such as in
Table~\ref{tab:pot_savings}, the odd behavior of the used auto markets over
our observation period
(i.e., Figure~\ref{fig:repay_all}) suggests borrowers may have even more
substantial gains if refinancing at unusually strong collateral values.
We thank
Chellappan Ramasamy for drawing our attention to the nuances of the asset
component required in any such refinance calculation.
}
We are optimistic that
continued increases in financial technology may lower these possible hurdles
for both borrowers and lenders.

To spur future research, we suggest two potential solutions.
The first is that we
see a market ripe for financial innovation.  Specifically, we propose that 
lenders offer a loan structure with a reducing payment based on good 
performance, an \textit{adjustable payment loan}.%
\footnote{President Barack Obama remarked during the signing of the
Dodd-Frank Wall
Street Reform and Consumer Protection Act that, ``We all win when consumers are
protected against abuse.  And we all win when folks are rewarded based on how
well they perform, not how well they evade accountability''\citep{obama_2010}.
The terms ``abuse'' and ``evade accountability'' feel strong, but attempting
to reward borrowers based on good performance feels aligned in spirit with an
ideal of merit-based economic gains.}%
\textsuperscript{,}%
\footnote{See \citet{zhang_2023} for an automatically refinancing mortgage
model based on changes in interest rates.}
It is likely lenders already
possess the data needed to provide pricing structures capable of adjusting for
a borrower's updated risk profile.  We postulate that a lower future payment
may act as an incentive for a borrower to remain active and paying, which could
work to offset potential profit losses from lowering rates to these
high-interest rate loans that perform well.  We caution lenders from making
opposite adjustments, however, as increasing payments in response to poor
performance (i.e., sudden delinquencies) may further discourage a likely
overwhelmed borrower or lead to adverse selection (though late payment
penalties are common).\footnote{
It is possible that the process of
securitization, whereby default risk is transferred off a lender's balance
sheet, creates a disincentive for lenders to maintain a continued interest in
loan performance.  At the same time, for loans not securitized, it is difficult
to ask a for-profit lender to actively seek out lower profits.  We suspect the
branch with the most potential fruit is a second opportunistic lender, or
perhaps some speciality finance companies connected with \textit{responsible
investing} (i.e., environmental, social, and governance (ESG), socially
responsible investing (SRI), or impact investing).
}
Second, there is always the regulatory angle, which has
been successful in other consumer lending spaces
\citep[e.g.,][]{stango_2011,agarwal_2014}.\footnote{
One very positive example is Reg AB II \citep{reg_ab2}, which has made all the
data used herein freely available to the public.
}
For example, there is potentially
minimal additional cost to lenders to require ongoing loans to be underwritten
again after a set period of good performance, say 36 months, especially given
the lender will already have most of the borrower's information.  Ideally,
this update would not count as a formal inquiry against the borrower's
credit report.\footnote{
Relatedly, sending reminder notices about refinancing has had some success
\citep[e.g.,][]{byrne_2023}.
}
Further, given the results of Figure~\ref{fig:repay_all}, an
initial cash payment incentive to borrowers may provide sufficient
motivation to get borrowers to refinance.\footnote{
The overall economic impact of such a program may be mixed, given the results
for the ``cash for clunkers'' program \citep{mian_2012}.  Alternatively,
competing lenders themselves may offer cash to borrowers in exchange
for refinancing.
}
On the other hand, regulatory intervention to increase the cost of lending may
lead to these extra costs being pushed back to the borrowers.  We leave these
proposals open to further research.

\section{Conclusion}
\label{sec:conclusion}

This article tells a familiar financial story in a new way, with added details.
We arrive at the familiar aspect, collateralized loan seasoning and
consumers behaving in a way that is financially inefficient, but we do so in
a relatively unstudied asset class with rapidly depreciating collateral values
-- consumer automobile loans -- and find
an inefficiency due to a {\it credit-based} rather than interest rate-based
refinancing,
all else equal.  Furthermore, we employ novel financial econometric techniques
to statistically pinpoint convergence ages between risk bands and provide
economic estimates of forgone potential credit-based refinance savings.
Specifically, we estimate the point of
credit risk convergence using a financial econometric hypothesis test we
derive via large-sample asymptotic statistics from the field of survival
analysis.
We analyze over 140{,}000 consumer automobile loans from
three different samples taken from ABS bonds spanning nearly six years:
Spring 2017 through Winter 2023.  Our techniques allow for the analysis of
loans sampled from ABS, rather than be restricted to direct consumer loan
data, and they are appropriately calibrated for discrete-time data.

We estimate that conditional credit risk converges between all disparate risk
bands in our sample of 72-73 month auto loans prior to scheduled termination.
The rate of convergence depends on the two risk bands being compared,
with the full transition matrix in Table~\ref{tab:risk_conv_mat}.
These results are
robust to various sensitivity tests for the Coronavirus pandemic, loans secured
with new or used vehicles, and the business model of the loan originator's
parent company.
Further, the financial econometric tool we derive is completely
informed by the data; it is model agnostic and thus robust to undue model
assumptions.  Hence, our probabilistic estimates
in Section~\ref{sec:CRC} may be used as a benchmark to calibrate future
economic models.  Finally, the estimated LTVs north of 100\% throughout the
lifetime of the auto loans we study (see Figure~\ref{fig:ltv_by_age}) reveal
credit risk behavior that is non-obvious for collateralized loans (i.e., a
standard in-the-moneyness analysis, such as \citet{campbell_2015}).

We follow the
theoretical results and empirical credit risk convergence point estimates
with a thorough financial analysis of these loans using our techniques for
added precision.  We combine the empirical point estimates of
credit risk convergence with risk band
APRs to estimate month-by-month expected annualized
risk-adjusted returns for lenders.  It assigns
values to the typical back-loading of profits commonly found in high-risk
loans: deep subprime borrowers have early unstable expected risk-adjusted
returns around 5\% before eventually increasing to nearly 15\%.
We then consider a consumer perspective.  Because a borrower's APR
reflects a single point-in-time assessment of credit risk at origination
(i.e., risk-based pricing),
a high-risk borrower that remains current eventually outperforms the
stale APR.  We find borrowers in all risk bands below super-prime delay
a credit-based prepayment in a way that is economically inefficient, all
else equal.  This is a potentially added form of inefficient consumer
refinance behavior, given the interest-rate inefficiency
found in mortgages
\citep[e.g.,][]{keys_2016, agarwal_2017, andersen_2020}.

We estimate prime
and near-prime borrowers leave up to \$2{,}327 and \$2{,}206 in total
potential savings on the table, respectively.  In a surprise, this outpaces
subprime and deep subprime borrowers, who leave up to \$1{,}616 and
\$1{,}153 in total potential savings on the table, respectively. We then
utilize the survival analysis estimators to model a current borrower's
prepayment behavior.  In a visual analysis, we find borrowers' prepayment
behavior is potentially
motivated by increases in used auto values over the
observation period and economic stimulus payments.  This suggests that current
borrowers may not look to refinance into lower rates based on an improving risk
profile.  We then opine that market frictions may exist for both the borrower
and lender that lead to these inefficiencies to persist.

In closing, our theoretical and empirical contributions complement each other
to establish a novel benchmark within consumer auto loans
for future calibration and exploration, as well
as arm researcher's with new financial econometric tools to do so.
For example, our methods may be used to estimate the point of credit risk
convergence in other
forms of fixed-income debt, such as consumer loans outside of autos, and
even more broadly, such as with corporate and sovereign debt.
To make these claims formally, however, more study is needed.  Within consumer
and household finance more specifically, our
credit risk convergence estimates inform a thorough analysis of consumer
automobile loans.  We construct an actuarial approach built on our
probabilistic estimates to estimate an expected
risk-adjusted return for lenders, which recovers an expected back-loading of
profits. In shifting our perspective to the consumer, we further apply our
novel methods to estimate borrowers
delay credit-based prepayment in a manner that is
economically inefficient.  We suspect these household finance results will
similarly extend to other forms of consumer debt, such as residential
mortgages.

\appendix

\begin{center}
{\Large \bf Appendix}
\end{center}

\renewcommand{\thesection}{\Alph{section}}%
\renewcommand{\thesubsection}{\Alph{subsection}}%
\renewcommand{\thesubsubsection}{\Alph{subsection}.\arabic{subsubsection}.}%

\counterwithin{figure}{section}
\counterwithin{table}{section}
\renewcommand\thefigure{\thesection\arabic{figure}}
\renewcommand\thetable{\thesection\arabic{table}}

\section{Asymptotic Properties}
\label{subsec:csh_props}

The vector of CSH estimators using \eqref{eq:csh_est} for 
$\Delta + 1 \leq x \leq \xi$ has convenient asymptotic properties, which
we now summarize.

\begin{proposition}[$\hat{\bm{\Lambda}}^{0i}_{\tau,n}$ \textit{Asymptotic Properties}] 
\label{thm:asymH}
For $i \in \{1,2\}$, define 
$\hat{\bm{\Lambda}}^{0i}_{\tau,n} = 
(\hat{\lambda}^{0i}_{\tau,n}(\Delta + 1),
\ldots, 
\hat{\lambda}^{0i}_{\tau,n}(\xi))^\top$, where
\begin{equation*}
    \hat{\lambda}^{0i}_{\tau,n}(x)
    = \frac{ \hat{f}^{0i}_{*, \tau, n}(x) }{ \hat{U}_{\tau,n}(x) }
    = \frac{
    \sum_{j=1}^{n} 
    \mathbf{1}_{X_j \leq C_j} 
    \mathbf{1}_{Z_{X_j}=i} 
    \mathbf{1}_{\min(X_j,C_j)=x}
    }
    {
    \sum_{j=1}^{n} 
    \mathbf{1}_{Y_j \leq x \leq \min(X_j, C_j)}
    }.
\end{equation*}
Then,
\begin{enumerate}[(i)]
\item
\begin{equation*}
    \hat{\bm{\Lambda}}^{0i}_{\tau,n} 
    \overset{\mathcal{P}}{\longrightarrow}
    \bm{\Lambda}_{\tau}^{0i},
    \text{ as } n \rightarrow \infty;
\end{equation*}
\item
\begin{equation*}
    \sqrt{n}( \hat{\bm{\Lambda}}^{0i}_{\tau,n} - \bm{\Lambda}_{\tau}^{0i})
    \overset{\mathcal{L}}{\longrightarrow}
    N(\bm{0}, \bm{\Sigma}^{0i}), \text{ as } n \rightarrow \infty,
\end{equation*}
\end{enumerate}
where $\bm{\Lambda}_{\tau}^{0i} = 
(\lambda^{0i}_{\tau}(\Delta+1), 
\ldots, 
\lambda^{0i}_{\tau}(\xi))^\top$ with
$\lambda^{0i}_{\tau}(x) = f^{0i}_{*, \tau}(x) / U_{\tau}(x)$ and
\begin{equation*}
    \bm{\Sigma}^{0i} = \textup{diag} \bigg(
    \frac{ 
    f_{*,\tau}^{0i}(\Delta + 1) 
    \{ U_{\tau}(\Delta+1) - f_{*,\tau}^{0i}(\Delta+1)\}
    }
    {U_{\tau}(\Delta+1)^3}, 
    \ldots,
    \frac{ f_{*,\tau}^{0i}(\xi) \{ U_{\tau}(\xi) - f_{*,\tau}^{0i}(\xi)\}}
    {U_{\tau}(\xi)^3}
    \bigg).
\end{equation*}
That is, the cause-specific hazard rate estimators 
$\hat{\lambda}_{\tau,n}^{0i}(\Delta + 1), 
\ldots, 
\hat{\lambda}_{\tau,n}^{0i}(\xi)$ are 
asymptotically unbiased, asymptotically multivariate normal, and
asymptotically independent.
\end{proposition}
\begin{proof}
See the Online Appendix~\ref{subsec:proof2}.
\end{proof}

Often, it is of interest to construct confidence intervals for the 
cause-specific hazard rate estimators such that the confidence intervals have
the desirable property of falling within the interval $(0,1)$.  We may do so as
follows.
\begin{Lemma}[$\lambda_{\tau}^{0i}(x)$ $(1-\theta)$\% \textit{Confidence Interval}]
\label{cor:haz_ci}
The $(1-\theta)$\% asymptotic confidence interval bounded within $(0,1)$ for
$\lambda_{\tau}^{0i}(x)$, $x \in \{\Delta+1, \ldots, \xi\}$, $i = 1, 2$ is
\begin{equation}
    \exp \bigg{ \{ }
    \ln \hat{\lambda}_{\tau, n}^{0i}(x) \pm
    \mathcal{Z}_{(1-\theta/2)}
    \sqrt{
    \frac{\hat{U}_{\tau, n}(x) - \hat{f}^{0i}_{*, \tau, n}(x)}
    {n \hat{U}_{\tau, n}(x) \hat{f}^{0i}_{*, \tau, n}(x) }
    }
    \bigg{ \} },
    \label{eq:haz_ci_cor}
\end{equation}
where $\mathcal{Z}_{(1-\theta/2)}$ represents the $(1-\theta/2)$th percentile 
of the standard normal distribution.
\end{Lemma}
\begin{proof}
See the  Online Appendix~\ref{subsec:proof2}.
\end{proof}

\section{Section~\ref{subsec:emp_res}: Additional Details}
\label{subsec:emp_add_det}

The purpose of this section is to provide additional details related to
Section~\ref{subsec:emp_res}.
We plot the full five-by-five matrix of CSH rate estimates for default in 
Figure~\ref{fig:haz_grid_default} for the sample of 58{,}118 loans issued in
2017.  It is a complete extension of the
subprime versus prime plot in Figure~\ref{fig:conv_demo}.
That is, Figure~\ref{fig:conv_demo} is a zoomed-in view of the subprime-prime
cell (row 4, column 2) in Figure~\ref{fig:haz_grid_default}.
There are a few consistent observations.  
First, we generally see that the monthly conditional default rate declines
as the credit quality of the risk band improves, as expected.  We again see
a large increase in the hazard rate for the deep subprime, subprime, and near
prime risk bands around loan age 40.
With some approximate date arithmetic from the first payment
month of the ABS bonds (March-April-May 2017), we find that a loan age of 40
months corresponds to roughly Spring 2020 (when adjusted for left-truncation).
This corresponds to the economic impact of the Coronavirus pandemic,
which effectively stopped most economic activity in Spring 2020.
It is interesting that the economic shutdown of
the Coronavirus appears to have had minimal impact on the prime risk band and 
almost no notable impact on the super-prime risk band.
By comparing the asymptotic confidence intervals within each risk band
comparison by loan age,
we may estimate the point of credit risk convergence.
Figure~\ref{fig:haz_grid_default} may be compared with the matrix in the top
row of Table~\ref{tab:risk_conv_mat}.

\begin{figure}[t!]
    \centering
    \includegraphics[width=\textwidth]{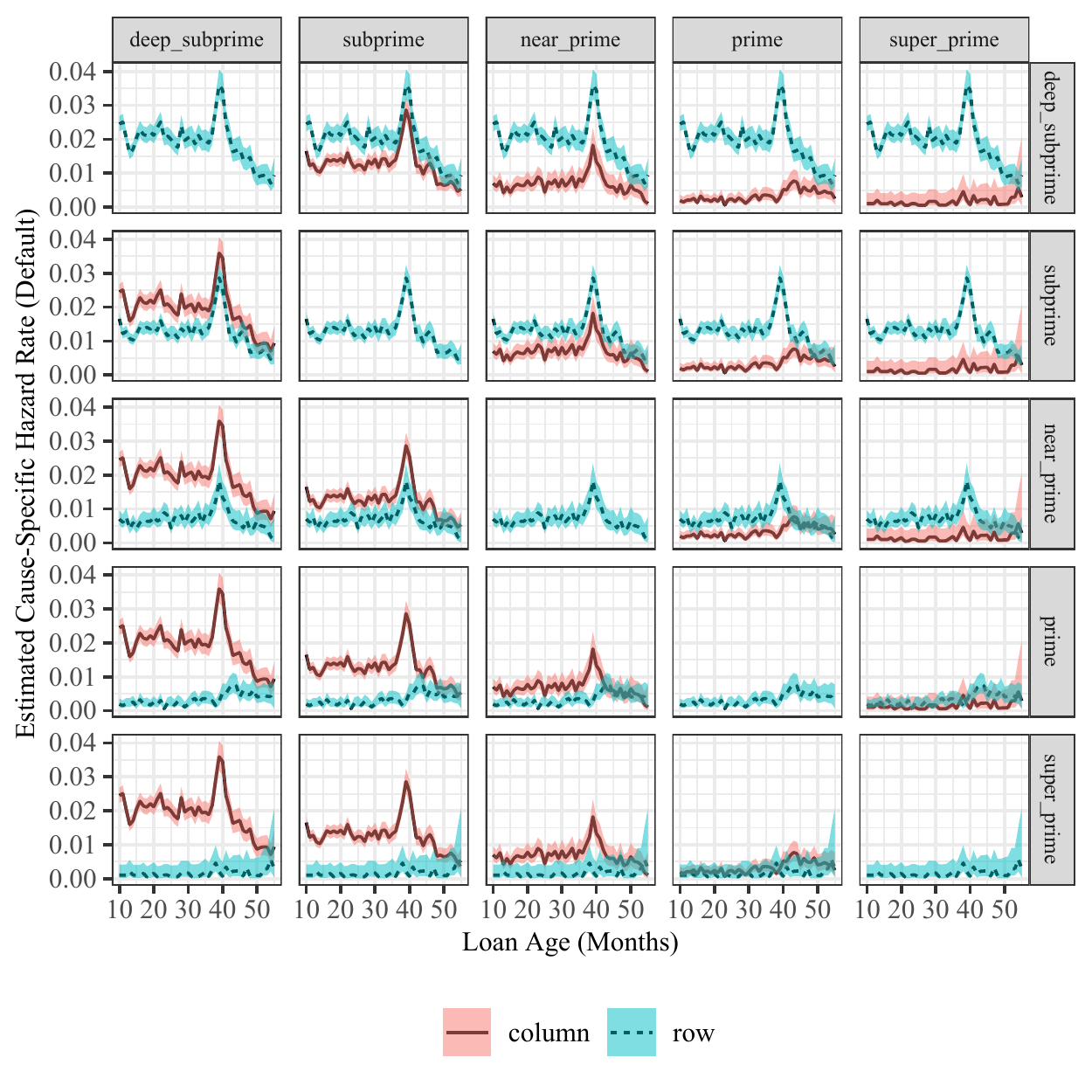}
    \caption{
\footnotesize{  
\textbf{Credit Risk Convergence: All Risk Bands (2017).}\newline   
A plot of $\hat{\lambda}_{\tau,n}^{01}$ (defaults) defined in
\eqref{eq:csh_est} by loan age for all five risk bands
within the sample of 58,118 loans
(Table~\ref{tab:bond_summary}), plus 95\% confidence intervals using
Lemma~\ref{cor:haz_ci}.  We may use the hypothesis test described in
\eqref{eq:H0} by searching for the minimum age that the confidence intervals
overlap between two disparate risk bands.  The
large upward spike in $\hat{\lambda}_{\tau,n}^{01}$ for the deep subprime,
subprime, and near-prime risk bands
around loan age 40 is related to the economic impact of COVID-19, a point
discussed more fully in Section~\ref{subsec:COVID}.
}    
}
    \label{fig:haz_grid_default}
\end{figure}

\section{Section~\ref{subsec:COVID}: Additional Details}
\label{subsec:covid_detail}

The purpose of this section is to provide additional details related to
Section~\ref{subsec:COVID}.
We plot the full five-by-five matrix of CSH rate estimates for default in
Figure~\ref{fig:haz_grid_default2019} for the sample of 65{,}802 loans issued
in 2019.  It is a complete extension of the subprime versus prime plot in
Figure~\ref{fig:COVID_demo}.  That is, Figure~\ref{fig:COVID_demo} is a
zoomed-in view of the subprime-prime cell (row 4, column 2) in
Figure~\ref{fig:haz_grid_default2019}.  The purpose of considering the 2019
issuance is for the sensitivity testing related to COVID-19 (see
Section~\ref{subsec:COVID}).
If the timing of credit risk convergence is completely driven by the Spring
2020 economic shutdown, we would expect to see it occur much earlier in the
2019 sample of bonds when subject to the same loan selection process and risk
band definitions of Section~\ref{subsec:loan_filter}.

As expected, we see the large spike in the
cause-specific hazard rate for defaults around loan age 10, which, when 
adjusted for left-truncation, corresponds to the Spring 2020 economic shutdown.
It occurs much sooner in the comparison to the 2017 issuance.  Overall,
we see some evidence of earlier convergence in superior risk bands, and so
the shock of the economic shutdown of Spring 2020 has played some role.
It is not consistent throughout the 2019 transition matrix in
Table~\ref{tab:risk_conv_mat}, however, and so the timing of
credit risk convergence is not solely a product of COVID-19.
In other words, loan age, in addition to the economic shutdown of Spring 2020,
plays a role in the point estimates of credit risk convergence.

\begin{figure}[t!]
    \centering
    \includegraphics[width=\textwidth]{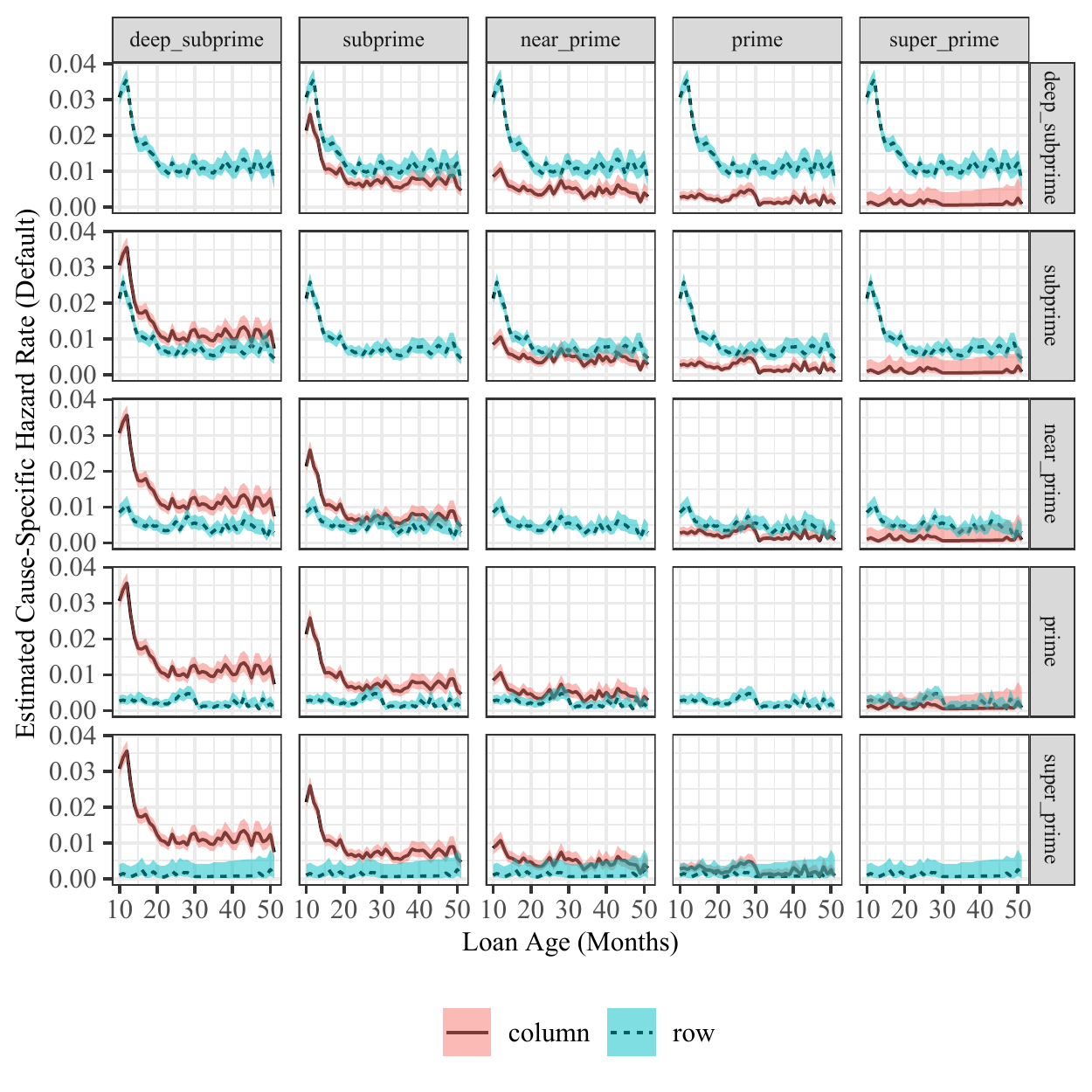}
\caption{
\footnotesize{
\textbf{Credit Risk Convergence: All Risk Bands (2019).}\newline   
A plot of $\hat{\lambda}_{\tau,n}^{01}$ (defaults) defined in
\eqref{eq:csh_est} by loan age for all five risk bands
within the sample of 65{,}802 loans
(Section~\ref{subsec:COVID}), plus 95\% confidence intervals using
Lemma~\ref{cor:haz_ci}.  It is a repeat of Figure~\ref{fig:haz_grid_default}
for the 2019 issuance as a sensitivity check that the economic shock of
COVID-19 is not the sole reason for the estimated timing of
credit risk convergence between disparate
risk bands.
}
}
\phantomsection    \label{fig:haz_grid_default2019}
\end{figure}

\section{Estimating Recovery Upon Default}
\label{subsec:recov}

Consumer auto loans are secured with the collateral of the attached automobile.
In the event of a defaulted loan, the lender has legal standing to repossess 
the vehicle to make up the outstanding balance of the loan.  In most cases, 
particularly for deep subprime and subprime borrowers, the estimated value of a
repossessed automobile in the event of default is an important component in the
initial pricing of a loan.  In this section, therefore, we briefly discuss our
process to estimate a recovery assumption by loan age, which is ultimately
defined as a percentage of the initial loan balance.  Our estimates are
used in the analysis of Section~\ref{subsec:lend_prof}, but we acknowledge the
empirical results may also be of interest to readers more generally.  We thus
present our estimated recovery curve for the 2017 issuance (see
Section~\ref{subsec:summary}) in Figure~\ref{fig:recovery_est}.

The results of Figure~\ref{fig:recovery_est} utilize the detailed reporting of
the loan level data of \citet{cfr_229} to perform the estimation for both the
filtered sample of 58{,}118 loans issued in 2017 and summarized in
Section~\ref{subsec:summary} and the filtered sample of 65{,}802 loans issued
in 2019 and summarized in Section~\ref{subsec:COVID}.
Specifically, we calculate a
sum total of the \texttt{recoveredAmount} field for all loans that ended in 
default.  The \texttt{recoveredAmount} field includes any additional loan
payments made by the borrower after defaulting, legal settlements, and 
repossession proceeds \citep{cfr_229}.  We then divide the total
\texttt{recoveredAmount} by the \texttt{originalLoanAmount} for each defaulted
loan.  Finally, we take an average of these recovery percentages by age of
default in months.  The point estimates may be found in
Figure~\ref{fig:recovery_est}.  Next, for convenient use within the
lender profitability analysis of Section~\ref{subsec:lend_prof},
we nonparametrically smooth the
point estimates using the \texttt{loess()} function in \textsf{R}
\citep{r_citation}.  See the dashed line in Figure~\ref{fig:recovery_est}.
This nonparametric \texttt{loess} curve is then fitted to a gamma-kernel via
ordinary minimization of a sum-of-squared differences, which allows for
extrapolation beyond the recoverable sample space.  See the solid line in
Figure~\ref{fig:recovery_est}.

\begin{figure}[t!]
    \centering
    \includegraphics[width=\textwidth]{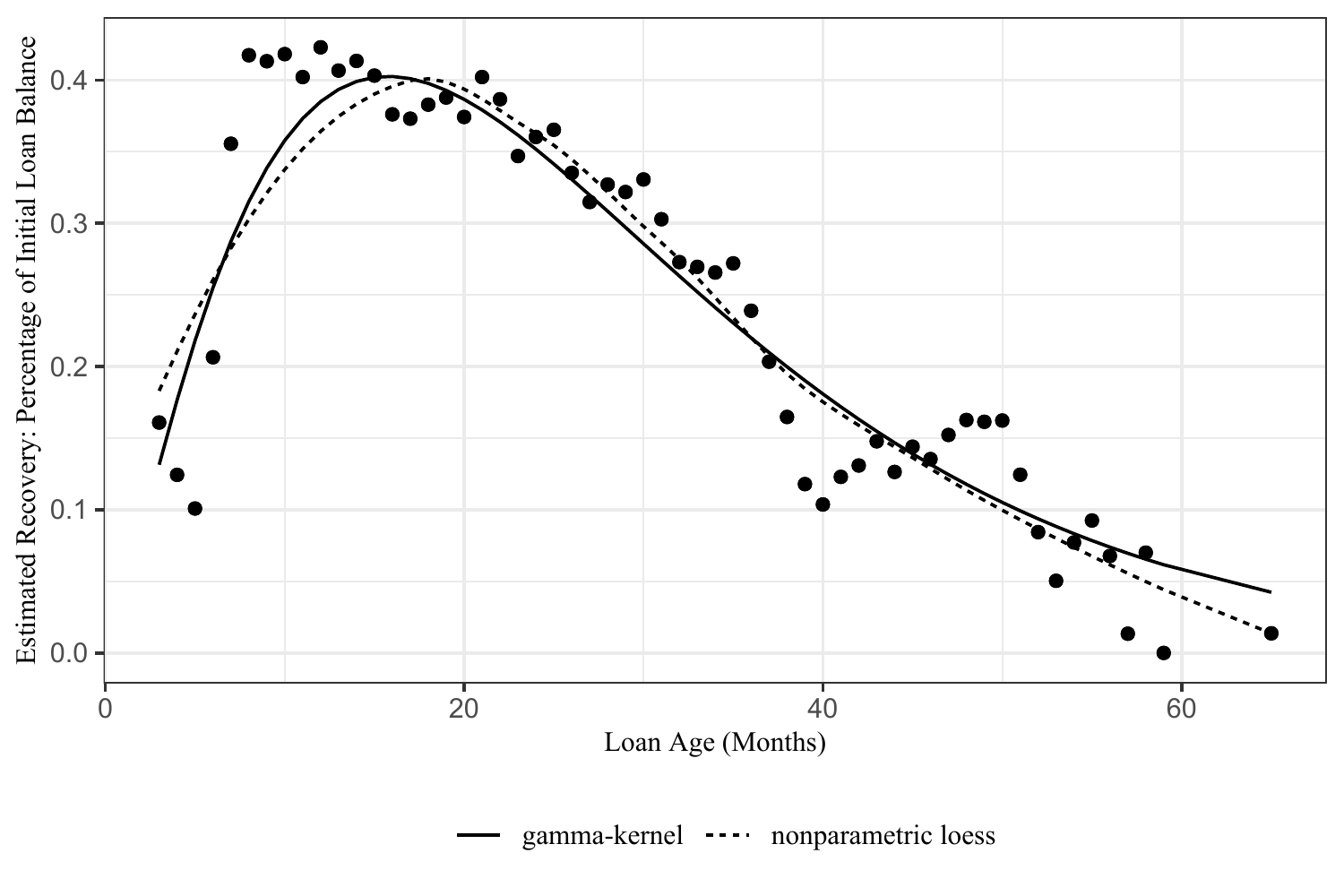}
    \caption{
\footnotesize{
\textbf{Estimation of the Recovery Upon Default Assumption.}\newline
The point estimates are formed using the asset-level data of \citet{cfr_229} 
for the 58{,}118 filtered loans summarized in Section~\ref{subsec:summary}.
Specifically, they are the monthly average of the
sum total of the \texttt{recoveredAmount} field, which
includes any additional loan payments made by the borrower after defaulting,
legal settlements, and repossession proceeds \citep{cfr_229}, divided by
the \texttt{originalLoanAmount} field for each loan that ended in default.
Smoothing techniques are also presented.  The shape of the
recovery curve is similar for the sample of 65{,}802 loans issued in 2019. 
}
}
\phantomsection    \label{fig:recovery_est}
\end{figure}

The shape of the recovery curve warrants some commentary.  Loans that default
shortly after origination generally have a low recovery amount as a percentage
of the initial loan balance, between 10-20\%.  This is likely because a loan 
that defaults so quickly after origination may be due to fraud in the 
initial loan application, extreme circumstances for the borrower (i.e., 
rapid decline in physical health), or severe damage to the 
vehicle.  In the case of damage to the vehicle, it is possible the borrower 
has also lapsed on auto insurance or removed collision insurance.
Overall, it can be difficult to recover a meaningful amount in these
circumstances.  The recovery percentage then peaks 
at month 12 at just over 42\% before declining towards zero as the loan age
approaches termination (72--73 months).  Since all vehicles in our sample are
used, the decline in recoveries reflects the typical depreciating value of the 
automobile over time \citep[e.g.,][]{storchmann_2004}.

We close this section by noting
the economic welfare of an automobile repossession has attracted the attention
of researchers.  Generally, the results are mixed.
On the one hand, \cite{pollard_2021} discuss a vicious cycle of subprime auto 
lending where the same car may be bought, sold, and repossessed 20-30 times.
This suggests repossessions may negatively impact economic welfare. A earlier
result by \cite{cohen_1998} finds that manufacturers prefer to offer 
prospective borrowers interest discounts over equivalent cash rebates because a
legal technicality finds such a discount is financially beneficial to the 
lender in the event of repossession.  In this case, the legal circumstances of
a repossession may influence market behavior.  Along the same lines and an
argument for the potential economic benefits of repossession, 
\cite{assuncao_2013} find that a 2004 credit reform in Brazil, which simplified
the sale of repossessed cars, lead to an expansion of credit for riskier, 
self-employed borrowers.  In other words, a reform designed to make recouping 
money from a repossessed automobile easier for lenders improved the ability of 
riskier borrowers to access credit.  It is noteworthy, however, that the reform
also lead to increased incidences of delinquencies and default.

{\singlespacing \small
\bibliographystyle{jf}
\bibliography{rfs}
}

\clearpage

\begin{center}
{\Large \bf Internet Appendix}
\end{center}

\renewcommand{\thesection}{\Alph{section}}%
\renewcommand{\thesubsection}{\Alph{subsection}}%
\renewcommand{\thesubsubsection}{\Alph{subsection}.\arabic{subsubsection}.}%

\counterwithin{figure}{section}
\counterwithin{table}{section}
\renewcommand\thefigure{\thesection\arabic{figure}}
\renewcommand\thetable{\thesection\arabic{table}}

\section{Proofs: Section~\ref{sec:CRC}}
\label{subsec:proof2}

\begin{proof}[Proof of Proposition \ref{thm:asymH}.]
Statement \textit{(i)} follows from \textit{(ii)}, so it is enough to show \textit{(ii)}. Let
$\Delta + 1 \leq k \leq \xi$ and observe
\begin{align*}
    &\hat{\lambda}_{\tau,n}^{0i}(k) - \lambda_{\tau}^{0i}(k)\\
    &= \frac{
    \frac{1}{n} \sum_{j=1}^{n} 
    \mathbf{1}_{X_j \leq C_j} 
    \mathbf{1}_{Z_{X_j}=i} 
    \mathbf{1}_{\min(X_j,C_j)=k}
    }
    {
    \hat{U}_{\tau,n}(k)
    }
    - \frac{ f_{*,\tau}^{0i}(k) }{ U_{\tau}(k) }\\
    &= \frac{
    \{ \sum_{j=1}^{n} 
    \mathbf{1}_{X_j \leq C_j} 
    \mathbf{1}_{Z_{X_j}=i} 
    \mathbf{1}_{\min(X_j,C_j)=k} \}
    U_{\tau}(k) - f_{*,\tau}^{0i}(k) \hat{U}_{\tau,n}(k)
    }
    {\hat{U}_{\tau,n}(k) U_{\tau}(k)}\\
    &= \bigg[ \frac{1}{\hat{U}_{\tau,n}(k) U_{\tau}(k)} \bigg]
    \frac{1}{n} \sum_{j=1}^{n} \{
    \mathbf{1}_{X_j \leq C_j} 
    \mathbf{1}_{Z_{X_j}=i} 
    \mathbf{1}_{\min(X_j,C_j)=k}U_{\tau}(k) - 
    f_{*,\tau}^{0i}(k) \mathbf{1}_{Y_j \leq k \leq \min(X_j,C_j)}
    \}.
\end{align*}
Define
\begin{equation*}
    H^{0i}_{\tau,k(j)} = 
    \mathbf{1}_{X_j \leq C_j} 
    \mathbf{1}_{Z_{X_j}=i} 
    \mathbf{1}_{\min(X_j,C_j)=k} U_{\tau}(k) - 
    f_{*,\tau}^{0i}(k) \mathbf{1}_{Y_j \leq k \leq \min(X_j,C_j)},
\end{equation*}
for $1 \leq j \leq n$ and
\begin{equation*}
    \mathbf{A}_{\tau,n} = 
    \text{diag}([\hat{U}_{\tau,n}(\Delta+1) U_{\tau}(\Delta+1)]^{-1},
    \ldots,
    [\hat{U}_{\tau,n}(\xi) U_{\tau}(\xi)]^{-1}).
\end{equation*}
Then,
\begin{equation*}
    \hat{\bm{\Lambda}}_{\tau,n}^{0i} - \bm{\Lambda}^{0i}_{\tau} =
    \mathbf{A}_{\tau,n} \frac{1}{n} \sum_{j=1}^{n}
    \begin{bmatrix}
    H^{0i}_{\tau,\Delta+1(j)}\\
    \vdots\\
    H^{0i}_{\tau,\xi(j)}
    \end{bmatrix},
\end{equation*}
or, letting $\mathbf{H}^{0i}_{\tau,(j)} = 
(H^{0i}_{\tau,\Delta+1(j)}, \ldots, H^{0i}_{\tau,\xi(j)})^\top$ denote independent
and identically distributed random vectors, we have compactly
\begin{equation*}
    \hat{\bm{\Lambda}}_{\tau,n}^{0i} - \bm{\Lambda}^{0i}_{\tau} =
    \mathbf{A}_{\tau,n} \frac{1}{n} \sum_{j=1}^{n} \mathbf{H}^{0i}_{\tau,(j)}.
\end{equation*}
It is noteworthy the components of $\mathbf{H}^{0i}_{\tau,(j)}$ are 
uncorrelated. More specifically,
\begin{equation}
    \text{Cov}[H^{0i}_{\tau,k(j)}, H^{0i}_{\tau,k'(j)}]
    = 
    \begin{cases}
    U_{\tau}(k)f^{0i}_{*,\tau}(k)[U_{\tau}(k) - f^{0i}_{*,\tau}(k)], & k = k'\\
    0, & k \neq k'.
    \end{cases}
    \label{eq:cov_H}
\end{equation}
To see this, first notice the indicator functions 
$\mathbf{1}_{X_j \leq C_j} 
\mathbf{1}_{Z_{X_j}=i} 
\mathbf{1}_{\min(X_j,C_j)=k}$ and
$\mathbf{1}_{Y_j \leq k \leq \min(X_j,C_j)}$ are Bernoulli random variables 
with probability parameters $f^{0i}_{*,\tau}(k)$ and $U_{\tau}(k)$, 
respectively.  Hence,
\begin{align*}
    \mathbf{E} H^{0i}_{\tau,k(j)} &=
    \mathbf{E} 
    \mathbf{1}_{X_j \leq C_j} 
    \mathbf{1}_{Z_{X_j}=i} 
    \mathbf{1}_{\min(X_j,C_j)=k}
    U_{\tau}(k) -
    f^{0i}_{*,\tau}(k)
    \mathbf{E}
    \mathbf{1}_{Y_j \leq k \leq \min(X_j,C_j)}\\
    &= f^{0i}_{*,\tau}(k) U_{\tau}(k) - f^{0i}_{*,\tau}(k) U_{\tau}(k)\\
    &=0.
\end{align*}
Therefore,
\begin{align*}
    &\text{Cov}[H^{0i}_{\tau,k(j)}, H^{0i}_{\tau,k'(j)}]\\
    ={}&
    \mathbf{E} H^{0i}_{\tau,k(j)} H^{0i}_{\tau,k'(j)}\\
    ={}& \mathbf{E} \{ 
    \mathbf{1}_{X_j \leq C_j} \mathbf{1}_{Z_{X_j}=i} 
    \mathbf{1}_{\min(X_j,C_j)=k} U_{\tau}(k) - 
    f_{*,\tau}^{0i}(k) \mathbf{1}_{Y_j \leq k \leq \min(X_j,C_j)}
    \}\\
    & \times
    \{
    \mathbf{1}_{X_j \leq C_j} 
    \mathbf{1}_{Z_{X_j}=i} 
    \mathbf{1}_{\min(X_j,C_j)=k'} U_{\tau}(k') - 
    f_{*,\tau}^{0i}(k') \mathbf{1}_{Y_j \leq k' \leq \min(X_j,C_j)}
    \}\\
    ={}& 
    U_{\tau}(k)U_{\tau}(k') 
    \mathbf{E} 
    \mathbf{1}_{X_j \leq C_j} \mathbf{1}_{Z_{X_j}=i} 
    \mathbf{1}_{\min(X_j,C_j)=k}
    \mathbf{1}_{X_j \leq C_j} \mathbf{1}_{Z_{X_j}=i} 
    \mathbf{1}_{\min(X_j,C_j)=k'}\\
    &- U_{\tau}(k) f^{0i}_{*,\tau}(k') \mathbf{E} 
    \mathbf{1}_{X_j \leq C_j} \mathbf{1}_{Z_{X_j}=i} 
    \mathbf{1}_{\min(X_j,C_j)=k}
    \mathbf{1}_{Y_j \leq k' \leq \min(X_j,C_j)}\\
    &- U_{\tau}(k') f^{0i}_{*,\tau}(k) \mathbf{E} 
    \mathbf{1}_{X_j \leq C_j} \mathbf{1}_{Z_{X_j}=i} 
    \mathbf{1}_{\min(X_j,C_j)=k'}
    \mathbf{1}_{Y_j \leq k \leq \min(X_j,C_j)}\\
    &+ f^{0i}_{*,\tau}(k) f^{0i}_{*,\tau}(k')
    \mathbf{E} \mathbf{1}_{Y_j \leq k \leq \min(X_j,C_j)} 
    \mathbf{1}_{Y_j \leq k' \leq \min(X_j,C_j)}.
\end{align*}
We proceed by cases.

Case 1: $k = k'$.

Working through each expectation in 
$\text{Cov}[H^{0i}_{\tau,k(j)}, H^{0i}_{\tau,k'(j)}]$,
we have
\begin{align*}
    &\mathbf{E} 
    \mathbf{1}_{X_j \leq C_j} \mathbf{1}_{Z_{X_j}=i} 
    \mathbf{1}_{\min(X_j,C_j)=k}
    \mathbf{1}_{X_j \leq C_j} \mathbf{1}_{Z_{X_j}=i} 
    \mathbf{1}_{\min(X_j,C_j)=k'}\\
    ={}&
    \mathbf{E} 
    \mathbf{1}_{X_j \leq C_j} \mathbf{1}_{Z_{X_j}=i} 
    \mathbf{1}_{\min(X_j,C_j)=k}\\
    ={}& f^{0i}_{*,\tau}(k),
\end{align*}
\begin{align*}
    &\mathbf{E} 
    \mathbf{1}_{X_j \leq C_j} \mathbf{1}_{Z_{X_j}=i} 
    \mathbf{1}_{\min(X_j,C_j)=k}
    \mathbf{1}_{Y_j \leq k' \leq \min(X_j,C_j)}\\
    ={}& \mathbf{E} 
    \mathbf{1}_{X_j \leq C_j} \mathbf{1}_{Z_{X_j}=i} 
    \mathbf{1}_{\min(X_j,C_j)=k'}
    \mathbf{1}_{Y_j \leq k \leq \min(X_j,C_j)}\\
    ={}& \mathbf{E}
    \mathbf{1}_{X_j \leq C_j} \mathbf{1}_{Z_{X_j}=i} 
    \mathbf{1}_{\min(X_j,C_j)=k}
    \mathbf{1}_{Y_j \leq k \leq \min(X_j,C_j)}\\
    ={}& \mathbf{E}
    \mathbf{1}_{X_j \leq C_j} \mathbf{1}_{Z_{X_j}=i} 
    \mathbf{1}_{\min(X_j,C_j)=k}\\
    ={}& f^{0i}_{*,\tau}(k),
\end{align*}
and
\begin{equation*}
    \mathbf{E} \mathbf{1}_{Y_j \leq k \leq \min(X_j,C_j)} 
    \mathbf{1}_{Y_j \leq k' \leq \min(X_j,C_j)}
    =
    \mathbf{E} \mathbf{1}_{Y_j \leq k \leq \min(X_j,C_j)}
    =
    U_{\tau}(k).
\end{equation*}
Thus,
\begin{equation*}
    \text{Cov}[H^{0i}_{\tau,k(j)}, H^{0i}_{\tau,k'(j)}]
    =
    U_{\tau}(k) f^{0i}_{*, \tau}(k) [
    U_{\tau}(k) - f^{0i}_{*, \tau}(k)
    ].
\end{equation*}

Case 2: $k \neq k'$.

Working through each expectation in 
$\text{Cov}[H^{0i}_{\tau,k(j)}, H^{0i}_{\tau,k'(j)}]$,
we have
\begin{equation*}
    \mathbf{E} 
    \mathbf{1}_{X_j \leq C_j} \mathbf{1}_{Z_{X_j}=i} 
    \mathbf{1}_{\min(X_j,C_j)=k}
    \mathbf{1}_{X_j \leq C_j} \mathbf{1}_{Z_{X_j}=i} 
    \mathbf{1}_{\min(X_j,C_j)=k'}
    = 0,
\end{equation*}
\begin{align*}
    \mathbf{E} &
    \mathbf{1}_{X_j \leq C_j} \mathbf{1}_{Z_{X_j}=i} 
    \mathbf{1}_{\min(X_j,C_j)=k}
    \mathbf{1}_{Y_j \leq k' \leq \min(X_j,C_j)}\\
    &=
    \begin{cases}
    \Pr(X_j \leq C_j, Z_{X_j} = i, \min(X_j,C_j) = k, Y_j \leq k'),
    & k > k'\\
    0, & k < k',
    \end{cases}
\end{align*}
\begin{align*}
    \mathbf{E} &
    \mathbf{1}_{X_j \leq C_j} \mathbf{1}_{Z_{X_j}=i} 
    \mathbf{1}_{\min(X_j,C_j)=k'}
    \mathbf{1}_{Y_j \leq k \leq \min(X_j,C_j)}\\
    &=
    \begin{cases}
    0, & k > k'\\
    \Pr(X_j \leq C_j, Z_{X_j} = i, \min(X_j,C_j) = k', Y_j \leq k),
    & k < k',
    \end{cases}
\end{align*}
and
\begin{align*}
    &\mathbf{E} \mathbf{1}_{Y_j \leq k \leq \min(X_j,C_j)} 
    \mathbf{1}_{Y_j \leq k' \leq \min(X_j,C_j)}\\
    ={}&
    \Pr(Y_j \leq k \leq \min(X_j, C_j), Y_j \leq k' \leq \min(X_j,C_j)).
\end{align*}
Thus,
\begin{align*}
    \text{Cov}&[H^{0i}_{\tau,k(j)}, H^{0i}_{\tau,k'(j)}]
    =f^{0i}_{*,\tau}(\min(k,k')) \bigg{\{}\\
    &-U_{\tau}(\max(k,k'))
    \Pr(X_j \leq C_j, Z_{X_j} = i, 
    \min(X_j, C_j) = \max(k,k'), 
    Y_j \leq \min(k,k'))\\
    &+ f^{0i}_{*,\tau}(\max(k,k'))
    \Pr(Y_j \leq k \leq \min(X_j, C_j), Y_j \leq k' \leq \min(X_j,C_j))
    \bigg{ \} }.
\end{align*}
However, because of the independence between $Y$ and $(X, Z_X)$,
\begin{align*}
    U_{\tau}(\max(k,k')) 
    &= \Pr(Y_j \leq \max(k,k') \leq \min(X_j,C_j))\\
    &= \Pr(Y \leq \max(k,k'), 
    X \geq \max(k,k'), 
    C \geq \max(k,k') \mid Y \leq X)\\
    &= 
    \{ \Pr(Y \leq \max(k,k') \leq C) \Pr(X \geq \max(k,k')) \} / \alpha,
\end{align*}
\begin{align*}
    \Pr & (X_j \leq C_j, Z_{X_j} = i, 
    \min(X_j, C_j) = \max(k,k'), 
    Y_j \leq \min(k,k'))\\
    &= \Pr(C \geq \max(k,k'), 
    Z_X = i, X = \max(k,k'), 
    Y \leq \min(k,k') \mid Y \leq X)\\
    &= \{ \Pr(X = \max(k,k'), 
    Z_X = i) \Pr(Y \leq \min(k,k'), 
    C \geq \max(k,k')) \}
    / \alpha,
\end{align*}
\begin{align*}
    f^{0i}_{*,\tau}(\max(k,k'))
    &= \Pr(X = \max(k,k'), C \geq \max(k,k'), Z_x = i \mid Y \leq X)\\
    &= \{ \Pr(X = \max(k,k'), Z_X = i) \Pr(Y \leq \max(k,k') \leq C) \} 
    / \alpha,
\end{align*}
and
\begin{align*}
    \Pr & (Y_j \leq k \leq \min(X_j, 
    C_j), Y_j \leq k' \leq \min(X_j,C_j))\\
    &= \Pr(Y \leq \min(k,k'), 
    C \geq \max(k,k'), X \geq \max(k,k') \mid Y \leq X)\\
    &= \{ \Pr(Y \leq \min(k,k'), 
    C \geq \max(k,k')) \Pr(X \geq \max(k,k')) \} / \alpha
\end{align*}
Therefore, 
\begin{align*}
    U_{\tau} & (\max(k,k'))
    \Pr(X_j \leq C_j, 
    Z_{X_j} = i, \min(X_j, C_j) = \max(k,k'), Y_j \leq \min(k,k'))\\
    &=f^{0i}_{*,\tau}(\max(k,k'))
    \Pr(Y_j \leq k \leq \min(X_j, C_j), Y_j \leq k' \leq \min(X_j,C_j)),
\end{align*}
and so $\text{Cov}[H^{0i}_{\tau,k(j)}, H^{0i}_{\tau,k'(j)}] = 0$ when 
$k \neq k'$. This confirms \eqref{eq:cov_H}.  Now define the diagonal matrix
\begin{equation*}
\mathbf{D}^{0i}_{\tau} = \text{diag}
\begin{bmatrix}
U_{\tau}(\Delta+1)f^{0i}_{*,\tau}(\Delta+1)
    [U_{\tau}(\Delta+1) - f^{0i}_{*,\tau}(\Delta+1)]\\
    \vdots\\
    U_{\tau}(\xi)f^{0i}_{*,\tau}(\xi)[U_{\tau}(\xi) - f^{0i}_{*,\tau}(\xi)]
\end{bmatrix}
\end{equation*}
and
\begin{equation*}
    \bar{\mathbf{H}}_{\tau,n}^{0i} 
    = \frac{1}{n} \sum_{j=1}^{n} \mathbf{H}^{0i}_{\tau,(j)}.
\end{equation*}
By the multivariate Central Limit Theorem 
\citep[Theorem 8.21, pg. 61]{lehmann_1998}, therefore,
\begin{equation*}
    \sqrt{n}( \bar{\mathbf{H}}_{\tau,n}^{0i} - \bm{0})
    \overset{\mathcal{L}}{\longrightarrow} N( \bm{0}, \mathbf{D}^{0i}_{\tau}),
    \text{ as } n \rightarrow \infty.
\end{equation*}
Next, define 
$\mathbf{V}_{\tau} = 
\text{diag}(U_{\tau}(\Delta+1)^{-2}, \ldots, U_{\tau}(\xi)^{-2})$.
By \citet[][Lemma 1]{lautier_2023}, 
$\mathbf{A}_{\tau,n} \overset{\mathcal{P}}{\longrightarrow} \mathbf{V}_{\tau}$,
as $n \rightarrow \infty$.  Thus, by the multivariate version of Slutsky's
Theorem \citep[Theorem 5.1.6, pg. 283]{lehmann_1998},
\begin{equation*}
    \sqrt{n}( \mathbf{A}_{\tau,n} \bar{\mathbf{H}}_{\tau,n}^{0i} )
    \overset{ \mathcal{L} }{ \longrightarrow }
    N( 
    \bm{0}, \mathbf{V}_{\tau} \mathbf{D}^{0i}_{\tau} \mathbf{V}_{\tau}^\top
    ),
    \text{ as } n \rightarrow \infty.
\end{equation*}
We may complete the proof by observing 
$\mathbf{V}_{\tau} \mathbf{D}^{0i}_{\tau} \mathbf{V}_{\tau}^\top
= \bm{\Sigma}^{0i}$ and
$\mathbf{A}_{\tau,n} \bar{\mathbf{H}}_{\tau,n}^{0i} 
= \hat{\bm{\Lambda}}^{0i}_{\tau,n} - \bm{\Lambda}_{\tau}^{0i}$.
\end{proof}

\begin{proof}[Proof of Lemma~\ref{cor:haz_ci}]
The classical method dictates first finding a $(1 - \theta)$\% confidence 
interval on a log-scale and then converting back to a standard-scale to ensure
the estimated confidence interval for the hazard rate, which is a probability,
remains in the interval $(0,1)$.  By an application of the Delta Method
\citep[Theorem 8.12, pg. 58]{lehmann_1998}, we have for
$x \in \{\Delta + 1, \ldots, \xi\}$ and $i = 1, 2$,
\begin{equation*}
\sqrt{n} \big( 
\ln \hat{\lambda}_{\tau, n}^{0i}(x) - \ln \lambda^{0i}_{\tau}(x)
\big)
\overset{\mathcal{L}}{\longrightarrow}
N \bigg( 
0,
\frac{f_{*,\tau}^{0i}(x)\{ U_{\tau}(x) - f_{*,\tau}^{0i}(x)\}}
{U_{\tau}(x)^3}
\frac{1}{\lambda^{0i}_{\tau}(x)^2}
\bigg).
\end{equation*}
The result follows from \eqref{eq:csh_fc}, the Continuous Mapping
Theorem \citep[Theorem 5.2.5, pg. 249]{nitis_2000}, the pivotal approach
\citep[\S 9.2.2]{nitis_2000}, and converting back to the standard scale.
\end{proof}

%
%

\section{Large Sample Simulation Study}
\label{sec:sim_study}

We present a simulation study in support of Proposition~\ref{thm:asymH} and
Lemma~\ref{cor:haz_ci}.  Let the true distribution for the lifetime random
variable $X$ and bivariate distribution of $(X, Z_X)$ be as in Table
\ref{tab:X_probs}.  The column $p(x)$ denotes the probability of event type 1
given an event at time $X$. This allows us to populate the joint distribution 
for $\Pr(X=x,Z_X=i)$ for $i = 1, 2$.  The cause-specific hazard rates then 
follow from \eqref{eq:csh}, and we also report the all-cause hazard rate in the
final column. Notice that, for each $x$,
\begin{equation*}
    p(x) = \frac{ \lambda^{01}(x) }{\lambda^{01}(x) + \lambda^{02}(x) }.
\end{equation*}

For the truncation random variable, we assume $Y$ is discrete uniform with
sample space $\mathcal{Y} \in \{1, 2, 3, 4, 5\}$.  This results in 
$\alpha = 0.864$.  For the purposes of the simulation, we further assume 
$\tau = 5$.  We use the simulation procedure of \citet{beyersmann_2009} but
modified for random truncation.  Specifically,
\begin{enumerate}
    \item Simulate the truncation time, $Y$.
    \item Set the censoring time to be $Y + \tau$.
    \item Simulate the event time, $X$.
    \item Simulate a Bernoulli event with probability $p(x)$ to determine if
    the event $X$ was caused by type 1 with probability $p(x)$ or type 2 with 
    probability $1-p(x)$.
\end{enumerate}

We simulated $n = 10{,}000$ lifetimes using the above algorithm.  We then 
tossed any observations that were truncated (i.e., $Y_j > X_j$, for 
$j = 1, \ldots, n$).  This left a sample of competing risk events subject to 
censoring, which would be the same incomplete data conditions as a trust of 
securitized loans.  We then used the results of Section 
\ref{subsec:stat_results} to
estimate $\hat{f}^{0i}_{*, \tau, n}(x)$, $\hat{U}_{\tau, n}(x)$, and 
$\hat{\lambda}^{0i}_{\tau,n}(x)$ for $i = 1, 2$ and $x \in \{1, \ldots, 10\}$
over $r = 1{,}000$ replicates.

To validate the asymptotic results of Proposition~\ref{thm:asymH}, we compare the
empirical covariance matrix against the derived asymptotic covariance matrix, 
$\bm{\Sigma}^{0i}$, by examining estimates of the confidence intervals using
Lemma \ref{cor:haz_ci}.  Figure \ref{fig:sim_study} presents the results 
for the cause-specific hazard rate for cause 01 and 02, respectively.  The
empirical estimates and 95\% confidence intervals are indistinguishable from
the true quantities using Proposition~\ref{thm:asymH} and estimated quantities 
using Proposition~\ref{thm:asymH} but replacing all quantities with their
respective estimates from Section \ref{subsec:stat_results}.  This agreement
further confirms Proposition~\ref{thm:asymH}.

\begin{table}[t!]
\caption{
\footnotesize{
\textbf{Simulation Study Lifetime of Interest Probabilities.}\newline
The true probabilities of the lifetime random variable, $X$, for the
simulation study results of Figure~\ref{fig:sim_study}.  The probabilities
$p(x)$ and $\Pr(X=x)$ for $x \in \{1, \ldots, 10\}$ are selected at onset,
and the remaining probabilities in this table may be derived from these
quantities.  Not summarized here is the truncation random variable, $Y$,
which was assumed to be discrete uniform over the integers
$\{1, \ldots, 5\}$.
}
}
{
\begin{adjustbox}{max width=\textwidth}
\begin{tabular}{cccccccc}
$p(x)$ & $X$ & $\Pr(X=x)$ & $\Pr(X=x,Z_x=1)$ & $\Pr(X=x,Z_x=2)$ & $\lambda^{01}(x)$
         & $\lambda^{02}(x)$ & $\lambda(x)$\\
         \toprule
            0.66&1&0.04&0.026&0.014&0.026&0.014&0.04\\
            0.20&2&0.06&0.012&0.048&0.013&0.050&0.06\\
            0.45&3&0.10&0.045&0.055&0.050&0.061&0.11\\
            0.87&4&0.14&0.122&0.018&0.152&0.023&0.18\\
            0.20&5&0.09&0.018&0.072&0.027&0.109&0.14\\
            0.81&6&0.06&0.049&0.011&0.085&0.020&0.11\\
            0.05&7&0.14&0.007&0.133&0.014&0.261&0.27\\
            0.78&8&0.18&0.140&0.040&0.379&0.107&0.49\\
            0.25&9&0.07&0.018&0.053&0.092&0.276&0.37\\
            0.42&10&0.12&0.050&0.070&0.420&0.580&1.00\\
        \bottomrule
\end{tabular}
\end{adjustbox}
}
\phantomsection    \label{tab:X_probs}
\end{table}

\begin{figure}[t!]
    \centering
    \includegraphics[width=\textwidth]{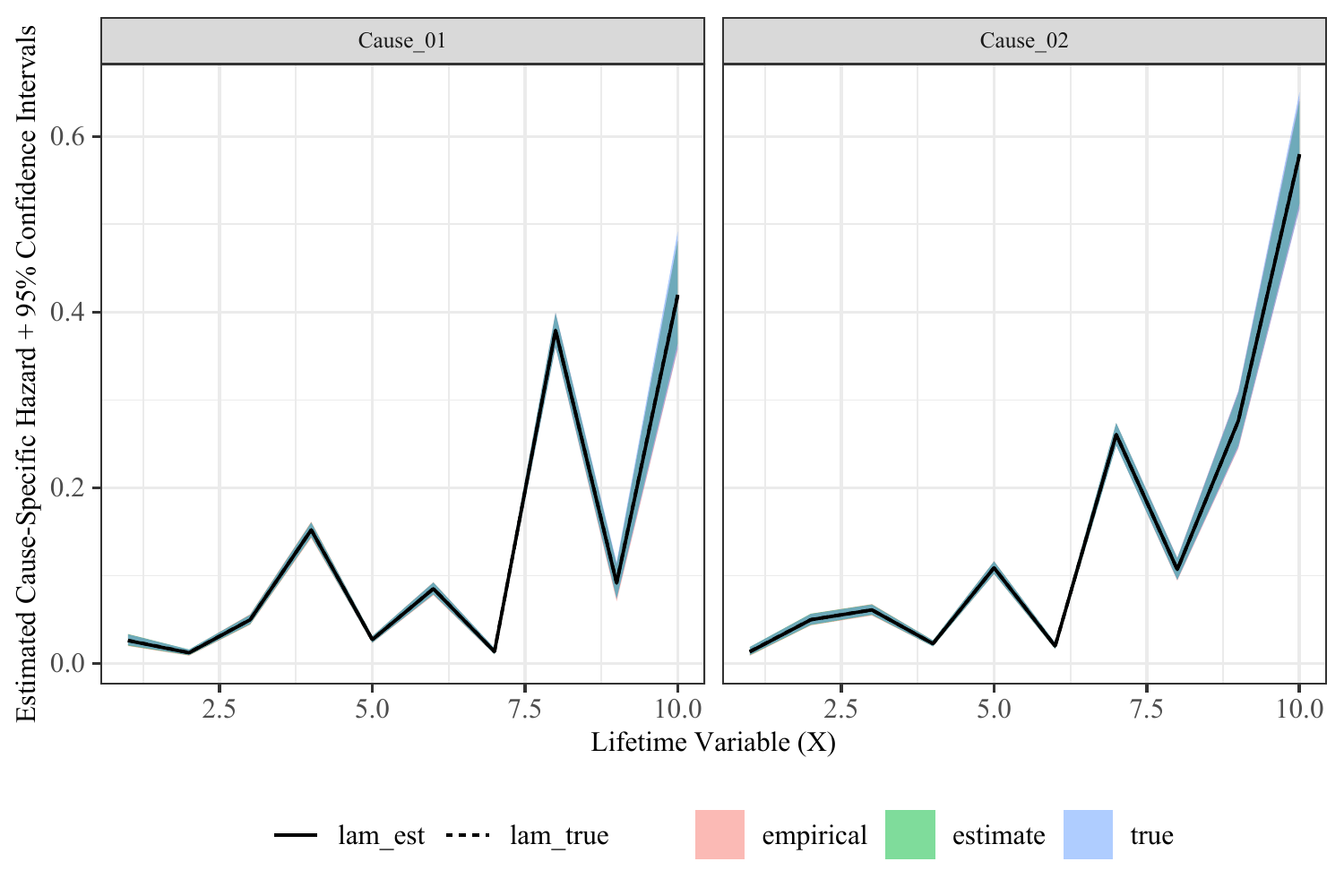}
        \caption{
\footnotesize{
    \textbf{Simulation Study Results.}\newline
	A comparison of true $\lambda^{0i}_{\tau}(x)$ (\texttt{lam\_true}) and
    estimated $\hat{\lambda}^{0i}_{\tau,n}(x)$ (\texttt{lam\_est}),
    including confidence intervals,
    for the distribution in Table \ref{tab:X_probs} and $i = 1, 2$. The ``true"
    values are from Proposition~\ref{thm:asymH} and Lemma \ref{cor:haz_ci}.
    The ``estimate" values use the formulas from Proposition~\ref{thm:asymH} and
    Lemma \ref{cor:haz_ci} but replace the true values with the estimates
    from Section \ref{subsec:stat_results} calculated from the simulated data.
    The ``empirical" values are empirical confidence interval and mean 
    calculations directly from the simulated data.  All three quantities are
    indistinguishable for $n = 10{,}000$ and 1,000 replicates, which
    indicates the asymptotic properties hold in this instance.
}
}
\phantomsection \label{fig:sim_study}
\end{figure}

\section{Determination of Loan Outcome}
\label{subsec:loan_algo}

The detail of the loan-level data is extensive, but it remains up to the data
analyst to use the provided fields to determine the outcome of an individual 
loan (see \citet{cfr_229} for detail on available field names). To do so, we
aggregate each month of active trust data into a single source file. This
allows us to review each bond's monthly outstanding principal balance, monthly 
payment received from the borrower, and the portion of each monthly payment 
applied to principal.

Our algorithm to determine a loan outcome proceeds as
follows.  For each remaining bond after the filtering of Section 
\ref{subsec:loan_filter}, we extract three vectors, each of which was the same
length as the number of months a trust was active and paying. The first vector
represents the ordered monthly balance, the second is the ordered monthly
payments, and the third is the ordered monthly amount of payment applied to 
principal. We then consider a loan to be repaid if the sum total principal 
received was greater than the outstanding loan balance as of the first month 
the trust was actively paying. In this case, the timing of a repayment is set 
to be the first month with a zero outstanding principal balance. Note that we 
do not differentiate between a prepayment or naturally scheduled loan 
amortization; i.e., all repayments have been treated as a ``non-default". If 
the sum total principal received is less than the first month's outstanding
loan balance, we then consider a loan outcome to be either right-censored or
defaulted. To make this determination, we search the monthly payments 
received vector for three consecutive zeros (i.e., three straight months of
missed payments).  If we find three consecutive missed payments, we assume 
the loan to be defaulted with a time-of-default set to be the month in which
the first of three zeros is observed.
If we do not find three consecutive months of 
missed payments, the loan is assumed to be a right-censored observation
and assigned
an event time as of the last month the trust was actively paying.  For the
pseudo-code of this algorithm, see Figure~\ref{fig:algo}.  Please contact the
corresponding author for further details.

\begin{figure}[t!]
\begin{algorithmic}[1]
\State $B \gets \texttt{bond\_data}$ 
\Comment{bond\_data is a row of the loan performance data}
\State $\texttt{bal\_vec} \gets \text{each month's sequential outstanding principal balance}$
\State $\texttt{pmt\_vec} \gets \text{each month's sequential actual payment}$
\State $\texttt{prc\_vec} \gets \text{each month's sequential payment applied to principal}$
\State $\texttt{init\_bal} \gets \text{current balance as of the first trust month}$
\State $\texttt{paid\_princ} \gets \texttt{sum}( \texttt{prc\_vec} )$
\Comment{plus \$10 pad to avoid odd tie behavior}
\If{$\texttt{paid\_princ} >= \texttt{init\_bal}$}
      \State $D = 0$
      \State $R = 1$
      \State $C = 0$
      \State $X \gets$ location of first zero in $\texttt{bal\_vec}$
      \Comment{loan repaid}
    \Else
      \State $z \gets$ starting time of three consecutive zero payments in $\texttt{pmt\_vec}$
        \If{$z$ empty}
                \State $D = 0$
                \State $R = 0$
                \State $C = 1$
                \State $X \gets$ length of $\texttt{pmt\_vec}$
                \Comment{loan censored}
            \Else
                \State $D = 1$
                \State $R = 0$
                \State $C = 0$
                \State $X \gets z$
                \Comment{loan defaults}
        \EndIf
    \EndIf
\end{algorithmic}
\caption{
\footnotesize{
    \textbf{Determination of Loan Outcome.}\newline
	We first extract three vectors, each of which is the same
	length as the number of months the trust was active and paying.
	The first vector (\texttt{bal\_vec}) represents the ordered monthly
	balance, the second (\texttt{pmt\_vec}) is the ordered monthly
	payments, and the third (\texttt{prc\_vec}) is the ordered monthly
	amount of payment applied to principal. We consider a loan to
	be repaid if the sum total principal received is greater than the
	outstanding loan balance as of the first month the trust was actively
	paying. In this case, the timing of a repayment is set to be the first
	month with a zero outstanding principal balance. If the sum total
	principal received is less than the first month's outstanding loan
	balance, we consider a loan outcome to be either right-censored
	or defaulted. To make this determination, we search the monthly payments 
	received vector for three consecutive zeros (i.e., three straight months
	of missed payments).  If we find three consecutive missed payments, we
	assume the loan to be defaulted with a time-of-default set to be the
	month in which the first of three zeros is observed.  If we do not find
	three consecutive months of missed payments, the loan is assumed to be a
	right-censored observation and assigned an event time as of the last month
	the trust was actively paying.
}
}
\phantomsection \label{fig:algo}
\end{figure}

\section{Lifetime Risk-Adjusted Return}
\label{subsec:cer}

We present an expansion of the actuarial methods in
Section~\ref{subsec:lend_prof} to consider the full remaining lifetime of a
loan rather than assuming a prepayment in the next month.
Denote the risk-adjusted rate of return for a loan in risk band $a$ as
$\rho_a$.  Given reliable estimates of borrower default and prepayment
probabilities, such as those in Section~\ref{subsec:stat_results},
we may estimate $\rho_a$ for a given loan in risk band $a$.  In particular, we 
may estimate $\rho_a$ for each month a loan is still active and paying to find
a \textit{conditional risk-adjusted rate of return} over a loan's full
remaining lifetime.\footnote{
Contrast this with Section~\ref{subsec:lend_prof}, in which we calculate a
one-month risk-adjusted return.  
}
Pleasingly, $\rho_a$ equals the loan contract effective rate of return in the
event the future loan payments will proceed as scheduled with no uncertainty,
which we state formally in Proposition \ref{prop:rho}.

\begin{proposition}[\textit{Risk-Adjusted Rate of Return, No Payment
Uncertainty}] \label{prop:rho}
Suppose a loan is originated with an initial balance, $B$, a monthly rate of
interest, $r_a$, and a term of $\psi$ months.  Let
$\rho_{a \mid x}$ denote the risk-adjusted rate of return given the loan has 
survived to month $x$.  If the probability that all payments will follow the 
amortization schedule exactly is unity (i.e., no payment uncertainty), then
$\rho_{a \mid x} = r_a$ for all $x \in \{1, \ldots, \psi\}$.
\end{proposition}
\begin{proof}
See the  Online Appendix~\ref{subsec:proof3}.
\end{proof}

We now formalize the estimation of $\rho_{a \mid x}$, as defined in
Proposition~\ref{prop:rho}.  For convenience of
notation, we will drop $a$ to denote the arbitrary risk band and assume the
proceeding calculations will be performed entirely within one risk band.
Assume we consider a loan with a $\psi$-month schedule.  Denote the current
age of a loan by $x$, $1 \leq x \leq \psi$.\footnote{
Depending on the impact of left-truncation and right-censoring, the recoverable
range of $X$ may not be the entire original loan termination schedule (see
Section~\ref{subsec:stat_results} for details).  In such an instance,
assumptions about the probability distribution may be necessary.  Assuming a
geometric right-tail (i.e., a constant hazard rate that follows the last
recoverable value) is common in survival
analysis \citep[Section 12.1]{klugman_2012}.  We will proceed as though the full
distribution is recoverable and allow readers to adjust as needed.
}
Let the cause-specific hazard rate for default
at time $x$ be denoted by $\lambda^{01}(x)$ and the cause-specific hazard rate
for repayment at time $x$ be denoted by $\lambda^{02}(x)$. Assuming no other 
causes for a loan termination, the all-cause hazard rate is then 
$\lambda(x) = \lambda^{01}(x) + \lambda^{02}(x)$. Further, recall 
\eqref{eq:haz_surv} and observe for $i = 1, 2$, $x \leq j \leq \psi$,
\begin{align*}
\Pr(X = j, Z_x = i) &= \frac{ \Pr(X = j, Z_x = i) }{ \Pr(X \geq x) }
\Pr(X \geq x)\\
&= \Pr(X = j, Z_x = i \mid X \geq x) \Pr(X \geq x)\\
&= \lambda^{0i}(j) \prod_{k=x}^{j-1} \{1 - \lambda(k)\}, 
\end{align*}
again with the convention $\prod_{k=x}^{x-1} \{1 - \lambda(k)\} = 1$. For
convenience, denote $p^{0i}_x(j) = \Pr(X = j, Z_j = i \mid X \geq x)$ for
$i = 1, 2$, $x \leq j \leq \psi$.  Hence,
\begin{equation*}
p^{0i}_x(j) =
\begin{cases}
\lambda^{0i}(x), & j = x\\
\lambda^{0i}(j) \prod_{k=x}^{j-1} \{1 - \lambda(k) \}, & j > x,
\end{cases}
\quad \quad
i = 1, 2.
\end{equation*}
One may verify $\sum_{j=x}^{\psi} \sum_{i=1}^{2} p^{0i}_x(j) = 1$ for every
$x$.\footnote{
It may be of help to review the numeric example of Table~\ref{tab:X_probs} in
 Online Appendix~\ref{sec:sim_study}.
}

We estimate $\rho_x$ as follows.  Let the scheduled amortization loan balance
of a consumer auto loan at month $x$, $1 \leq x \leq \psi$ be denoted 
by $B_x$, where $B_{\psi} = 0$.  Denote the scheduled monthly payment by $P$.  
If we denote the recovery of a defaulted consumer auto loan at month $x$, 
$1 \leq x \leq \psi$, by $R_x$, then the default matrix at loan age 
$x \leq \psi-1$ for the possible future default paths is
\begin{equation*}
    \mathbf{DEF}_{(\psi - x + 1) \times (\psi - x + 1)} =
    \begin{bmatrix}
    R_x & 0 & 0 & \ldots & 0 & 0\\
    P & R_{x+1} & 0 & \ldots & 0 & 0\\
    P & P & R_{x+2} & \ldots & 0 & 0\\
    \vdots & \vdots & \vdots & \ddots & \vdots & \vdots\\
    P & P & P & \ldots & R_{\psi-1} & 0\\
    P & P & P & \ldots & P & R_{\psi}
    \end{bmatrix}.
\end{equation*}
Note that row $1$ of $\mathbf{DEF}$ would be the cash flows assuming a default
at loan age $x$, which occurs with probability $p_{x}^{01}(x)$. Similarly, row
$2$ of $\mathbf{DEF}$ would be the cash flows assuming a default at loan age
$x + 1$, which occurs with estimated probability $p_{x}^{01}(x+1)$, and so on
and so forth.  In the same way, we can define the prepayment matrix at loan age
$x \leq \psi - 1$ as
\begin{equation*}
    \mathbf{PRE}_{(\psi - x + 1) \times (\psi - x + 1)} =
    \begin{bmatrix}
    B_x + P & 0 & 0 & \ldots & 0 & 0\\
    P & B_{x+1} + P & 0 & \ldots & 0 & 0\\
    P & P & B_{x+2} + P & \ldots & 0 & 0\\
    \vdots & \vdots & \vdots & \ddots & \vdots & \vdots\\
    P & P & P & \ldots & B_{\psi-1} + P & 0\\
    P & P & P & \ldots & P & P
    \end{bmatrix}.
\end{equation*}
As with defaults, row $1$ of $\mathbf{PRE}$ would be the cash flows assuming a
prepayment at loan age $x$, which occurs with estimated probability 
$p_{x}^{02}(x)$. Similarly, row $2$ of $\mathbf{PRE}$ would be the cash flows 
assuming a prepayment at loan age $x + 1$, which occurs with estimated 
probability $p_{x}^{02}(x+1)$, and so on and so forth.  Therefore, if we denote 
the $(\psi-x+1) \times 1$ dimensional discount vector assuming the unknown
monthly rate of $\rho_x$ as
\begin{equation*}
    \big( \bm{\nu}_x \big)^{\top} = 
    \begin{pmatrix} 
    (1 + \rho_x)^{-1} & (1 + \rho_x)^{-2} & \ldots & (1 + \rho_x)^{-(\psi-x+1)} 
    \end{pmatrix}^{\top},
\end{equation*}
and the $(\psi-x+1) \times 1$ dimensional cause-specific probability vector as
\begin{equation*}
    \big( \bm{p}^{0i}_x \big)^{\top} = 
    \begin{pmatrix} 
    p^{0i}_x(x) & p^{0i}_{x}(x+1) & \ldots & p^{0i}_{x}(\psi) 
    \end{pmatrix}^\top,
\end{equation*}
then the expected present value (EPV) of a loan at age $x \leq \psi - 1$ is
\begin{equation*}
    \text{EPV}_x=
    \big( \bm{p}^{01}_x \big)^\top
	\textbf{DEF}_x \bm{\nu}_x +
	\big( \bm{p}^{02}_x \big)^\top
	\textbf{PRE}_x \bm{\nu}_x.
\end{equation*}
Therefore, $\rho_x$ is the interest rate such that $B_x = \text{EPV}_x$; 
that is,
\begin{equation}
    \{\rho_x : B_x = \text{EPV}_x \}.
    \label{eq:rho_x}
\end{equation}
In words, $\rho_x$ represents the expected return realized by lending $B_x$ 
and taking into account the original monthly payments $P$ and default and
prepayment risk over the remaining lifetime of the loan.  We have 
$\rho_x \leq r$ for a given contract, with equality only in the circumstances
of Proposition \ref{prop:rho}.
Finally, we of course do note know the true distribution of $X$.  We do have
the estimators in \eqref{eq:csh_est}, however, and Proposition~\ref{thm:asymH}. 
Thus, we may estimate $\rho_x$ by replacing the cause-specific hazard rates
$\lambda^{0i}$ with the estimate in \eqref{eq:csh_est}.  For completeness, we
close this section with the following lemma.

\begin{Lemma}[$\hat{\rho}_{n,x}$ \textit{Asymptotic Properties}]
\label{cor:rho_hat}
Replace the cause-specific hazard rates in \eqref{eq:rho_x} with the estimators
from \eqref{eq:csh_est}.  Define the estimated risk-adjusted rate of return
over the remaining lifetime given a loan has survived to month $x$ as
$\hat{\rho}_{n, x}$.  Then,
\begin{equation*}
\hat{\rho}_{n, x}
\overset{\mathcal{P}}{\longrightarrow} \rho_x,
\text{ as } n \rightarrow \infty.
\end{equation*}
\end{Lemma}
\begin{proof}
See the  Online Appendix~\ref{subsec:proof3}.
\end{proof}

\section{Proofs: Appendix~\ref{subsec:cer}}
\label{subsec:proof3}

\begin{proof}[Proof of Proposition~\ref{prop:rho}]
For a loan with initial balance, $B$, monthly interest rate, $r_a$, and initial
term of $\xi$, the monthly payment, $P$, is
\begin{equation*}
P = B \bigg[ \frac{1 - (1 + r_a)^{-\xi}}{r_a} \bigg]^{-1}.
\end{equation*}
Assume $x \in \{1, \ldots, \xi\}$.  The balance at month $x$, $B_x$ is
\begin{align}
B_x &= B(1 + r_a)^x - P \bigg[ \frac{(1+r_a)^x - 1}{r_a} \bigg] \nonumber\\
&= B(1 + r_a)^x 
- B \bigg[ \frac{1 - (1 + r_a)^{-\xi}}{r_a} \bigg]^{-1}
\bigg[ \frac{(1+r_a)^x - 1}{r_a} \bigg].
\label{eq:B_x}
\end{align}
Thus, $\rho_{a \mid x}$ is the rate such that the expected present value of the
future monthly payments equals $B_x$.  The payment stream is constant, however,
and so
\begin{align*}
B_x &= P \bigg[ \frac{1}{(1+\rho_{a \mid x})} + \cdots + 
\frac{1}{(1+\rho_{a \mid x})^{\xi-x}} \bigg]\\
&= B \bigg[ \frac{1 - (1 + r_a)^{-\xi}}{r_a} \bigg]^{-1}
\bigg[\frac{1 - (1 + \rho_{a \mid x})^{-(\xi-x)}}{\rho_{a \mid x}} \bigg].
\end{align*}
Use \eqref{eq:B_x} and solve for $\rho_{a \mid x}$ to complete the proof.
\end{proof}

\begin{proof}[Proof of Lemma~\ref{cor:rho_hat}]
The result follows by Proposition~\ref{thm:asymH}, part $(i)$ and the Continuous
Mapping Theorem \citep[Theorem 5.2.5, pg. 249]{nitis_2000}.
\end{proof}

\end{document}